\documentclass[acmsmall,screen]{acmart}

\usepackage{tikz}
\usepackage{amsmath}
\usepackage{amsfonts}
\usepackage{color}
\usepackage{proof}
\usepackage{mathrsfs}  
\usepackage{bm}
\usepackage{stmaryrd}
\usepackage{csquotes}
\usepackage[inline]{enumitem}
\input{macros}

\title{On Feller Continuity and Full Abstraction} 

\author{Gilles Barthe}
\orcid{0000-0002-3853-1777} 
\affiliation{\institution{MPI Security and Privacy}\country{Germany}}
\email{gilles.barthe@mpi-sp.org}

\author{Rapha\"elle Crubill\'e}
\orcid{0000-0002-4197-0439}
\affiliation{\institution{CNRS}\country{France}}
\email{raphaelle.crubille@lis-lab.fr}

\author{Ugo Dal Lago}
\orcid{0000-0001-9200-070X}
\affiliation{\institution{University of Bologna}\country{Italy}}
\affiliation{\institution{INRIA}\country{France}}
\email{ugo.dallago@unibo.it}

\author{Francesco Gavazzo}
\orcid{0000-0002-2159-0615}
\affiliation{\institution{University of Bologna}\country{Italy}}
\affiliation{\institution{INRIA}\country{France}}
\email{francesco.gavazzo2@unibo.it}

\longv{\setcopyright{none}}

\shortv{
\setcopyright{rightsretained}
\acmJournal{PACMPL}
\acmYear{2022}
\acmVolume{6}
\acmNumber{ICFP}
\acmArticle{120}
\acmMonth{8}
\acmPrice{}
\acmDOI{10.1145/3547651}
\copyrightyear{2022}
\acmSubmissionID{icfp22main-p91-p}}



\begin{abstract}
  We study the nature of applicative bisimilarity in $\lambda$-calculi 
  endowed with operators for sampling from  \emph{continuous} distributions. 
  On the one hand, we show that bisimilarity, logical equivalence, and 
  testing equivalence all coincide with contextual equivalence when real 
  numbers can be manipulated only through continuous functions. The key    
  ingredient towards this result is a novel notion of Feller-continuity for 
  labelled Markov processes, which we believe of independent interest, being 
  a broad class of LMPs for which coinductive and logically inspired 
  equivalences coincide. On the other hand, we show that if no constraint is 
  put on the way real numbers are manipulated, characterizing contextual 
  equivalence turns out to be hard, and most of the aforementioned notions of 
  equivalence are even \emph{unsound}. 
\end{abstract}

\longv{\hypersetup{draft}}
\begin{document}
\theoremstyle{acmplain}
\newtheorem{remark}[theorem]{Remark}
\newtheorem{notation}[theorem]{Notation}

\maketitle   

\section{Introduction}
There exist many different definitions of program equivalence for
higher-order programs. These definitions come in different flavours,
including contextual equivalence~\cite{Morris/PhDThesis}, coalgebraic and 
bisimulation-based definitions~\cite{Abramsky1990,Sangiorgi/TheLazyLambdaCalculusInAConcurrencyScenarion/1994,Lassen/BismulationUntypedLambdaCalculusBohmTrees}, 
logical definitions based on modal logics~\cite{Abramsky1990,Ong/PhDThesis/1988},
and operational definitions based on tests~\cite{Ong/LICS/1993}. 
In the baseline setting,
where program evaluation does not produce computational effects, 
all these definitions coincide,
and provide a rich set of methods for proving program equivalence.

In contrast, these different notions are not known to coincide in 
 effectful programming languages; and even for some specific effect, such as sampling
from distributions, it can be difficult to relate, or even define,
the underlying relations. On the positive side, it is known that the
aforementioned notions of equivalence coincide for \emph{discrete}
probabilistic effects, i.e.\, when values are sampled from discrete
distributions~\cite{CDL14,DBLP:conf/birthday/CrubilleLSV15}.  However,
little is known about the correspondence between different notions of
equivalence for higher-order languages with sampling from
\emph{continuous} distributions, which are crucial in many application
domains such as Bayesian programming. One main challenge with these
languages is that their underlying operational semantics goes beyond
well-understood classes of probabilistic processes.  When sampling
from continuous distributions is available, programs no longer form a
labelled Markov chain, but form instead a \emph{labelled Markov
  process}~\cite{Prakash2009} (LMP in the following). Unfortunately, even if 
  these LMPs have an
\emph{analytic} state space, they come with an \emph{uncountable} set of
labels. As a consequence, they fall outside of the class of LMPs with
analytic state spaces \emph{and} countable sets of labels, for which
correspondence between the different definitions of equivalence still
holds~\cite{Prakash2009}.

The question addressed in this paper is therefore: what is (\emph{if any}) a natural
condition on LMPs (with \emph{analytic} state space and an
\emph{uncountable} set of labels), such that the different notions of
equivalence coincide? The idea explored in this paper is to ask that
programs are only built from \emph{continuous} functions. And indeed,
asking basic functions to be continuous has the consequence of making
contextual equivalence, both forms of bisimilarity from the literature
(i.e. state bisimilarity~\cite{DBLP:journals/iandc/DesharnaisEP02} 
and event bisimilarity~\cite{danos2006bisimulation,danos2005almost}),
logical equivalence and testing equivalence to coincide. The path to
this result is not trivial, and goes through the definition of a new
class of LMPs, called \emph{Feller-continuous} LMPs, whose transition functions satisfy a continuity constraint.

Now equipped with Feller continuity as a sufficient condition for
proving correspondence between the different notions, we ask for the
\emph{necessity} of Feller continuity. We do not provide a definitive answer
to this question. However, our partial exploration provides several
results of independent interest. On the positive side, we prove that
several key technical results used to establish equivalence between
the different notions extend to the general setting. On the negative
side, we show that \emph{event} bisimilarity, as well as 
\emph{testing} and \emph{logical} equivalence, are not \emph{sound} in
the general case. We also point to some undesirable behavior of
\emph{state} bisimilarity. These negative results do not rule out the
existence of alternative better behaved definitions. However,
we also show that contextual equivalence is not measurable. This, in
itself, points to the broader challenge of proving program equivalence
in absence of Feller continuity.

At a higher level, our results show a conceptual gap between discrete
and continuous distributions.  When working with calculi with
\emph{discrete} probabilistic choice, operational reasoning is mostly
set-theoretic: programs form an ordinary set, and program execution is
modeled as a set-theoretic function; accordingly, operational
reasoning is mostly performed relying on the classic notion of a
(program) relation.  When we enter the realm of continuous
probability, operational reasoning drastically changes, moving from a
set-theoretic to a \emph{measure-theoretic} base: programs now form
not just a set, but a \emph{measurable} space, and their evaluation
becomes a \emph{measurable function}. What about relations? One would
expect that relations should become, morally, measurable too: and
indeed, mathematically speaking, things work smoothly if one restricts
to measurable relations, with beautiful, classic theorems --- such as
the well-known Monge-Kantorovich duality \cite{villani2008optimal} --- 
extending to measurable spaces, functions, and relations. 
That suggests that to have well-behaved 
program equivalences, one should require such equivalences to be measurable. 
However, such a requirement is simply too strong: in absence of continuity 
constraints, contextual equivalence is \emph{not} measurable. Consequently, 
even if in principle we could define \emph{measurable} notions of program 
equivalence, the latter are
doomed not to capture the full power of contextual equivalence: the
very expressive power of the calculus simply goes beyond the world of
measurable relations. Looking back at the notion of Feller-continuity
from this perspective, we can see the latter as a somehow minimal
condition ensuring contextual equivalence (and, \emph{a fortiori},
all the aforementioned notions of equivalence) to be measurable, this
way establishing a new, deep connection between the syntax of a
calculus and the semantics of its associated operational reasoning.

\subsection{Organization of the Paper and Summary of Contributions}
The paper is organized as follows:
\begin{varitemize}
\item In Section~\ref{sect:lambda}, we introduce $\Lambda_{\probb}$, a simply-typed probabilistic
  $\lambda$-calculus with recursive types and real numbers, which we
  use as a vehicle language for our study. We show that for some
  choice of primitive functions, contextual equivalence for $\Lambda_{\probb}$
  is not measurable. We also introduce technical notions, such as that of a
  preterm, which are central to our development. Finally, we introduce
  a fragment of $\Lambda_{\probb}$, called $\lambdacont$, 
  in which all function symbols
  are continuous.

\item 
  We then present, in Section~\ref{sect:equivalences_on_lmps}, a unified 
  overview of different notions of equivalence and of their correspondence, 
  in the case of LMPs with analytic state spaces and countable set of 
  labels. We also single out correspondences that extend to the case of 
  arbitrary LMPs. We also define the LMP describing the interactive 
  behaviour of $\Lambda_{\probb}$ programs, and observe that its 
  set of labels is uncountable; as a consequence, it falls out of the class of
  LMPs for which different notions of correspondence are related;

\item We introduce the notion of a Feller continuous LMP and prove that
  state and event bisimilarity coincide for Feller continuous LMPs 
  (Theorem~\ref{prop:coincide_event_state_generic}). It
  follows that the correspondence between different notions of
  equivalence for Feller continuous LMPs holds, similar to LMPs with
  analytical state space and discrete set of labels. All this is in
  Section~\ref{sect:feller}.

\item We prove that the LMP resulting from $\lambdacont$ is indeed
  Feller continuous. As a corollary, we get that full abstraction holds
  for $\lambdacont$, and that all the aforementioned notions of equivalence 
  coincide (Theorem~\ref{theorem:full_abstraction_lambdac}). 
  This is in Section~\ref{sect:completeness}.

\item We present in Section~\ref{sect:perspective} a preliminary analysis of 
the general case. Our main  result is that event bisimilarity, as defined in 
Section~\ref{sect:equivalences_on_lmps}, is unsound. Additionally, we 
illustrate by way of examples how measurability issues pose a challenge to the 
role of state bisimilarity in the context of $\Lambda_{\probb}$.
  \end{varitemize}

\shortv{
\noindent An extended version of this paper with more details is 
available~\cite{LV}.}

\section{The Stochastic $\lambda$-Calculus}\label{sect:lambda}
The target language of this work is a (call-by-value)
$\lambda$-calculus with finite sums and recursive 
types endowed with specific operators for sampling from continuous distributions, 
and for first-order functions on the real numbers. As such, the calculus 
is very expressive, both at the level of terms and at the level of types.
We separate values from terms \cite{Levy/InfComp/2003}, this way requiring 
terms to be explicitly sequenced via the let-in construct. 

\subsection{Syntax and Static Semantics}

The syntax and static semantics of $\Lambda_{\probb}$ are defined 
in Figure \ref{figure:typing-lambda-p}. 
\begin{figure*}[t]
%	\footnotesize
	\hrule
	\begin{align*}
	\typeone & \bnf \alpha 
	\mid \typereal
	\mid \sumtype{i \in I}{\typeone_i}
	\mid \typeone \to \typeone 
	\mid  \mu \alpha.\typeone
	\\
	\valone &\bnf  \varone 
	\mid \makereal{r} 
	\mid \inject{\imath}{\valone}
	\mid \abs{\varone}{\termone}
	\mid \fold{\valone}
	\\
	\termone &\bnf \valone 
	\mid \sample
	\mid \op(\valone, \hh, \valone)
                             \mid \makereal{\boolfun}(\valone, \hh, \valone) 
	\mid \casesum{\valone}{\termone}
	\mid \valone \valone 
	\mid \seq{\termone}{\termone}
	\mid \unfold{\valone}
\end{align*}

	\begin{center}
	$$
	\infer{\envone, \varone: \typeone \valimp \varone: 
	\typeone}{\phantom{\varone}}
	\quad
	\infer{\envone \valimp \makereal{r}: \typereal}{r \in \mathbb{R}}
	\quad 
	\infer{\envone \compimp \sample: \typereal}{\phantom{\varone} }
	$$
	$\vspace{-0.05cm}$
	\[
	%%%%%%%%%%%%
	\infer
	{\envone \compimp \makereal{\op}(\valone_1, \hh, \valone_n): \typereal}
	{\envone \valimp \valone_1 : \typereal, \hh,
	\envone \valimp \valone_n: \typereal
	& \op \in \mathcal{C}_n}
	\quad
	\infer
	{\envone \compimp \makereal{\boolfun}(\valone_1, \hh, \valone_n): \booltype}
	{\envone \valimp \valone_1 : \typereal, \hh, 
	\envone \valimp \valone_n: \typereal
		& \boolfun \in \Boolfun_n}
	\]
	$\vspace{-0.05cm}$
	\[
	\infer{\envone \compimp \valone: \typeone}{\envone \valimp \valone: 
	\typeone}
	\quad
	\infer{\envone \compimp \seq{\termone}{\termtwo}: \typetwo}
		{\envone \compimp \termone : \typeone 
			&
			\envone, \varone: \typeone \compimp \termtwo: \typetwo}
	%%%%%%%%%%%%%%%
	\quad
	%%%%%%%%%%%%%%%
	\infer{\envone \valimp \abs{\varone}{\termone}: \typeone\to \typetwo}
	{\envone, \varone:\typeone \compimp \termone: \typetwo}
	%%%%%%%%%%%%%%%
	\quad
	%%%%%%%%%%%%%%% 
	\infer{\envone \compimp \valone\valtwo: \typetwo}
	{\envone \valimp \valone : \typeone \to \typetwo
		&
		\envone \valimp \valtwo: \typeone}
    \]
	$\vspace{-0.05cm}$
	\[
	%%%%%%%%%%%%%%%%%%
	\quad
	%%%%%%%%%%%%%%%%%%
	\infer{\envone \valimp \inject{\hat{\imath}}{\valone}: \sumtype{i \in 
	I}{\typeone_i}}
	{\envone \valimp \valone: \typeone_{\hat{\imath}}}
	%%%%%%%%%%%%
	\quad
	%%%%%%%%%%%%%%
	\infer{\envone \compimp \casesum{\valone}{\termone_i}: \typeone}
	{\envone \valimp \valone: \sumtype{i \in I}{\typeone_i} 
		& \envone, \varone: \typeone_i \compimp \termone_i: \typeone\ 
		(\forall i \in I)}
	  \]
	$\vspace{-0.05cm}$
	\[
	\infer{\envone \valimp \fold{\valone}: \rectype{\typevar}{\typeone}}
	{\envone \valimp \valone: 
	\subst{\typeone}{\typevar}{\rectype{\typevar}{\typeone}}}
	%%%%%%%%%%%%%%%
	\quad
	%%%%%%%%%%%%%%%
	\infer{\envone \compimp \unfold{\valone}: 
	\subst{\typeone}{\typevar}{\rectype{\typevar}{\typeone}} 
	}
	{\envone \valimp \valone: \rectype{\typevar}{\typeone} }	
	\]
	$\vspace{-0.09cm}$
	\end{center}
	\hrule
	\caption{Syntax and Static semantics of $\Lambda_{\probb}$.}
	\label{figure:typing-lambda-p}
\end{figure*}
Types of $\Lambda_{\probb}$ are built 
starting from a countable set of type variables 
(denoted by the letter $\typevar$ in Figure \ref{figure:typing-lambda-p}) and 
the basic type $\typereal$ for real numbers, using finite sums, arrows, and 
recursion. In particular, in a type of the form 
$\sumtype{i \in I}{\typeone_i}$ we assume the letter $I$ to stand for a finite 
set 
whose elements are denoted by $\hat{\imath}, \hat{\jmath}, \hh$. 
We write $\voidtype$ for the empty sum type, $\unittype$ for  
$\voidtype \to \voidtype$, and $\booltype$ for $\unittype + \unittype$. 
Concerning the latter, 
we also use standard syntactic sugar, viz. $\trueval$, $\falseval$, and 
$\ite{\valone}{\termone}{\termtwo}$. As usual, we write 
$\valone \termone$ for $\seq{\termone}{\valone x}$ 
and $\termone\termtwo$ for $\seq{\termone}{x\termtwo}$.
% and $\nattype$ for 
%$\rectype{\typevar}{\unittype + \typevar}$.

We use \emph{term} judgments (resp. \emph{value} judgments)
of the form $\envone \compimp \termone: \typeone$ (resp.
$\envone \valimp \valone: \typeone$) to state that $\termone$
(resp. $\valone$) is a term (resp. value) of \emph{closed} type
$\typeone$ in the environment $\envone$, whereas an environment is a
finite sequence $\varone_1: \typeone_1, \hh, \varone_n: \typeone_n$ of
distinct variables with associated \emph{closed} types (we denote by
$\cdot$ the empty environment). Notice that we work with closed types
only.  We use letters $\termone, \termtwo, \hh$ and
$\valone, \valtwo, \hh$ to denote terms and values, respectively.
We write $\terms_{\typeone}$, $\values_{\typeone}$ for the set of
closed terms and values, respectively, of type $\typeone$.

%\textcolor{blue}{We write $\terms_{\typeone}$, $\values_{\typeone}$ for the 
%%%set of closed expressions and values respectively, of type $\typeone$.}
The syntax of $\Lambda_{\probb}$ consists of three distinct parts: the first 
one is a standard $\lambda$-calculus with recursive types and coproducts; the second one allows one to perform 
real-valued computations; finally, the third one is the actual probabilistic 
core of $\Lambda_{\probb}$. Concerning real-valued computations, 
we indicate real number in boldface with metavariables like $\makereal{r}$, 
while we take $\Lambda_{\probb}$ parametric with respect to 
a $\mathbb{N}$-indexed family $\mathcal{C}$ of \emph{countable} sets 
$\mathcal{C}_n$, each of which contains measurable functions from 
$\mathbb{R}^n$ to $\mathbb{R}$. In particular, we assume standard real-valued 
arithmetic operations to be in $\mathcal{C}$.
We also assume the existence of a comparison operator $op_{\leq} \in 
\mathcal{C}_2$ such that $op_{\leq}: \RR \times \RR \to [0,1]$, and 
$op_{\leq}(x,y) = 1$ if and only if $x \leq y$.
We use letters $f,g, \hh$ to denote elements in $\mathcal{C}_n$, 
writing their $\Lambda_{\probb}$ syntactic representation in boldface. 
%Notice that a function $\op \in \mathcal{C}_n$ has values as arguments, 
%rather than terms. Indeed, we can encode the expression 
%$\makereal{\op}(\termone_1, \hh, \termone_n)$ as
% $\mathbf{let}\ x_1 = \termone_1\ \mathbf{in}\ (\mathbf{let}\ x_2 = \termone_2\ 
% \mathbf{in}\ \hh (\mathbf{let}\ x_n = \termone_n\ \mathbf{in}\ 
% \op(x_1, \hh, x_n)))$.  
%$\mathbf{let}\ x_1 = \termone_1, \hh, x_n = \termone_n\ \mathbf{in}\ 
%\makereal{\op}(x_1, \hh, x_n)$. 
To relate numerical and general computations, we assume to have 
collections $\Boolfun_n$ of functions $b: \mathbb{R}^n \to \{0,1\}$ 
to perform Boolean tests on real numbers. For simplicity, we 
just take $\Boolfun = \{=,<\}$. Further functions can be encoded 
relying on programs on $\booltype$ representing Boolean operators. 

Probabilistic behaviours are triggered by the term $\sample$. 
The latter samples from the uniform distribution over the unit interval, 
which is nothing but the Lebesgue measure $\lebesgue$ on $[0,1]$. 
It is well-known~\cite{DBLP:journals/pacmpl/EhrhardPT18}
that starting from $\sample$ several probabilistic 
measures (e.g. binomial, geometric, and exponential distribution) can be 
defined through functions in $\mathcal{C}$.

\begin{example}\label{example:distributions}
	Given closed terms $\termone_{\mu}$, $\termone_{\sigma}$ of type 
	$\typereal$ 
	(encoding mean $\mu$ and standard deviation $\sigma$), we 
	represent the normal distribution with mean $\mu$ and standard deviation 
	$\sigma$ as the expression 
	$\mathbf{normal}\ \termone_{\mu}\ \termone_{\sigma}$, 
	where the term $\mathbf{normal}_{\mathbf{std}}$ 
	encodes the normal distribution with mean $0$ and 
	standard deviation $1$. 
	%Notice that the expressions $\termone_{\mu}$, 
	%$\termone_{\sigma}$ may be 
	%themselves defined in terms of other distributions:
	\begin{align*}
		\mathbf{normal}_{\mathbf{std}}
		&\defeq \seq{\sample}{(\seqy{\sample}{(\sqrt{-2 \log(x)} 
		\cos(2\pi y)})})
		\\
		\mathbf{normal}\ \termone_{\mu}\ \termone_{\sigma} 
		&\defeq
		\mathbf{let}\ x_{\mu} = \termone_{\mu},\ x_{\sigma} = 
		\termone_{\sigma}, y = 
		\mathbf{normal}_{\mathbf{std}} \ \mathbf{in}\
		((x_{\sigma} * y) + x_{\mu}).
		\end{align*}
	For any type $\typeone$, we have the Bernoulli choice term 
	$\textbf{bernoulli}: \typeone \to \typeone \to \typereal \to \typeone$ 
	acting as a probabilistic choice construct. Such a term is
	defined as $\abs{m,n,p}{\seq{\sample}{(\ite{x > p}{n}{m}})}$
	In particular, the term $\textbf{bernoulli}\ \termone\ \termtwo\ \makereal{0.5}$ 
	is the fair probabilistic choice between $\termone$ and $\termtwo$. 
	We employ the notation $\termone\ps\termtwo$ in place of 
	$\textbf{bernoulli}\ \termone\ \termtwo\ \makereal{0.5}$.
% 	\[
% \infer{\envone \compimp \termone\ps\termtwo:\typeone}
%     {\envone\compimp\termone:\typeone & \envone\compimp\termtwo:\typetwo}
%     \]
\end{example}

Finally, we adopt standard notational conventions \cite{Barendregt/Book/1984}. 
In particular, we denote by $FV(\termone)$ the collection 
of free variables of $\termone$ (resp. $\valone$) and we refer to closed expressions as \emph{programs}.
We denote by $\subst{\termone}{\varone}{\valone}$ the capture-avoiding 
substitution of 
the value $\valone$ for all free occurrences of $\varone$ in $\termone$.
We extend the aforementioned conventions to types. For instance, we denote 
by $\subst{\typeone}{\typevar}{\typetwo}$ the result of 
capture-avoiding substitution 
of the type $\typetwo$ for the type variable $\typevar$ in $\typeone$.

%\begin{align*}
%  \termone&\bnf\valone
%    \midd\termone\ps\termone
%    \midd\letin{\termone}{\varone}{\termtwo}
%    \midd\fsone(\termone_1,\ldots,\termone_n)
%    \midd \\
%  \valone&\bnf\valone
%    \midd\abs{\varone}{\termone}
%    \midd\rnone    
%\end{align*}

\subsection{Dynamic Semantics}

The dynamic semantics of $\Lambda_{\probb}$ is given by a type-indexed
family of evaluation functions $\sem{-}_{\typeone}$ mapping programs
of type $\typeone$ to \emph{sub-probability measures} over values of
type $\typeone$.  As such, the semantics is measure-theoretic and
defining it requires some preliminary definitions that we are going to
give. Due to lack of space, we assume that the reader is familiar with
basic definitions from measure theory and refer to classic textbooks
for additional background~\cite{Billingsley}.

\paragraph{Measurable Spaces}
We denote by $\Meas$ the category of measurable spaces and measurable 
functions. When $X$ is a measurable space, we write $\Sigma_X$ for the 
underlying $\sigma$-algebra. Moreover,
given a measurable space $(X, \Sigma_X)$ with a measure $\mu$ on it
and a function $f: X \to [0,\infty]$, we denote by 
$\int_X f \de \mu$ or $\int_X f(x) \mu(\de x)$ the number in 
$[0,\infty]$ obtained by Lebesgue integrating $f$ over $X$ with respect to $\mu$. 
Oftentimes, we will work with measurable spaces associated to a Polish space. 
More precisely, we shall work with standard Borel spaces, 
i.e. those measurable spaces generated by complete, separable metric spaces equipped with their 
Borel $\sigma$-algebra. We denote by $\Pol$ the full subcategory of $\Meas$ whose objects are standard Borel spaces 
(notice that arrows in $\Pol$ remain measurable functions).
Finally, we remark that 
both $\Meas$ and $\Pol$ have countable products and coproducts.

\paragraph{Probability Measures} 
We denote by $\Delta(X)$ the collection of 
sub-probability measures 
on a measurable space $(X, \Sigma_X)$, using the notation $\emptyset$ and $\dirac{x}$ 
for the empty-subdistribution (on $X$) and the Dirac distribution on $x$, respectively.
We endow $\Delta(X)$ with the $\sigma$-algebra generated by the set of evaluation maps 
$\{\mathit{ev}_A: \Delta(X) \to [0,1] \mid A \in \Sigma_X\}$, 
where $\mathit{ev}_A(\distrone) = \distrone(A)$. 
Recall that, given a measurable function $f: X \to Y$, the push forward of $f$ 
is the (measurable) function $\forward{f}: \Delta(X) \to \Delta(Y)$ defined by 
$\forward{f}(\distrone)(B) = \distrone(f^{-1}(B))$. 
These data induce a functor $\giry$ on $\Meas$, which is the functor 
part of the \emph{Giry monad} \cite{10.1007/BFb0092872}. 
In fact, $\giry$ has a unit map $\unit$
associating to each element $x$ the Dirac distribution $\dirac{x}$, 
and for a measurable function $f: X \to \giry(Y)$, the map
$\kleisli{f}$ defined by 
$\kleisli{f}(\distrone)(B) = \int f(x)(B) \distrone(dx)$ 
gives the Kleisli extension of $f$. 
When $(X, \Sigma_X)$ is a standard Borel space, then $\Delta(X)$ is is standard Borel space too 
--- so that $\giry$ is also a monad on $\Pol$ --- 
and the $\sigma$-algebra of $\Delta(X)$ coincides with the Borel 
algebra of
the weakest topology making the integration map $\distrone \mapsto \int f d\distrone$, 
continuous, for any bounded continuous real-valued function $f$ 
(see \cite{kechris}, Theorem 17.23 and 17.24).

\paragraph{Probability Kernels}
Given measurable spaces $X,Y$, we recall that a 
\emph{(sub)probability kernel}~\cite{DBLP:journals/entcs/Panangaden99} 
is a map $f: X \times \Sigma_Y$ such that:
	\begin{varitemize}
		\item For every $A \in \Sigma_Y$, the map $x \in X \mapsto f(x,A)$ is 
		measurable;
		\item For every $x \in X$, the map $A \in \Sigma_Y \mapsto f(x,A)$ is a 
		(sub)probability measure.
	\end{varitemize}
	We write $f:X \kernel Y$ when $f$ is a sub-probability kernel, and $f:X 
	\kernelp Y$ when $f$ is a probability kernel (also called Markov kernel).  
For any space $X$, the Dirac measure on $X$ induces a kernel $X \kernel X$.
Moreover, given $f:X \kernel Y$ and $g:Y \kernel Z$, we define their composition
$g \circ f: X \kernel Z$ by $(g \circ f)(x,C) = \int_Y g(y,C) \cdot f(x,dy)$. 
By the monotone convergence theorem, the composition operation $\circ$ 
is associative and has Dirac measures as unit, meaning that 
kernels form a category. Such a category is nothing but the Kleisli category 
of $\giry$.
Panangaden~\cite{DBLP:journals/entcs/Panangaden99} showed that the category of sub-probability kernels is partially additive~\cite{manesArbib86}. 
As a consequence, for all spaces $X$, $Y$, 
the set of sub-probability kernels from $X$ to $Y$ 
forms an $\omega$-cppo when ordered pointwise. In particular, 
the bottom element associates to any element $x$ the empty distribution $\emptyset$, 
whereas the least upper bound of an $\omega$-chain $\{f_n\}_{n \geq 0}$ of kernels 
is defined as $\sup_n f_n$. Additionally, kernel composition is monotone and 
continuous with respect to such $\omega$-cppo structure, meaning that 
the category of sub-probability kernels is $\ocppo$-enriched 
\cite{Kelly/EnrichedCats}, where $\ocppo$ is the category of 
$\omega$-cppos and continuous functions. All of this works both on 
$\Meas$ and $\Pol$ (with the notion of a probability kernel restricted to 
standard Borel  spaces).

\longv{
	Observe that  the space of probability kernels (sub-probability 
	kernels respectively) is equipped with a convex structure. Moreover, the 
	composition commutes with the convex operations, as stated in 
	Lemma~\ref{lemma:addition_kernels} below. We prove this here since we will 
	need to use it in later sections.
	\begin{lemma}\label{lemma:addition_kernels}
		Let $k: X \kernel Y$ and $g,h: Y \kernel Z$ be sub-probability kernels. 
		Then $k \circ(\alpha.g+\beta.h) = \alpha.k\circ g + \beta.k \circ h$, 
		and similarly $(\alpha.g+\beta.h)\circ k = \alpha.g\circ k + \beta.h 
		\circ k$.
	\end{lemma}
	\begin{proof}
		Let $x \in X$, $A \in \Sigma_Y$. Then $k \circ (\alpha.f+\beta.g)(x,A) 
		= \int_{Y} k(\cdot,A) d(\alpha.f+\beta.g)(x) = \int_{Y} k(\cdot,A) 
		d(\alpha.f(x)+\alpha.g(x)) = \alpha.\int_{Y} k(\cdot,A) 
		df(x)+\beta.\int_{Y} k(\cdot,A) dg(x)$.
	\end{proof}
}

\paragraph{Dynamic Semantics}
We now turn to the dynamic semantics of $\Lambda_{\probb}$.  Such a
semantics is given as a monadic evaluation semantics \cite{DLGL17} for
the Giry monad in the category $\Pol$. More concretely, the semantics
is given by type-indexed (sub-probability, due to the presence of full
recursion) kernels $\terms_{\typeone} \kernel \values_{\typeone}$.
Recall that for any type $\typeone$, the sets $\terms_{\typeone}$ and
$\values_{\typeone}$ carry a a standard Borel space
structure~\cite{DBLP:journals/pacmpl/EhrhardPT18,DBLP:journals/pacmpl/VakarKS19,DBLP:conf/lics/StatonYWHK16,borgstrom2016lambda}.
We recall this structure here, since we will use it explicitly in the
technical developments.
\begin{definition}\label{definition:terms_as_polish_spaces}
We call \emph{pre-terms} and \emph{pre-values} the syntactic objects
obtained from terms and values respectively, by removing all
occurrences of real numbers and replacing them by numbered holes
$\hole^n$. Moreover, we add the constraint that whenever $P$ is
a pre-term (or a pre-value) with $n$ holes, it means that each hole
$\hole^i$ for $1 \leq i \leq n$ must occur exactly once in $P$.
Whenever $P$ is a pre-term (or a pre-value) with $n$ holes, and $\vec
r \in \RR^n$, we write $P[\vec{\makereal{r}}]$ for the term obtained by
replacing each hole $\hole^i$ by the corresponding real value $\makereal{r}_i$. We
write $\terms_{P}$ (resp. $\valset_{V}$) for the set of all terms
(resp. values) of the form $P[\vec{\makereal{r}}]$ (resp. $V[\vec{\makereal{r}}]$), where
$\vec r$ ranges over $\RR^n$. 
\end{definition}
Observe that, for any pre-term $P$ with
$n$ holes, the set $\terms_{P}$ is in bijection with
$\RR^n$, so that we can naturally endow $\terms_{P}$ with a standard Borel space
structure. Since $\terms$ consists of the disjoint union of the sets 
$\terms_P$, it is in fact the countable coproduct of standard Borel spaces, 
thus a standard Borel space itself. A standard Borel space structure can be 
similarly given to $\valset$.

\begin{definition} 
	\label{definition:monadic-semantics}
	Define the type-indexed family of ($\mathbb{N}$-indexed) kernels 
	$\semn{-}{n}_{\typeone}: \terms_{\typeone} \kernel \values_{\typeone}$ 
	as follows (for readability, we omit type annotations):
	\begin{align*}
		\semn{\termone}{0}(A) 
		&= 0
		\\
		\semn{\valone}{n+1}(A) 
		&= \dirac{\valone}(A)
		\\
		\semn{\op(\makereal{r_1}, \hh, \makereal{r_n})}{n+1}(A)
		&=
		\semn{\makereal{\op(r_1, \hh, r_n)}}{n}(A)
		\\
		\semn{\makereal{\boolfun}(\makereal{r_1}, \hh, \makereal{r_n})}{n+1}(A)
		&=
		\semn{\makereal{b(r_1, \hh, r_n)}}{n}(A)
		\\
		\semn{\sample}{n+1}(A)
		&= \lebesgue \{r \mid \makereal{r} \in A\}
		\\
		\semn{(\abs{\varone}{\termone})\valone}{n+1}(A) 
		&= \semn{\subst{\termone}{\varone}{\valone}}{n}(A)
		\\
		\semn{\seq{\termone}{\termtwo}}{n+1}(A)
		&=
		 \int \semn{\subst{\termtwo}{\varone}{\valone}}{n}(A)
		\semn{\termone}{n}(\de \valone)
		\\
		\semn{\casesum{\inject{\hat{\imath}}{\valone}}{\termone_i}}{n+1}(A)
		&= \semn{\subst{\termone_{\hat{\imath}}}{\varone}{\valone}}{n}(A)
		\\
		\semn{\unfold{(\fold{\valone})}}{n+1}(A)
		&= \semn{\valone}{n}(A).
	\end{align*} 
\end{definition}

By observing that 
for any term $\Gamma, x: \typeone \compimp \termone: \typetwo$, 
the map 
$\subst{\termone}{\varone}{-}: 
\values_{\Gamma \imp \typeone} \to \terms_{\Gamma \imp \typetwo}$ is 
measurable \cite{DBLP:journals/pacmpl/EhrhardPT18,DBLP:conf/lics/StatonYWHK16,borgstrom2016lambda}, 
we immediately notice that each $\semn{-}{n}$ is indeed a kernel and
that $\{\semn{-}{n}_{\typeone}\}_{n \geq 0}$ forms
an $\omega$-chain, so that we 
can define $\sem{\termone}_{\typeone}$ as 
$\sup_n \semn{\termone}{n}_{\typeone}$. 
By $\ocppo$-enrichment, $\sem{\termone}$ is the least solution 
to the system of equations induced by Definition~\ref{definition:monadic-semantics}. 

\subsection{Contextual Equivalence}
Before moving to the study of coinductively-based notions of equivalence, 
we recall the notion of \emph{probabilistic contextual equivalence}.
%\begin{definition}[Context Equivalence]
%For every type $\typeone$, an \emph{observable context for $\typeone$} is a program $\ctxone$ with a hole $\hole$ such that $\hole:\typeone \imp \ctxone:\typereal$.
%\begin{align*}
%\forall \termone,\termtwo \in \terms_{\typeone},\,\termone \equiv^{ctx} \termtwo \text{ when }\forall \ctxone \in \contexts \typeone, \forall B \in \borels{[0,1]}, \sem {\ctxone \holefill \termone}(B) = \sem {\ctxone \holefill \termtwo}(B) \\
%\forall \valone,\valtwo \in \valset_{\typeone}, \, \valone \equiv^{ctx} \valtwo \text{ when }\forall \ctxone \in \contexts \typeone, \forall B \in \borels{[0,1]}, \sem {\ctxone \holefill \valone}(B) = \sem {\ctxone \holefill \valtwo}(B) 
%\end{align*}

%\end{definition}

%\begin{remark}
%\textcolor{blue}{TO DO:say that we  observe reals, instead of divergence. In presence of if, and caracteristic functions for each borelien, it is the same. However, at some point we will not have discontinuous function, and neither if.}
%\end{remark}

\begin{definition}[Context Equivalence]\label{def:context_equivalence}
For every type $\typeone$, an \emph{observable context for $\typeone$} is a 
program $\ctxone$ with a hole $\hole$ such that $\hole:\typeone \imp 
\ctxone:\typereal$. The set of all observable contexts for $\typeone$
is indicated as $\contexts\typeone$. Contextual equivalence of terms
and values is defined as usual, by quantifying over all possible observable
contexts:
\begin{align*}
\forall \termone,\termtwo \in \terms_{\typeone},\,\termone \ctxeq \termtwo \text{ when }\forall \ctxone \in \contexts \typeone,
%\forall B \in \borels{[0,1]}, \sem {\ctxone \holefill \termone}(B) = \sem {\ctxone \holefill \termtwo}(B) \\
\sem {\ctxone \holefill \termone} = \sem {\ctxone \holefill \termtwo} \\
\forall \valone,\valtwo \in \valset_{\typeone}, \, \valone \ctxeq \valtwo \text{ when }\forall \ctxone \in \contexts \typeone,
%\forall B \in \borels{[0,1]}, \sem {\ctxone \holefill \valone}(B) = \sem {\ctxone \holefill \valtwo}(B) 
 \sem {\ctxone \holefill \valone} = \sem {\ctxone \holefill \valtwo}
\end{align*}
\end{definition}
By Definition~\ref{def:context_equivalence}, the underlying notion of 
observation is the whole distribution produced in output by any observable 
context. This notion can, as usual, be proved robust to small 
changes in the language and in the notion of observation\shortv{, 
see~\cite{LV} for further discussion.}\longv{
. \subsubsection{Alternative Caracterisations of Context Equivalence.}
\begin{remark}\label{remark:altern_carac_ctx_eq}
Recall that, in generic terms, two programs are context equivalent if and only 
if their \emph{observables} are the same in any context. If we reframe 
Definition~\ref{def:context_equivalence} in this general framework, we see 
that it uses a very strong notion of \emph{observation} at real type: 
indeed Definition~\ref{def:context_equivalence} means that when $M$ is a 
real-type 
program, its precise operational semantics $\sem M : \RR \rightarrow 
[0,1] \in \distrs \RR$ can be observed. A weaker notion of observation would be 
to suppose that only the \emph{expected value} of $\sem M$ can be 
observed. This could seem to be a more reasonnable computational choice, 
since--from 
the law of large numbers--we know that the expected value of a 
distribution 
can be approximated from a large number of trials. Whether these 
two notions 
of observation lead to the same context equivalences depends on 
the primitive 
functions on reals that we put into our language. We are going to show in Proposition~\ref{prop:caract_equiv_ctx} below, that it is the case for $\Lambda_{\probb}$, due to the fact that $\mathcal C_2$ contains multiplication and the comparaison operator $op_{\leq}$, and more generically that it is the case for any  language that is: 
\begin{enumerate*}
\item able to express multiplication between reals, and 
\item\label{item:2_in_criteria_ctx} contains a countable 
collection $\mathcal D$ of functions $\RR \to \RR$ that \emph{separates 
points}, in the sense that 
for every reals $x \neq y$, there exists a function $f \in \mathcal D$  such 
that $f(x) \neq f(y)$.
\end{enumerate*}
Our main tool for showing this result is Lemma~\ref{lemma:expectation_as_observable_1} below, that comes from the literature on weak convergence for probability measure.
%In the following, we will consider 
%that this constraint holds for $\Lambda_{\probb}$.
%We can check that these conditions holds for $\Lambda_{\probb}$ as soon as 
%$\mathcal C_2$ contains multiplication: indeed 
%for~\ref{item:2_in_criteria_ctx} we can take $\mathcal D = \{ 
%op_{\leq}(\cdot, q): \RR \rightarrow \{0,1\} \mid q \in \QQ\}$ that we implement with the comparaison operator $op_{\leq} \in \mathcal{C_2}$.}
%%hthat there exists a countable set $\mathcal{C'}_1 $ of functions $\RR \to 
%%%%\RR$ that moreover are continuous, and such that as soon as $\mathcal{C'}_1 
%%%%%\subseteq \mathcal{C}_1$, the two notions of context equivalence coincide. 
%%%%%%This fact is useful in Section~\ref{sect:LMPapplicative}, when we define 
%%%%%%%applicative bisimilarity.
\end{remark}
\begin{lemma}\label{lemma:expectation_as_observable_1}
Let $\mathcal F$ be a countable set of bounded mesureable functions $\RR \to \RR$ such that:
\begin{enumerate}
\item \label{item:separating_1} $\forall x,y \in \RR, \exists f \in \mathcal{F},  \text{s.t. } f(x) \neq f(y)$; 
\item \label{item:separating_2} $\forall x \in \RR, \exists f \in \mathcal F, \text{s.t.} f(x) > 0$.
\end{enumerate}
%There exists a countable set $\mathcal{C'}_1 $ of continuous functions $\RR \to [0,1]$
Then for every
$\distrone, \distrtwo$ sub-probability measures on $\RR$:
\begin{equation}\label{eq:separating_space_1}
(\forall f \in \mathcal M(\mathcal F), \, \int f\cdot d\distrone = \int f \cdot d \distrtwo) \qquad \Leftrightarrow \qquad \distrone = \distrtwo,
\end{equation}
%Moreover, we can take such a set $\mathcal {C'}_1$ as any set $\mathcal M(\mathcal F)$,
where for any set of functions $A$, $\mathcal M(A) :=\{f_1 \cdot f_2 \ldots \cdot f_n  \mid n \in \NN \wedge\forall k \in \{1,\ldots,n\}, f_k \in A\}$.
%, and the following property holds for $\mathcal F$:
%stands for the set of all trigonometric polynomials with rationnal coefficents. 
\end{lemma}
\begin{proof}
%First, observe that it is enough to look for a set $\mathcal{C'_1}$ of \emph{bounded} continuous functions, since we can then renormalize these functions to have them at values in $[0,1]$.
In the literature, e.g~\cite{Ethier-Kurz}, a set of bounded measureable functions that verify~\eqref{eq:separating_space_1} is called \emph{separating}--a concept which is actually defined for general topological space $E$, and not only for $\RR$. We have to be careful here, because the literature consider \emph{proper} probability distributions instead of sub-distributions: it means that we need to rewrite our goal as:
for every
$\distrone, \distrtwo$ probability measures on $\RR \cup \{\bot\}$:
\begin{equation}\label{eq:separating_space}
(\forall f \in \mathcal M(\mathcal F), \, \int f_{\bot}\cdot d\distrone_{\bot} = \int f_{\bot} \cdot d \distrtwo_{\bot}) \qquad \Leftrightarrow \qquad \distrone_{\bot} = \distrtwo_{\bot},
\end{equation}
where $f_{\bot} : x \in \RR \cup \{\bot\} \mapsto f(x) \text{ for }x \in \RR, \, 0 \text{ otherwise}$, and $\distrone_{\bot}(A) = \distrone(A \cap \RR) + (1 - \mu(\RR))\cdot\delta_{\bot}(A)$.
%It is shown that--when $E$ is a polish space--in~\cite{Ethier-Kurz} that this condition is implied by a stronger one, called \emph{convergence determining}, that states that it is a sufficient family of functions to determine the weak convergence.
Theorem 11.b. in~\cite{blount2010convergence} show that a sufficient condition for a set of bounded measureable functions to be separating is to be countable, to be stable under multiplication, and to \emph{separate points}, i.e. that:
$\forall x,y \in \RR \cup \{\bot\}$, whenever $x \neq y$ there exist $ f \text{ s.t. } f(x) \neq f(y)$. It is immediate that the two first conditions hold for $\{f_{\bot} \mid f \in \mathcal M(\mathcal F)\}$. Moreover the hypothesis on $\mathcal F$ gives us exactly the third condition.
\end{proof}

\begin{example}\label{example:enough_primitives}\label{ex:enough_primitives}
We give here an example of a countable set $\mathcal F$ of bounded functions $\RR \to [0,1]$ such that the hypothesis~\eqref{item:separating_1} and~\eqref{item:separating_2} from Lemma~\ref{lemma:expectation_as_observable_1} hold (it can be shown by using the density of $\QQ$ in $\RR$ ):
%\begin{itemize}
%\item $\mathcal F_{\text{bool tests}} := \{ (? > q) \mid q \in \RR \} \cup \{x \in \RR \mapsto 1\}$; 
$$\mathcal F:=\{op_{\leq}(\cdot,q) \mid q \in \QQ,n \in \NN \}\cup \{x \in \RR \mapsto 1\}.$$
%, where
%$\mathcal F_{\text{cont}}:=\{approx_{q,n} \mid q \in \QQ,n \in \NN \}\cup \{x \in \RR \mapsto 1\}$, where
%$$approx_{q,n}(x) = \begin{cases} 0 \text { if } x \leq q - \fra%c 1 {2^n}\\
%1 \text{ if } x \geq q + \frac 1 {2^n}\\
%2^{n-1}\cdot (x - (q- \frac 1 {2^n})).
%\end{cases}$$
%Observe that all functions in $F_{\text{cont}}$ are continuous.
%\end{itemize}
%In this case, we use the fact that $\QQ$ is dense into $\RR$ to show that the hypothesis from Lemma~\ref{lemma:expectation_as_observable_1} hold.
%Let us prove that~\ref{eq:condition_on_conv_determining_set} holds: let $x \in \RR$, and $a,b$ such that $a < x < b$. Using the fact that $q$ is dense into $\RR$, we can take $q,n$ such that $a < q - \frac 1 {2^n}  < x  <  q + \frac 1 {2^n}< b$. From there, we take $g=approx_{q,n}\in \mathcal F$, and $\inf_{y \not \in ]a,b[}  \lvert{g(y)-g(x)}\rvert>\inf_{y \not \in ]q - \frac 1 {2^n},q + \frac 1 {2^n}[}  \lvert{g(y)-g(x)}\rvert = \inf \{g(x), 1 - g(x) \} > 0 $.
\end{example}

\begin{proposition}\label{prop:caract_equiv_ctx}
Suppose that the sets $\mathcal C_n, \mathcal B_n$ of primitive functions  allow to implement in $\Lambda_{\probb}$ a family $\mathcal F$ such that first the hypothesis~\eqref{item:separating_1} from Lemma~\ref{lemma:expectation_as_observable_1}
%--e.g the ones from Example~\ref{example:enough_primitives}--
holds, and second $\mathcal F$ contains the
%they also allow to implement the
multiplication function: $(\cdot):\RR \times \RR \to \RR$. Then it holds that $\forall \termone,\termtwo \in \terms_{\typeone}, \valone,\valtwo \in \valset_{\typeone}$:  
\begin{align*}
&\,\termone \equiv^{ctx} \termtwo \text{ iff }(\forall \ctxone \in \contexts \typeone, \int_{x \in \RR} x \cdot   d\sem {\ctxone \holefill \termone} = \int_{x \in \RR} x \cdot d \sem {\ctxone \holefill \termtwo}) \\
& \valone \equiv^{ctx} \valtwo \text{ iff }(\forall \ctxone \in \contexts \typeone, \int_{x \in \RR}  x \cdot d\sem {\ctxone \holefill \valone} = \int_{x \in \RR} x \cdot d\sem {\ctxone \holefill \valtwo})
\end{align*}
\end{proposition}
\begin{proof}
We do the proof for terms (it is exactly the same for values). 
First, observe that the left-to-right implication is immediate; in the following we show the right-to-left one. Suppose that $\forall \ctxone \in \contexts \typeone, \int_{x \in \RR} x \cdot   d\sem {\ctxone \holefill \termone} = \int_{x \in \RR} x \cdot d \sem {\ctxone \holefill \termtwo}$. Let $\ctxone \in \contexts \typeone$: we want to show that $\sem{\ctxone \holefill \termone} = \sem{\ctxone \holefill \termtwo}$.
Let $f \in \mathcal M(\mathcal F \cup \{x \in \RR\mapsto 1\})$: we write $f = f_1 \cdot \ldots\cdot f_n$, with $f_i \in \mathcal F\cup \{x \in \RR\mapsto 1\})$. We look at the context: $C_f: \seq{\hole}{f_1(x)\cdot \ldots \cdot f_n(x)} \in \contexts{\typereal}$--observe that this context is indeed in $\Lambda_{\probb}$, by hypothesis on $(\mathcal C_n)_{n \in \NN}$ and $(\mathcal B_n)_{n \in \NN}$, and since $\makereal 1 \in \Lambda_{\probb}$. Recall from Definition~\ref{definition:monadic-semantics} that for any real-typed term $\termthree$, $\sem { {\ctxone_f}\holefill \termthree} = f_{\star} (\sem \termthree)$, where $f_{\star}$ stands for the push-forward operation for $f$.  From there, we can see:
\begin{align}
\int_{x \in \RR} x \cdot   d\sem {\ctxone_f  \holefill {\termthree}}  =  &\int_{x \in \RR} x \cdot d(f_{\star} (\sem \termthree)) \\
= &\int_{x \in \RR} f(x) \cdot d (\sem \termthree), \label{eq:change_of_variable}
\end{align}
where~\eqref{eq:change_of_variable} above holds by the change of variable formula--a standard result on push-forward measures.
By hypothesis, we know that $\int_{x \in \RR} x \cdot   d\sem {\ctxone_f \holefill {\ctxone \holefill {\termone}}} = \int_{x \in \RR} x \cdot d \sem{\ctxone_f \holefill {\ctxone \holefill {\termtwo}}}$, thus~\eqref{eq:change_of_variable} implies that $\int_{x \in \RR} f(x) \cdot d \sem {\ctxone \holefill {\termone}} = \int_{x \in \RR} f(x) \cdot d \sem {\ctxone \holefill {\termtwo}}$. Since it is true for every $f \in \mathcal M(\mathcal F)$, we can conclude by Lemma~\ref{lemma:expectation_as_observable_1} 
\end{proof}
%Observe that, if $\mathcal C_1$ is such that is contais all characteristic function $\Xi_A:= (r \in \RR \mapsto 1 \text{ if } r \in A, \, 0 \text{ otherwise })$ for each $A$ in a set $S \subseteq \borels \RR$, and $S$ generates the $\sigma$-algebra $\borels \RR$, those two definitions are equivalent.
\subsubsection{(Non)-Measurability of Context Equivalence.}
}\shortv{
	It is now instructive to observe that contextual equivalence is \emph{not} 
	a \emph{measurable relation}, meaning that 
	the graph of $\ctxeq$ is not a measurable (i.e. Borel, in our setting) 
	set.\footnote{Even more, we can rely on descriptive set theory~\cite{kechris} 
	to bound the non-measurability of contextual equivalence (and its equivalence classes), 
	showing, e.g., that the graph of $\ctxeq$ is a co-analytic set.}}
\longv{
It is now instructive to observe that  contextual equivalence is \emph{not}
	a \emph{measurable relation}, meaning that 
	the graph of $\ctxeq$ is not a measurable (i.e. Borel, in our setting) 
	set, as we show in Lemma~\ref{lemma:graph_equiv_ctx_analytic} below.
\footnote{Even more, we can rely on descriptive set theory~\cite{kechris} 
	to bound the non-measurability of contextual equivalence (and its equivalence classes), 
	showing, e.g., that the graph of $\ctxeq$ is a co-analytic set.}}
%We can however rely on descriptive set theory (see for instance~\cite{kechris} for an introduction  to this field) to bound the non-measurability of contextual equivalence (and its equivalence classes), showing, e.g., that the graph of $\ctxeq$ is a co-analytic set:
 %       %. to bound \emph{how much} non measurable they are. We can make this notion more precise using Descriptive set theory, that defines a hierarchy--the projective hierarchy--that starts from the Borel sets, and define bigger and bigger classes using projections and complementation.  More precisely, we can show  that they are in the class of \emph{co-analytic} sets
  %      those are defined as the complementary of an \emph{analytic} set, while analytic sets are defined as the set obtained as the continuous image of a Borel set. It means that (the graph of) context equivalence is in the first layer of the projective hierarchy.}

\begin{proposition}
\label{prop:contextual-equivalence-is-not-measurable}
There exists a choice of measurable functions in $\mathcal{C}$ 
such that the graph of $\ctxeq$ is not a Borel set.
\end{proposition}
\begin{proof}
We are going to use the existence of a non-Borel co-analytical set $X \subseteq \RR$ to build measurable primitive functions that lead to a non-Borel equivalence class for $\equiv^{ctx}$ (and thus to a non-Borel graph for $\equiv^{ctx}$). We know that such a set exists from descriptive set theory (see~\cite{kechris}, 14.2). Moreover, we also know (see~\cite{kechris}, 32.B) that there exists some open set $O$ of $\RR \times \mathcal N$, such that $X = \{y \mid \forall x, (x,y) \in O\}$: here $\mathcal N$ is the Baire space, that is known (\cite{kechris}, 3.4) to be homomorphic to the set of all irrational numbers. It means that we can embedd $O$ in $\RR \times \RR$, and we again obtain a Borel set (observe that $O$ is not an open set anymore, but we can show it is an open set minus a set of rational numbers, and since $\QQ$ is countable this removal operation preserves Borel sets). From $O$, we define a function $f_O : \RR^2 \rightarrow \RR$ defined as $f_O(x,y) = 1$ whenever $(x,y) \in O$, and $0$ otherwise. Observe that since $O \subseteq \RR \times \RR$ is a Borel set, it holds that $f_O$ is measurable.
We are now going to consider the (very simple) program $M:= \lambda x.0 : \typereal \rightarrow \typereal$, and show that its equivalence class for $\equiv^{ctx}$ is not a Borel set. We are able to caracterise exactly the equivalence class of $M$: it consists of all the $N:\typereal \rightarrow \typereal$ such that for every $r \in \typereal$, $\sem {N r} = \dirac 0$. (In order to show that formally, we can use our results of soundness and completeness of event bisimulation and state bisimulation respectively, or we can for instance use denotational arguments in the quasi-Borel space model).
Now, looking at the way we define the standard Borel space of terms as a countable co-product, we see that it is enough to show that there exists one pre-term $P$, such that $\{\vec r \mid P[\vec r] \equiv^{ctx} M \}$ is not measurable. We consider the pre-term $P$ with one hole defined by $P[r]:= \lambda x. \text{ if }f_O(x,r) = 1 \text{ then }0 \text{ else }1$. Using the caracterisation of $M$'s equivalence class we have previously established, we see that $\{r \mid P[r] \equiv^{ctx} M\} = \{r \mid \forall r' \in \RR, f(r,r') \in O\} = X$, and we can conclude from there.
\end{proof}

In other words, even if we insist, like we did, in constraining the basic 
building blocks of our calculus in such a way that everything stays measurable, 
the central concept of contextual equivalence somehow breaks the mold. There is 
however a way to recover measurability of contextual equivalence, which we 
will now describe.
\longv{\subsubsection{Recovering Measurability}}
\shortv{\subsection{Recovering Measurability}}
We consider a natural restriction of $\lambdap$, called $\lambdacont$, in
which we require primitive functions to be continuous.
\begin{definition}
The language $\lambdacont$ is the fragment of $\lambdap$ obtained by 
requiring $\mathcal{C}$ to only contain symbols denoting \emph{continuous} functions, 
and $\mathcal{B}$ to be empty.
\end{definition} 
Noticeably, what we obtained is not a \emph{different} calculus, but rather an 
instance of $\lambdap$ in which $\mathcal{C}$ and $\mathcal{B}$ 
satisfy certain conditions. Even if functions in $\mathcal{C}$ are required to 
be continuous, the presence in $\mathcal{C}_2$ of a (now, continuous) 
comparison operator $\mathit{op}_{\leq}$ is still required. Observe that 
sampling from discrete distributions is not available anymore from within 
$\lambdacont$, because standard tests on real numbers are not continuous. 
Discrete distributions could however be added natively to $\lambdacont$, e.g. 
by endowing the language with a standard binary choice operator $\oplus_p$ 
complementing $\sample$. For the sake of simplicity, we do not pursue this 
further, and work with a rather minimal language. Our proof techniques, 
however, are robust enough to ensure that our results would not be affected by 
the presence of $\oplus_p$.

We will extensively study the nature of program equivalence for $\lambdacont$ 
in Section~\ref{sect:feller} and Section~\ref{sect:completeness}. 
For the moment, we simply observe that moving from 
$\lambdap$ to $\lambdacont$ indeed ensures contextual 
equivalence to be measurable.

\begin{proposition}\label{prop:eq_context_lambdacont_Borel}
The graph of contextual equivalence on $\lambdacont$ is a Borel set.
\end{proposition}
\begin{proof}
First, since Borel sets are stable by countable union, we see that it is sufficient to show that for every type $\typeone$, $\equiv^{ctx}$ restricted to programs of type $\typeone$ is Borel. By definition of context equivalence, we can write:
\begin{align*}
\equiv^{ctx}_\typeone &= \bigcap_{\substack{C \text{ a pre-context }\\\text{with }k\text{ holes }}} \{(M,N) \mid \forall \vec r \in \RR^{k}, \, \sem {C[\vec r][M]} = \sem{C[\vec r][N]}\}\\
 &= \bigcap_{\substack{C \text{ a pre-context} \\ \text{ with }  k\text{ holes }}} \bigcap_{q \in \QQ \cup{+\infty}} \quad \{(M,N) \mid \forall \vec r \in \RR^{k}, \, \sem {C[\vec r][M]}(]-\infty,q]) = \sem{C[\vec r][N]}(]-\infty,q)\}
 \end{align*}
where\shortv{ a formal proof for the last equality can be found 
in~\cite{LV}.}\longv{ the last holds by Example~\ref{ex:enough_primitives} and 
Lemma~\ref{lemma:expectation_as_observable_1}.}
Borel sets are also closed by countable intersections, thus it is enough to show that for every pre-context $C$, and every $q \in \QQ \cup \{+ \infty\}$, the set $S_{C,q} :=\{(M,N) \mid \forall \vec r \in \RR^{k}, \, \sem {C[\vec r][M]}(]-\infty,q]) = \sem{C[\vec r][N]}(]-\infty,q)\}$ is Borel. As a first step, observe that we can rewrite $S_{C,q}$ as the complementary of:
\begin{equation}
\overline{S_{C,q}} := \{(M,N) \mid \exists r \in \RR^{k},(\vec r,M,N) \in A_{C,q} \},
\end{equation}
where $A_{C,q} = \{(r,M,N) \mid \sem {C[\vec r][M]}(]-\infty,q]) \neq \sem{C[\vec r][N]}(]-\infty,q)\} \subseteq \RR^{k} \times \terms_{\typeone} \times \terms_{\typeone}$.
It means that it is enough to show that the projection of $A_{C,q}$ is a Borel set. Observe that it is not true that the projection of a Borel set is a Borel set; however it is true that the projection of an open set is an open set, and thus a Borel set. Using the fact that now all our primitive functions are continuous, we are now able to show that $A_{C,q}$ is an open set.
%, and from there we wilel be able to conclude using the well-known fact that open sets are stable by projection.
 We write $A_{C,q}$  as: $$A_{C,q} = (h_C,q)^{-1}{(\{(a,b) \in \RR^{2} \mid a \neq b\})},$$ where $h_C,q : \RR^{k} \times \terms_{\typeone} \times \terms_{\typeone} \rightarrow \RR^2$ is the sequential composition $h_{C,q}: = f_C ; (g_q \times g_q)$, where $f_C$ and $g_q$ are defined thus:
\begin{align*}
& f_{C}: (\vec r,M,N) \in  \RR^{k} \times \terms_{\typeone} \times \terms_{\typeone} \mapsto (C[\vec r][M],C[\vec r][N]) \in \terms_{\typereal} \times \terms_{\typereal}\\
& g_q:P \in \terms_{\typereal}  \mapsto \sem{P} (]-\infty,q]) \in \RR.
\end{align*}
Looking at the definition of the standard Borel space $\terms$, we can see that elementary syntactic operations are not only measurable, but also continuous: from there, we obtain that $f_C$ is continuous. Moreover, enforcing continuity of primitive functions allows us to show that the function $g_q$ is also continuous: we will show formally this result in Section~\ref{sect:completeness} (it is a direct corollary of Proposition~\ref{theorem:weakly_conv_sem}). Since continuous functions are stable by composition and cartesian product, it means that also the function $h_{C,q}$ is continuous. It means that $A_{C,q}$ is the inverse image of an open $\RR^2$-set by a continuous function, thus it is also open, and that concludes the proof.
\end{proof}

Having defined the target calculi of this work and their associated 
notion of contextual equivalence, we can now proceed to the introduction of 
bisimulation-based equivalences and their logical and testing-based 
characterizations.

\section{On Labelled Markov Processes and Equivalences on Them}\label{sect:equivalences_on_lmps}
In this section, we briefly recap the definition of \emph{state}-bisimilarity, 
and we give the alternative notion of \emph{event} bisimilarity, introduced 
by \citet{danos2005almost,danos2006bisimulation}. A summary of all results 
presented in this section can be found in 
Figure~\ref{fig:comparaison_equivalences}. 

\begin{figure}
\scalebox{0.6}{
\begin{minipage}{.7\textwidth}
\begin{center}
\begin{tikzpicture}[set/.style={fill=cyan,fill opacity=0.1}]

 \draw[fill=yellow,fill opacity=0.05,rotate =0, draw=none] (0,1.3) ellipse (6.6cm and 4.5cm);
  \node at (0,-2.7){$\sim_{\statet} \, \subseteq \, \sim_{\event} \, = \, \sim_{\logic} \, = \, \sim_{\test} $};
  \draw[fill=cyan,fill opacity=0.1,rotate =25,draw=none] (2,0.2) ellipse (4.5cm and 3cm);
  \node at (3.2,3.2){\begin{minipage}{0.2 \textwidth} \centering Analytical \\ 
  State Space  \end{minipage}};
   \node at (-3.2,3.2){\begin{minipage}{0.2 \textwidth} \centering Countable \\ 
   Labels  \end{minipage}};
  \draw[fill=green,fill opacity=0.1,rotate =-25,draw=none] (-2,0.2) ellipse (4.5cm and 3cm);
  \node[draw,circle,minimum size =2.8cm,fill=red!50, fill opacity=0.5] (circle1) at (0,-0.5){};
  \node at (0,-0.5){\begin{minipage}{0.2 \textwidth} \centering Discrete States 
  and \\ Countable Labels  \end{minipage}};
  \node at (0,2){\begin{minipage}{0.12 \textwidth} \centering
\begin{align*}
  \sim_{\statet} &= \sim_{\event} % \\ = \sim_{logic}&= \sim_{test}
\end{align*}
  \end{minipage}};
  \node at (3.5,0.8){$\bullet$};
  \node at (3.5,0){
  \begin{minipage}{0.19 \textwidth} \centering
  LMP for the\\ stochastic\\$\lambda$-calculus
  \end{minipage}
};
 \node at (4.2,2.4){$\bullet$};
  \node at (4.2,1.9){\begin{minipage}{0.37 \textwidth} \centering
 From~\cite{clerc2019expressiveness} \\ $\sim_{\statet} \, \neq \, \sim_{\event}$
  \end{minipage} };
   \node at (-4,2.2){$\circ$};
  \node at (-4,1.7){\begin{minipage}{0.3 \textwidth} \centering
 From~\cite{terraf2011unprovability} \\ $\sim_{\statet} \, \neq \, \sim_{\event}$
  \end{minipage} };
   % \node at (0,3){\begin{minipage}{0.15 \textwidth} \centering
%\begin{align*}
 % \sim_{event}  \\ &= \sim_{logic}
%\end{align*}
 % \end{minipage}};

  \end{tikzpicture}
\end{center}
\end{minipage}}
\caption{Comparison of LMP-equivalences}\label{fig:comparaison_equivalences}
\end{figure}
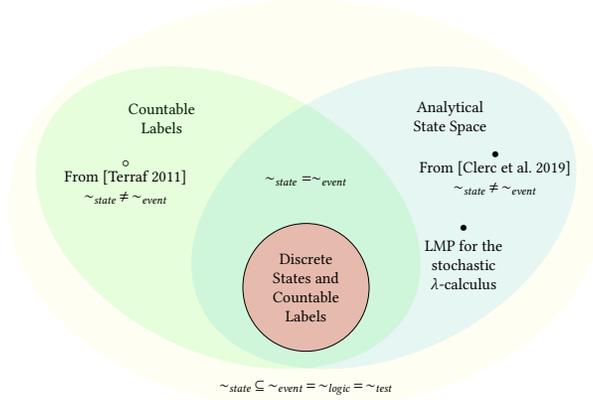

\subsection{Relational Reasoning} 

Before introducing bisimulation-based equivalences, it is convenient
to introduce some basic definitions on measurable spaces, relational
reasoning on them, and the crucial notion of a relational extension on 
probability measures. We denote by $\relone, \reltwo, \hh$ binary relations,
writing $\idrel$ for the identity relation (with subscripts, when
necessary), $\relone;\reltwo$ for relation composition (in
diagrammatic order), and $\converse{\relone}$ for the converse (or
transpose) of $\relone$.  Even if most of the results presented in
this section hold for arbitrary (endo)relations, it is convenient to
work right from the beginning with equivalence relations: given a set
$X$, we denote by $\rels(X)$ the complete lattice of equivalence
relations on $X$.  Finally, when reasoning relationally, we implicitly
view any function $f$ as a relation (via its graph).

The crucial ingredient to define bisimulation-like equivalences on
Markov processes (as introduced later in the paper) is the notion of an
extension of relations over a space $X$ to relations over
$\Delta(X)$. Relational extensions of that form are known in the
literature as \emph{relators} or \emph{lax extensions}
\cite{DBLP:conf/ifip2-1/BackhouseH93,Barr/LMM/1970}.  In this paper,
we are mostly concerned with the notion of relational extension used by
\citet{LagoG19} to define applicative bisimilarity on stochastic
$\lambda$-calculi. At a high level, this extension is obtained as the
codensity lifting of the Giry
monad~\cite{DBLP:journals/lmcs/KatsumataSU18}.  To define such an
extension formally, let us first recall that given a space
$(X, \Sigma_X)$ and a relation $\relone \in \rels(X)$, a set
$A \subseteq X$ is $\relone$-closed if $\relone[A] = A$. Moreover, we
define $\Sigma(\relone)$ as the sub-$\sigma$-algebra of $\Sigma$
defined as:
$\Sigma(\relone) = \{A \in \Sigma \mid A \text{ is
}\relone\text{-closed}\}$.

\begin{definition}[Probabilistic Relation Lifting]
\label{def:relation-lifting}
Let $(X, \Sigma)$ be a measurable space.
Define $\relator: \rels(X) \rightarrow \rels{(\distrs X)}$ by:\footnote{
  Notice that if $\relone$ is an equivalence, then so is $\relator{\relone}$.
}
%for $\distrone, \distrtwo \in \rels{(\distrs X)}$, 
$$\distrone \mathrel{\relator \relone} \distrtwo \iff \forall A \in \Sigma(\relone).\ 
\distrone(A) = \distrtwo(A). 
$$
\end{definition}

\longv{
The map $\relator$ essentially defines a relator \cite{} for the Giry monad, 
as stated by the following result.
\begin{proposition}
The map $\relator$ is monotone. Moreover, we have:
\begin{align*}
&\idrel_{\distrs X} \subseteq \relator{(\idrel_X)},
\text{ }
{\relator{\relone}}; {\relator{\reltwo}} \subseteq \relator{(\relone;\reltwo)},
\text{ }
\forward{f} \subseteq \relator{f},
\text{ }
\converse{\forward{f}} \subseteq \relator(\converse{f})
\\
&\relone; \unit_Y \subseteq \unit_X; \relator{\relone},
\quad
\relone; g \subseteq f; \relator{\reltwo} \implies \relator{\relone}; \kleisli{g} 
\subseteq \kleisli{f}; \relator{\reltwo}.
\end{align*}
\end{proposition}

\begin{proof}
By noticing--see~\cite{katsumata2015codensity}--that for any $\distrone,\distrtwo \in \Delta(X)$ and $\relone \in \rels(X)$, 
one has $\distrone \mathrel{\relator\relone} \distrtwo$ if and only if 
$\int_X f\de{\distrone} \leq \int_X f \de{\distrtwo}$, for any measurable function
$f: X \to [0,1]$ such that $x \mathrel{\relone} y \implies f(x) \leq f(y)$.
\end{proof}

In the literature, $\relator \relone$ is often considered as a relation on \emph{proper} distributions over $X$, instead of all sub-distributions as here. In order to bridge this gap, we can encode sub-distributions as proper distributions over the set $X \cup \{\bot\}$, where $\bot$ is a new point designed to represent divergence. We formalize this idea in Lemma~\ref{lemma:aux_subdistrs} below.
  \begin{notation}\label{notation:distrib_bot}
We can see a sub-distributions on $Z$ as a (proper) distribution on $Z_{\bot}:= Z \sqcup \{\bot\}$, where $\sqcup$ is the disjoint union operation on polish spaces. If $\eta$ is a sub-distribution on $Z$, we note $\eta_\bot$ the corresponding distribution on $Z_{\bot}$. Similarly, if $R$ is a binary relation on $X \times Y$, we write $R_{\bot}$ for the binary relation on $X_\bot \times Y_\bot$ defined as $R \cup \{\bot_X,\bot_Y\}$.
  \end{notation}
   \begin{lemma}\label{lemma:aux_subdistrs}
    Let $R$ be a binary relation on $X \times Y$, and $\mu$, $\nu$ two sub-disributions on respectively $X$ and $Y$. Then $\mu\, \Gamma R \,\nu$ if and only if $ \mu_{\bot} \,\Gamma R_{\bot}\, \nu_{\bot}$.
    \end{lemma}

}

When working with standard Borel spaces, 
there is another well-known relational extensions corresponding to 
the so-called Wasserstein-Kantorovich distance \cite{villani2008optimal}. 
Such an extension 
crucially relies on the notion of a probabilistic coupling. 

\begin{definition}
Let $(X, \Sigma_X)$, $(Y, \Sigma_Y)$ be standard Borel spaces 
and $\distrone \in \Delta(X)$, $\distrtwo \in \Delta(Y)$
be probability 
    measures. A \emph{coupling} 
    of $\distrone, \distrtwo$ is a probability measure 
    $\omega \in \Delta(X \times Y)$ such that for all $A\in \Sigma_X$, 
    $B \in \Sigma_Y$, we have: 
    $\omega(A \times Y) = \distrone(A)$ and $\omega(X \times B) = \distrtwo(B)$. 
    We denote by $\Omega(\distrone,\distrtwo)$ the set of 
    couplings of $\distrone$ and $\distrtwo$. Oserve that such a set 
    is non-empty (e.g., product distributions form a coupling).
\end{definition}
We now extend a relation $\relone$ to probability measures 
by asking the existence of a coupling compatible with $\relone$ 
(cf. the existential character of the following extension 
with the universal one of Definition~\ref{def:relation-lifting}).
\begin{definition}\label{def:relator_two}
Let $(X, \Sigma_X)$ be a standard Borel space. We define 
$\relatortwo: \rels(X) \to \rels(\Delta(X))$ by:
$$
\distrone \mathrel{\relatortwo\relone} \distrtwo \iff 
\exists \omega \in \Omega(\distrone, \distrtwo). \exists \reltwo \in \Sigma_{X \times X}.\ 
\reltwo \subseteq \relone \mathrel{\&} \omega(\reltwo) = 
\omega(X \times X).
$$
\end{definition}

It is easy to see that ${\relatortwo{\relone}} \subseteq {\relator\relone}$, 
although $\relator$ and $\relatortwo$ define different 
relational extensions, in general.\footnote{
  In fact, whereas $\relator$ gives a lax relational extension of $\giry$, 
  the map $\relatortwo$ does not. In particular, the map $\relatortwo$  
  corresponds to the so-called Barr lifting \cite{Barr/LMM/1970} 
  of $\giry$. Good properties (viz. being a lax extension) of 
  such a lifting are deeply connected with 
  weak-pullback preservation of $\giry$ which, in general, does not hold 
  \cite{DBLP:conf/calco/Viglizzo05}. 
} This is due to the fact that we are 
working with arbitrary
equivalence relations, which do not reflect the (standard) Borel structure of the spaces 
considered. Moving to relations reflecting such a structure --- known as 
\emph{Borel} relations ---
%or \emph{stable} relations ---
we indeed obtain the desired 
equality.  
%\textcolor{blue}{Is it true, that an equivalence relation is Borel if and only if its graph $\subseteq X \times X$ is measurable ? (if it is the case, maybe we could simplify the definition below)}
%\begin{definition}[Borel relation]
%Given a Polish space $(X,\Sigma_X)$ and $\relone \in \rels(X)$, we say that 
%$\relone$ is a \emph{Borel} relation if there exists a countable family 
%$\mathcal{C}$ of Borel subsets of $X$ such that:
%$$x \mathrel{\relone} y \iff
%\forall B \in \mathcal{C}.\ (x \in B \leftrightarrow y \in B).
%$$
%\end{definition} 
\begin{definition}
Given a standard Borel space $(X,\Sigma_X)$ and $\relone \in \rels(X)$, we say that $\relone$ is a \emph{Borel} relation when the set $\{(x,y) \mid x \relone y\}$ is a Borel set in $X \times X$.
\end{definition}

%\longv{
%Borel relations naturally 
%arise as notions of relation when working with Polish spaces and 
%\emph{Borel maps}. In fact, they are precisely kernels of Borel functions.

%\begin{lemma}
%Let $(X, \Sigma_X)$ be a Borel space and $\relone$ be a relation on it. 
%Then $\relone$ is Borel if and only if $\relone = f;\dual{f}$, for some 
%Borel space $(X, \Sigma_Y)$ and Borel function $f: X \to Y$. 
%\end{lemma}
%}

The advantage of working with Borel relations is that they are extremely 
well-behaved. 
%The identity relation is Borel, and the class of Borel relations 
%is closed under standard relation operations. But most importantly, 
In particular, Borel relations allow us to generalise Monge-Kantorovich duality 
to the measurable case, a result that we will need in 
Section~\ref{sect:completeness} .
%(this result is a direct consequence of a duality result for optimal coupling stated in~\cite{kellerer1984duality},~\cite{beiglbock2009general}, more details of the proof can be found in the long version of the present paper).

\begin{theorem}
\label{thm:theta-equal-gamma}
Let $(X,\Sigma_X)$ be a standard Borel space and $\relone \in \rels(X)$ be a 
Borel relation. Then $\relator{\relone} = \relatortwo{\relone}$.
\end{theorem}
\begin{proof}
It is a result from~\cite{katsumata2015codensity} that for any $\distrone,\distrtwo \in \Delta(X)$ and $\relone \in \rels(X)$, 
one has $\distrone \mathrel{\relator\relone} \distrtwo$ if and only if 
$\int_X f\de{\distrone} \leq \int_X f \de{\distrtwo}$, for any measurable function
$f: X \to [0,1]$ such that $x \mathrel{\relone} y \implies f(x) \leq f(y)$. From there, we can conclude by applying the Monge-Kantorowitch duality result for Borel cost functions from e.g.~\cite{kellerer1984duality}~\cite{beiglbock2009general}.
\end{proof}

% \begin{proposition}[\cite{}]
% Let $(X, \sigma)$ be a Polish space, $\distrone, \distrtwo \in \distrs X$, 
% and $\relone \in \rels(X)$.
% Then $\relaop \distrone {\relator \relone}  \distrtwo$ if and only if there exists a distribution $\couplingone \in \distrs {X \times X}$ such that:
% \begin{itemize}
% \item the first and second marginals of $\couplingone$ are $\distrone$ and $\distrtwo$ respectively;
% \item there exists $A \in \Sigma_{X\times X}$, such that $A \subseteq \relone$, and $\couplingone(A) = 1$.
% \end{itemize}
% \end{proposition}

% In the literature, distributions $\pi$ are usually called 
% \emph{couplings} or \emph{transference plans}.

\subsection{State Bisimilarity}\label{sect:statebisim}
The first notion of equivalence we consider is the natural generalization of 
bisimilarity to transition systems on continuous state spaces.
The natural generalization of a (probabilistic) transition system to continuous 
states is the notion of a Labelled Markov Process (LMP, for short).

\begin{definition}[Labelled Markov Processes]
  A \emph{Labelled Markov Process} (LMP) $\lmpone$  is a triple $(\markovone, 
  \actsone, \{h_a \mid a \in \actsone\})$, where $\markovone$ (the set of 
  states) is measurable, $\actsone$ (the set of labels) is an arbitrary set, 
  and for every $a \in \actsone$
  the map $h_a: \markovone \times \Sigma_{\markovone} \rightarrow [0,1]$
  is a sub-probability kernel. If moreover $\actsone$ is a measureable set, and 
  for every $s \in \markovone$ the map $a \in \actsone \times A \in 
  \Sigma_{\markovone} \mapsto h_a(s,A)$ is a sub-probability kernel, then we 
  say that $\lmpone$ is a \emph{Measurably Labelled Markov Process}.
\end{definition}

In the literature, the set $\actsone$ is very often taken to be countable. 
However, this is not the case for the LMPs which model parameter passing in 
calculi whose set of terms is itself not countable, since labels \emph{include}  
terms. Notice that a LMP 
$(\markovone, \actsone, \{h_a \mid a \in \actsone\})$ gives a $\actsone$-indexed 
family of measurable maps $h_a: \markovone \to \giry(\markovone)$. Exploiting this 
observation, a state bisimulation is defined in terms of relational extensions 
for the Giry monad. 

\begin{definition}[State Bisimulations]\label{def:state_bisimulation}
  Given a LMP $\lmpone$, a \emph{state bisimulation relation}
  is a relation $\relone \in \rels(\markovone)$ such that, 
  for all states $s,t$, if 
  %$\relone; h_a \subseteq h_a; \relator{\relone}$, for any $a \in \actsone$.
  $s \mathrel{\relone} t$ holds, then so does
  $h_a(s) \mathrel{\relator{\relone}} h_a(t)$, for any label $a$. 
\end{definition}

%Spelling out Definition~\ref{def:state_bisimulation}, a relation 
%$\relone \in \rels(\markovone)$ is a state bisimulation if
%$s \mathrel{\relone} t$ implies 
%$h_a(s) \mathrel{\relator{\relone}} h_a(t)$, for any label $a$. 
We say that $s$ and $t$ are \emph{state bisimilar} --- and write 
$s \stackrel{\lmpone}{\sim_{\statet}} t$ --- if there exists a state
bisimulation $\relone$ such that $s \mathrel{\relone} t$. 
When the LMP $\lmpone$ is clear from the context, we write 
$\sim_{\statet}$ 
in place of $\stackrel{\lmpone}{\sim_{\statet}}$ (and adopt the same 
convention for all the equivalences defined in this paper).

Notice the crucial presence of the map $\relator$ in 
Definition~\ref{def:state_bisimulation}. In fact, one 
can obtain another definition of state bisimilarity simply 
by replacing $\relator$ with $\relatortwo$ in 
Definition~\ref{def:state_bisimulation}. Since $\relatortwo{\relone} 
\subseteq \relator{\relone}$ holds for any relation $\relone$, 
bisimilarity based on $\relatortwo$ is included 
in bisimilarity based on $\relator$, and the two coincide when 
they are Borel relations (Theorem~\ref{thm:theta-equal-gamma}).
Unless otherwise specified, state bisimilarity 
will always refer to Definition~\ref{def:state_bisimulation}, 
and thus to the map $\relator$. The reason behind such a choice 
is twofold: on the one hand, $\relator$-based bisimilarity is 
the one usually finds in the literature; on the other hand, 
$\relator$ has in general nicer properties than $\relatortwo$ ensuring, 
e.g., $\relator$-based bisimilarity to be an equivalence relation. 

\subsection{Logical Equivalence}
Another classic approach to define program equivalence is via logic~\cite{LarsenS91,DBLP:conf/lics/DesharnaisEP98}. Informally, formulas of the logic are used as
discriminators between programs, and two programs are equivalent if
they yield the same interpretation for all formulas.  In our setting,
such a logic takes the form of a probabilistic modal logic.

\begin{definition}[Probabilistic Modal Logic]\label{def:probabilistic_moda_logic}
  Given a set $\actsone$ of actions, we define a class of logical formulas
  by way of the following grammar, where $a\in\actsone$ and $q\in \mathbb Q 
  \cap [0,1]$:
  \begin{align*}
    \phi \in \Logic &\bnf \top \midd \phi_1 \wedge \phi_2 \midd \act a q \phi.
  \end{align*}
  The semantics of formulas is given parametrically on a LMP 
  $(\markovone, 
  \actsone, \{h_a \mid a \in \actsone\})$
  by associating to each formula in $\Logic$ an element of $\Sigma$
  as follows,  where $\act a r A = \{s \mid h_a(s,A) > r \}$.
  \begin{align*}
    \sem \top &= \markovone; 
    &
    \sem{\phi_1 \wedge \phi_2} &= \sem {\phi_1} \cap \sem{\phi_2}; 
    &
    \sem {\act a q \phi} &= \act a q {\sem \phi}.
    \end{align*}
\end{definition}

 We say that two states $s,t \in \markovone$ are \emph{equivalent}
  with respect to $\Logic$ --- and write $s\sim_{\logic} t$ --- when for every
  $\phi \in \Logic$, it holds that $(s \in \sem \phi \Leftrightarrow t \in \sem  \phi)$.

\subsection{Testing Equivalence}
The idea behind testing equivalence \cite{LarsenS91,van2005domain} is close to the one behind logical 
equivalence: in both cases, syntactic entities play the role of state 
discriminators. However, in a testing scenario one is allowed to use tests 
whose outcome is not necessarily a \emph{truth value}, but rather the 
probability of passing the underlying test. The shift from a boolean to a 
quantitative setting allows one to define equivalence between states of a LMP 
using a restricted grammar of tests. 
%\subsection{Testing Caracterisations}
%Using Theorem~\ref{th:event_bisim_approx}, we can extend the well-known testing caracterisation for state bisimilarity on a LMC with countable labels.

\begin{definition}\label{def:test_language}
 Let $\actsone$ be a set of labels. We consider the following test language:
  $$\tests := \omega \midd \testone \wedge \testone \midd a \cdot \testone \qquad \text{ with } a \in \actsone. $$
 For a given $\lmpone = (\markovone, \Sigma, \{h_a \mid a \in \actsone\} )$, we define by induction on tests the \emph{success probability of the test $\testone$ starting from a state $s$} as the measurable $1$-bounded function $\probsucc \testone \cdot : \markovone \rightarrow [0,1] $:
  \begin{align*}
   \probsucc{\omega}{s} &=1; &
  \probsucc{a\cdot\testone}{s} &=\int\probsucc{\cdot}{\testone}dh_{a}(s,\cdot); &
  \probsucc{\testone\wedge\testtwo}{s} &=\probsucc{\testone}{s}\cdot 
  \probsucc{\testtwo}{s}. 
\end{align*}
\end{definition}

Notice that since tests have numerical outcomes, we do not need modalities 
in the test language.

\begin{definition}
Let $\lmpone$ be a LMP. We write $s \sim_{\test} t$ if $\probsucc s \testone = \probsucc t \testone$ for every $\testone \in \tests$. 
\end{definition}

\subsection{Event Bisimilarity}
The last notion of equivalence on LMPs we consider is \emph{event  
bisimilarity}; it has been introduced in 
\cite{danos2005almost,danos2006bisimulation} as a coinductively defined notion 
of equivalence larger than state bisimilarity and coinciding with testing and 
logical equivalence for a wide class of LMPs. From a categorical viewpoint, it 
can be seen as going from a span-based notion of bisimulation to a 
co-span-based one. From an operational point of view, it is obtained by 
shifting the focus from equivalent states to visible \emph{events} 
--- represented as sub $\sigma$-algebras. In order to introduce event 
bisimilarity, it useful to recall how state bisimilarity can be characterized 
measure-theoretically:
\begin{lemma}\label{lemma:carac_state_bisim}
  Given an LMP $(\markovone, \Sigma, \{h_a: \mid a \in \actsone\})$, a
  binary relation $\relone$ is a bisimulation if and only if
  $(\markovone, \Sigma(\relone), \{h_a \mid a \in \actsone\})$ is a
  LMP.
\end{lemma}
Event bisimilarity arises as a variation on the characterization
in Lemma~\ref{lemma:carac_state_bisim} obtained by removing the constraint
about the underlying $\sigma$-algebra:
\begin{definition}
  An event bisimulation on a LMP $(\markovone, \Sigma, \{h_a: \mid a
  \in \actsone\})$ is a sub-$\sigma$-algebra $\Lambda$ of $\Sigma$,
  such that $(\markovone, \Lambda, \{h_a \mid a \in \actsone\})$ is a
  LMP. Let $\Lambda$ be a $\sigma$-algebra on $\markovone$. We note
  $\rela \Lambda$ for the binary relation on states defined by $\rela
  \Lambda = \{(s,t) \mid \forall A \in \Lambda, s \in A
  \Leftrightarrow t \in A \}$.
  We say that two states $s,t \in \markovone$ are event bisimilar --- and 
  write $s \sim_{\event} t$ --- when
  there exists an an event bisimulation $\Lambda$ such that $(s,t) \in
  \rela \Lambda$.
%  $$\forall A \in \Lambda, \, (s \in \Lambda \Leftrightarrow t \in \Lambda). $$
\end{definition}

\subsection{Comparing the Equivalences}

Now that we have introduced four notions of equivalence on the state
of LMPs --- viz. \emph{state} and \emph{event} bisimilarity, and
\emph{logical} and \emph{testing} equivalence --- it is natural to
consider the correspondence between them.
As long as one focus on systems with discrete states (and
countable actions) all the equivalences mentioned so far coincide~\cite{DBLP:journals/tcs/BreugelMOW05, danos2006bisimulation, fijalkow2017expressiveness, danos2005almost}.
The correspondence still holds if LMPs have a \emph{countable} number
of labels and their underlying states space is \emph{analytic}:

\begin{theorem}\label{prop:logical_caracterisation_state_c}
  Let $\lmpone$ be a LMP with analytic state space and
  countably many labels. Then:
  $${\stackrel{\lmpone}{\sim_{\statet}}} = {\stackrel{\lmpone}{\sim_{\logic}}} 
  = {\stackrel{\lmpone}{\sim_{\event}}} = {\stackrel{\lmpone}{\sim_{\test}}}.
  $$
\end{theorem}
  
\begin{proof}
\citet{fijalkow2017expressiveness} have proved 
${\stackrel{\lmpone}{\sim_{\statet}}} = {\stackrel{\lmpone}{\sim_{\logic}}}$, whereas 
\citet{danos2005almost} have shown 
${\stackrel{\lmpone}{\sim_{\statet}}} = {\stackrel{\lmpone}{\sim_{\event}}}$. 
The remaining identities actually hold on arbitrary LMPs, as we are going to see.
\end{proof}

%\longv{
%Additionally, event bisimilarity and state
%bisimilarity coincide~\cite{danos2005almost} when the underlying space is analytic, and the
%set of labels $\actsone$ countable (the restriction on the set of
%labels
%is needed when proving Proposition \francesco{3.17} from Proposition \francesco{4.5}).
%However, it is always the case that state bisimilarity is coarser \francesco{finer?} 
%than event bisimilarity:
%\begin{proposition}[\cite{danos2006bisimulation}]\label{proposition:state_incl_event_bisim}
%\francesco{(${\sim_{\statet}} \subseteq {\sim_{\event}}$)} 
 % Let $\lmpone$ be a LMP, and $s,t$ two states that are
 % state-bisimilar. Then $s,t$ are also event-bisimilar.
%\end{proposition}

%The advantage of event bisimilarity over state bisimilarity --- at
%least for our purposes --- is that the logical characterization holds for
%non-analytic spaces, and \emph{without} any restriction on the cardinality
%of set of labels.
%}

Moving from LMPs with an analytic state space and countable actions 
to arbitrary LMPs, the logical 
and testing characterizations of event bisimilarity remains 
valid~\cite{DBLP:journals/tcs/BreugelMOW05, danos2006bisimulation}, 
whereas state bisimilarity has been proved to be finer then event 
bisimilarity~\cite{danos2006bisimulation}.

\begin{theorem}[\cite{DBLP:journals/tcs/BreugelMOW05, danos2006bisimulation}]\label{prop:logical_caracterisation}
  Let $\lmpone$ be any LMP. Then 
  $$
  {\stackrel{\lmpone}{\sim_{\statet}}} \subseteq 
  {\stackrel{\lmpone}{\sim_{\event}}}
  = {\stackrel{\lmpone}{\sim_{\logic}}} 
  = {\stackrel{\lmpone}{\sim_{\test}}}.
  $$
\end{theorem}

\begin{proof}
\citet{danos2006bisimulation} have proved 
${\stackrel{\lmpone}{\sim_{\statet}}} \subseteq 
{\stackrel{\lmpone}{\sim_{\event}}}$,
whereas \citet{DBLP:journals/tcs/BreugelMOW05} and \citet{danos2006bisimulation} 
have shown ${\stackrel{\lmpone}{\sim_{\event}}}
  = {\stackrel{\lmpone}{\sim_{\logic}}} 
  = {\stackrel{\lmpone}{\sim_{\test}}}$.
  \longv{
To be precise on bibliographical references for this last equality, we should note that--as clarified in a follow-up paper~\cite{DBLP:journals/entcs/MislovePW07}--the proof of the testing characterization for LMP given in~\cite{DBLP:journals/tcs/BreugelMOW05} works for event bisimulation with no restriction on states spaces. However, both in~\cite{DBLP:journals/entcs/MislovePW07} and~\cite{DBLP:journals/tcs/BreugelMOW05}, the labels are restricted to a countable set. It doesn't appear that this restriction is used in the proof, but nonetheless we can extend the testing characterization to the general case using the logical characterization from Theorem~\ref{prop:logical_caracterisation} and our approximation Lemma~\ref{th:event_bisim_approx}. 
}
\end{proof}

\longv{
\begin{remark}
Looking at Theorem~\ref{prop:logical_caracterisation}, we could ask ourselves wether there exist LMP such that all the equivalence notions we consider do not coincide. Of course, we know from Theorem~\ref{prop:logical_caracterisation_state_c} that such a LMP should have non-analytic state space, or uncountable labels. Such counter-examples have been built in the literature, for both cases:
%for the coincidence between state bisimulation and event bisimulation when going outside of the countable labels-analytic states space setting.
on the one hand, a LMP with uncountable labels, analytical states spaces have been presented~\cite{clerc2019expressiveness}; on the other hand, it has been shown~\cite{terraf2011unprovability} that, under the assumption that there exist non-measurable sets (which can be derived from the axiom of choice), there exists a LMP with countable labels, non analytic state space, where the two bisimulation do not coincide. 
\francesco{not sure I understand that}
\end{remark}
}

As corollaries of Theorem~\ref{prop:logical_caracterisation}, we obtain an approximation scheme for event bisimulation on LMPs with 
\emph{uncountably many} labels and that event bisimilarity on LMPs with countably many 
actions is a Borel relation, a result we will use in Section~\ref{sect:completeness} below.
\begin{lemma}\label{th:event_bisim_approx}
\begin{enumerate}
  \item Let $\actsone$ be a (possibly uncountable) set and let
  $(\markovone, \Sigma, \{h_a \mid a \in \actsone\} )$ be a LMP. 
  Then two states $s,t$ are event bisimilar in this LMP if and only if they are event-bisimilar in every LMP of the form
  $(\markovone, \Sigma, \{h_a \mid a \in \actstwo\} )$ where $\actstwo$ a \emph{countable} subset of $\actsone$.
  \item Let $\lmpone$ be a LMP with a countable set of actions, then 
  ${\stackrel{\lmpone}{\sim_{\event}}}$ is Borel.
\end{enumerate}
  \end{lemma}
  
\begin{proof}\shortv{[Proof sketch]}
\begin{enumerate}
\item The proof relies on the equality ${\stackrel{\lmpone}{\sim_{\event}}}
  = {\stackrel{\lmpone}{\sim_{\logic}}}$ 
  and the fact that to check whether a given formula $\phi$ holds for some state 
  $s$, it is sufficient to consider a finite numbers of labels. \shortv{More details are given 
  in the long version of the paper.}
  \longv{
%The proof uses Theorem~\ref{prop:logical_caracterisation}, and the fact that to check whether \emph{one} given formula $\phi$ holds for some state $s$, it is sufficient to consider a finite numbers of labels. 
Let $\lmpone = (\markovone, \Sigma, \{h_a \mid a \in \actsone\} )$ be a LMP.
  \begin{itemize}
  \item First, observe that the set $\Logic$ of logical formulas that caracterise event bisimilarity on $\lmpone$ contains the set of formulas that caracterise event bisimilarity on any of the $\lmpone' = (\markovone, \Sigma, \{h_a \mid a \in \actsone'\} )$ with $\actsone'$ a countable subset of $\actsone$. It means that two states $s,t$ event bisimilar in $\lmpone$ are also state bisimilar in $\lmpone'$.
    \item Let us consider now two states $s,t$ that are not event bisimilar in $\lmpone$: we need to show that there exists a countable subset $\actsone'$ of $\actsone$, such that $s,t$ are not event bisimilar in $\lmpone' = (\markovone, \Sigma, \{h_a \mid a \in \actsone'\} )$. To show that, we use again the logical caracterisation of event bisimilarity, and it gives us a formula $\phi \in \logic$ such that $s \in \sem \phi$, and $t \not \in \sem \phi$--or $t \in \sem \phi$, and $s \not \in \sem \phi$, but the reasonning would be identical in this case. Recall that $\phi$ is generated by the BNF grammar of Definition~\ref{def:probabilistic_moda_logic}, so in particular there is a \emph{finite} number of $a \in \actsone$ that occur in $\phi$. Looking at the Definition of $\sem \phi$, we see that it is identical in $\lmpone$ and $\lmpone' = (\markovone, \Sigma, \{h_a \mid a \in \actsone'\} )$, as long as all the action that occur in $\phi$ are in $\actsone'$. From there, we are able to conclude using Proposition~\ref{prop:logical_caracterisation}.
    \end{itemize}
}
\item We notice that if $\lmpone$ has countably many actions, then the set $\Logic$ 
  of logical formulas for $\lmpone$ is countable. Since for any such formula $\phi$ 
  the set $\sem{\phi}$ is Borel, we have that ${\stackrel{\lmpone}{\sim_{\logic}}}$, 
  and thus ${\stackrel{\lmpone}{\sim_{\event}}}$, is Borel.
\end{enumerate}
\end{proof}

\subsection{A LMP for $\Lambda_{\probb}$.}

The dynamic semantics of the language $\Lambda_{\probb}$ we introduced
in Section~\ref{sect:lambda} can be presented as a LMP, this way
enabling the various forms of equational reasoning principles
presented above.  In this section, we define the aforementioned LMP,
then discuss why this LMP is problematic, deferring a more thorough
discussion about it to Section~\ref{sect:perspective} below.
\longv{\subsubsection{Definition of the LMP $\lmponelambda$}

}
Actually, there is a standard way of turning (possibly 
effectful) typed $\lambda$-calculi into  labelled transition systems: states 
are partitioned into two  classes,  namely computations 
and values, and the environment interacts with the former  through an  
evaluation action and with the latter by inspecting the value (e.g. if it has 
a coproduct type) or by passing it an argument (if it is of an 
arrow type). We are going to follow this pattern, but an  additional  
difficulty is present here, namely the one of 
dealing with real numbers as values. How  should the environment inspect such 
a value? Ground types are generally managed  through  actions which fully 
reveal the underlying value. Remarkably, this works quite  well in  presence of 
discrete ground types and discrete probabilities. However, when we  go from  
discrete to continuous probabilities --- and thus from natural numbers to real  
numbers as  base type --- things become more complicated. In particular, a real 
number value can be inspected through the comparison operator 
$\mathit{op}_\leq$, which we ask to always be part of $\mathcal{C}$.

\begin{definition}\label{def:lmplambdap}
  We define $\lmponelambda$ as the LMP $(\markovone_{\Lambda_{\probb}}, \actsone_{\Lambda_{\probb}}, \{h_a \mid a \in   \actsone_{\Lambda_{\probb}}\} )$ where:
\begin{varitemize}
\item 
   The set of states is $\markovone_{\Lambda_{\probb}} = \bigcup_{\typeone \in 
   \types}\{(\termone,\sigma) \mid \termone \in \terms_\typeone\}\cup\{ 
   (\valone,\sigma) 
   \mid \valone \in \values_{\typeone}\}$; we take as topology on 
   $\markovone_{\Lambda_{\probb}}$ the countable disjoint union 
   topology, see~\cite{DBLP:journals/pacmpl/EhrhardPT18,borgstrom2016lambda} 
   (recall that $\Pol$ has countable coproducts).
 % \textcolor{blue}{say something on the Borel $\sigma$-algebra on $\markovone_{\Lambda_{\probb}}$ and $\actsone_{\Lambda_{\probb}}$, and remark that by definition, it is a polish space.}
\item 
  The set of actions is $\actsone_{\Lambda_{\probb}} = \types \cup 
  (\bigcup_{\typeone \in \types} \values_{\typeone}) \cup \{\evalact 
  \}
  \cup \{ \overset{\text{?}}{\leq} q \mid q \in \QQ \}
  %\cup \{f \mid f \in \mathcal C_1, f(\RR) \subseteq [0,1]\} 
  \cup \{\text{case}(\hat{\imath}) \mid i \in I\} \cup\{ \unboxact\}$. Again, we take 
  the disjoint 
  union topology on $\actsone_{\Lambda_{\probb}}$.
\item 
  The transition functions are defined as follows, where 
  $h_a(s) = \emptyset$ in all cases not listed below.
  \begin{align*}
    h_\typeone(s) &= 
    \dirac s 
    & & \mbox{if } s = (\termone,\typeone) \text { or } s = 
    (\valone,\typeone);
    \\
     h_\valone(s) &= 
    \dirac {(\valtwo \valone,\typetwo)}
    & &\mbox{if } s = (\valtwo,\typeone \to 
    \typetwo); 
    \\
     h_{\evalact}(s) &= 
    (\sem \termone,\typeone) 
    & & \mbox{if } s = (\termone,\typeone);
    \\
  h_{\overset{\text{?}}{\leq} q}(s) &= 
      op_{\leq}(r,q) \cdot \dirac s 
    & & \mbox{if }\exists r \in \RR, s = (\makereal r, 
    \typereal)
    \\
    h_{\text{case}(\hat{\imath})}(s) &= \dirac {(\valone,\typeone_i)}
    & & \mbox{if } 
    s = ((\hat{\imath},\valone), \sumtype{j \in I}{\typeone_j})
    \\
    h_{\unboxact}(s) &= \dirac 
    {(\valone,\subst{\typeone}{\typevar}{\rectype{\typevar}{\typeone}} )}
    & &\mbox{if } s = (\fold 
    \valone,  \rectype{\typevar}{\typeone})
  \end{align*}
  \end{varitemize}
\end{definition}

The LMP $\lmponelambda$ we have just introduced naturally captures the 
interaction between a term and the environment, in the spirit of Abramsky's 
applicative bisimilarity. In particular, \emph{state} bisimilarity on 
$\lmponelambda$ captures applicative bisimilarity 
on $\Lambda_{\probb}$~\cite{LagoG19}, meaning that 
$\termone, \termtwo \in \terms_{\typeone}$ are applicative 
bisimilar precisely when 
$(\termone,\typeone) \sim_{\statet}^{\lmponelambda} (\termtwo,\typeone)$ 
(oftentimes, we use the notation  
$\termone \sim_{\statet}^{\lmponelambda} \termtwo$)
and similarly for values.
Unfortunately, $\lmponelambda$ is not among those LMPs for which nice 
correspondence results exist between the various notions of equivalences. 
Although the 
\emph{state} space is analytic, indeed, the \emph{action} space is inherently 
uncountable, due to the presence of values among the labels. This is 
unavoidable, given the nature of applicative bisimilarity, and has some deep 
consequences, as we are going to see in Section~\ref{sect:perspective} below. 
But how about $\lambdacont$? Actually, one can naturally define 
$\lmponelambdac$ exactly as $\lmponelambda$, with the additional proviso that 
terms occurring in states and actions are taken from $\lambdacont$ rather than 
from $\lambdap$. In the next two sections we are going to show that going from 
$\lmponelambda$ to $\lmponelambdac$ indeed makes a difference as far as the 
nature of relational reasoning is concerned.

\longv{
\subsubsection{How To Act on Real Numbers}
  %\begin{remark}\label{remark:ground_type_actions}
    While the structure of $\lmponelambda$ with respect to evaluation, application of function to values 
  are quite standard for applicative bisimulations on typed 
  $\lambda$-calculi--see 
 for 
  instance the way the applicative bisimulation was defined 
  in~\cite{CDL14} 
  for 
  a discrete probabilistic $\lambda$-calculus--deciding which actions we put on $\typereal$ values is more subtle. The leading principle behind our choice was that for programs at base type, context equivalence, event bisimilarity and state bisimilarity should all coincide. Moreover, in order to avoid un-necessary complexity, we were looking for a \emph{countable} family of ground-type actions. It can be checked that these guidelines hold for the ground-type action we gave in Definition~\ref{def:lmplambdap}. It can also be observed that they would not hold for the--perhaps more intuitive--family of ground type actions indexed by $r \in \RR$  
   $(h_{r'}(\makereal {r}) 
   := (?(r = r')) \cdot \dirac {\makereal{r}})_{r' \in \RR}$: not only this family is not countable, but more importantly event bisimulation do not coincide with context equivalence for real typed programs. 
For instance, we can see using the testing caracterisation that the program $\sample$ and $(\sample + 1)$ would not be event bisimilar for this variation of $\lmponelambda$.
   %The leading principle is that the 
  %actions at 
  %base 
  %type should represent the \emph{observations} that we are able to do on those types. 
  %%%%%%%%In~\cite{CrubilleDalLago}, with natural numbers as base type, the 
  %%%%%%%%base type 
  %%%%%%%%actions were 
  %%%%%%%%%taken as the family $h_m$, for all natural numbers $n$, where  
  %%%%%%%%%$h_m(\makereal n) 
  %%%%%%%%%:= (?(n 
  %%%%%%%%%%= m)) \cdot \dirac n$. While this choice is natural from a testing 
  %%%%%%%%%%point of 
  %%%%%%%%%%view--each 
  %%%%%%%%%%%action corresponding to an atomic test, observe that any family of 
  %%%%%%%%%%%actions 
  %%%%%%%%%%%able to 
  %%%%%%%%%%%%\emph{separate} any pair of distinct natural numbers produce the 
  %%%%%%%%%%%%same 
  %%%%%%%%%%%%bisimulation.

   In this section, we explore in which sense we can say that our notions of state and event bisimilarities for $\lambdap$  are stable under changes to the way we define the actions on real type in $\lmponelambda$. To do that, we first define formally an applicative LMP parametric in which actions we choose on real type. We are actually interested on variation of the following form: an action on $a$ on a ground type value $\makereal r$ must do a loop on $\makereal r$, but this loop only succeeds with a probability that depends on $r$ and $a$. We formalize this idea in Definition~\ref{def:parametric_lmps} below.

   \begin{definition}\label{def:parametric_lmps}
    Let $\mathcal P$ be any family of measureable functions $\RR \to [0,1]$. We will write $\lmponelambda(\mathcal P)$ for the LMP obtained from $\lmponelambda$ by removing the actions ${\overset{\text{?}}{\leq} q}_{q \in \QQ}$, and adding the action $(h_p)_{p \in \mathcal P}$ where:
    $$h_p(s) = \begin{cases}
      p(r) \cdot \dirac{\makereal{r}} \text{ if } \exists r \text{ with } s = \makereal r; \\
\emptyset \text{ otherwise.}\end{cases}$$
%    We say that the family $\mathcal P$ \emph{separates all real}, when for every pair $r,r' \in \RR$ with $r \neq r'$, there exists $p \in \mathcal P$ such that $p(r) \neq p(r')$.    
   \end{definition}
   We first show a quite generic criteria for changes on ground-type action that preserves state bisimilarity for $\lmponelambda$: it says that whenever we keep a family of ground-type action that is generated by a family of function that separates points--in the sense of Remark~\ref{remark:altern_carac_ctx_eq}--then the state bisimilarity stays the same.
  \begin{lemma}\label{lemma:state_simulation_parametric}
For every  $\mathcal P$ that separates points, the state bisimulation on $\lmponelambda(\mathcal P)$ coincide with the state bisimulation on $\lmponelambda$.
  \end{lemma}
  \begin{proof}
    First, observe that the ground-type action on $\lmponelambda$ are generated by a family of function that separates points (the family $\mathcal F_{\leq}:= {(\overset{\text{?}}{\leq} q)}_{q \in \QQ}$). Now, we show that for every families $\mathcal P$, $\mathcal F$ that separates points any state bisimulation $\relone$ for $\lmponelambda(\mathcal F)$ is also a state bisimulation for $\lmponelambda(\mathcal P)$. Let be $s,t \in \markovone_{\lambdap}$ such that $s \relone t$, and $a \in \actsone_{\lambdap}$.  Let $a$ an action in $\lmponelambda(\mathcal P)$.  We want to show that $h_a^{\lmponelambda(\mathcal P)}(s) \Gamma \relone h_a^{\lmponelambda(\mathcal F)}(t)$. Observe that if $a$ is not in $\mathcal P$, the action $h_a^{\lmponelambda(\mathcal F)}$ coincide with $h_a^{\lmponelambda(\mathcal P)}$, so we can conclude using our hypothesis on $\relone$. Now, suppose that $a = p$, with $p \in \mathcal P$. Observe that since $s,t$ are bisimilar in $\lmponelambda(\mathcal F)$, then either $s,t$ are both $\typereal$ values, or none of them is. In the second case, we can immediately conclude. In the first case, since the original family of actions $\mathcal F$ separates all reals, we know that $s =t$, thus we can conclude.
    \end{proof}
  We look now for a criterion for event bisimilarity. As a first step, we formalize in Definition~\ref{def:Borel_generating} below the idea that a family of actions can generates all Borel sets on $\RR$. 
  \begin{definition}\label{def:Borel_generating}
We say that a family $\mathcal F$ be a family of measureable functions $\RR \to [0,1]$ is \emph{Borel-generating} whenever  the smallest $\sigma$-algebra containing all the sets $f^{-1}(I)$ for $I \in \borels{[0,1]}$, $f \in \mathcal F$ is the Borel algebra over $\RR$. 
    \end{definition}
  We can check easily that the family $\mathcal F_{\leq}:= {(\overset{\text{?}}{\leq} q)}_{q \in \QQ}$ is Borel generating.
  Now, Lemma~\ref{corollary:event_bisimulation_parametric_lmps_1} tells us that as long as the family of functions $\mathcal F$ we consider is Borel genrating, the event bisimilarity on $\lmponelambda(\mathcal F)$ stays the same as the event bisimilarity on $\lmponelambda$.
  %below tells us that to konw wether we obtain the same event bisimilarity on two LMPs $\mathcal P$, $\mathcal F$, it is enough to look at the states of the form $\makereal r$.
  \begin{lemma}\label{corollary:event_bisimulation_parametric_lmps_1}
Let $\mathcal F$ be Borel-generating family of functions.  Then $ \sim^{\lmponelambda}_{event}\, =\, \sim^{\lmponelambda(\mathcal F)}_{event}.$
  \end{lemma}
  The proof of Lemma~\ref{corollary:event_bisimulation_parametric_lmps_1} is more involved that the one of Lemma~\ref{lemma:state_simulation_parametric}, and as a consequence we first prove the auxilliary Lemma~\ref{lemma:aux_variant_event_bis} below:  it tells us that to show that we obtain the same event bisimilarity on  the two LMPs $\lmponelambda(\mathcal P)$ and  $\lmponelambda(\mathcal F)$, it is enough to look at the states of the form $\makereal r$.

\begin{lemma}\label{lemma:aux_variant_event_bis}
Let $\mathcal P$, $\mathcal F$ be two families of measureable functions $\RR \to \RR$, and $\Lambda_{\mathcal P}$, $\Lambda_{\mathcal F}$ the event bisimilarities in $\lmponelambda( \mathcal P)$ and $\lmponelambda( \mathcal F)$ respectively. We suppose that $\{ A \mid A \in \Lambda_{\mathcal P} \cap \valset_{\typereal}\} =: (\Lambda_{\mathcal P})_{\mid \valset_{\typereal}} \supseteq (\Lambda_{\mathcal F})_{\mid \valset_{\typereal}} :=\{ A \mid A \in \Lambda_{\mathcal F} \cap \valset_{\typereal}\}$. Then it holds that $\Lambda_{\mathcal P} \subseteq \Lambda_{\mathcal F}$ .
\end{lemma}

\begin{proof}
  We note $\lmponelambda^{\star}$ for the applicative LMP without action on reals, that is $\lmponelambda(\emptyset)$.
  We are going to show that $\Lambda_{\mathcal P}$ is an event bisimulation on $\lmponelambda( \mathcal F)$. To do that, we need to show that $(\markovone_{\lambdap}, \Lambda_{\mathcal P},\{h_a \mid a \in \actsone_{\mathcal F} \cup \actsone_{\lambdap^{\star}}\} )$ is a LMP, i.e. that for every $a \in \actsone_{\mathcal F} \cup \actsone_{\lambdap^{\star}} $, the restriction of $h_a$ to the map $h_a:\markovone_{\lambdap} \times \Lambda_{\mathcal P} \rightarrow [0,1]$ is a sub-probability kernel. First, observe that when $a \in \actsone_{\lambdap^{\star}}$, we obtain immediately the result using the fact that $\Lambda_{\mathcal P}$ is an event bisimulation on $\lmponelambda( \mathcal P)$, and moreover $h_a^{\lmponelambda(\mathcal F)}$ and $h_a^{\lmponelambda(\mathcal P)}$ coincide. Suppose now that  $a = f \in \mathcal F$: We are going to check that $h_f$ is indeed a kernel: 
%we can conclude using the fact that   $(\Lambda_{\mathcal F})_{\mid \valset_{\typereal}} \supseteq (\Lambda_{\mathcal F})_{\mid \valset_{\typereal}}$.  \begin{itemize}
\begin{itemize}
      \item For every $A \in \Lambda_{\mathcal P}$, the map $m_A: s \in (\markovone_{\lambdap},\Lambda_{\mathcal P}) \mapsto h_{f}(s,A) \in (\RR, \borels \RR)$ is measureable.  Let $I \in \borels \RR$: we see that $m_A^{-1}(I) = (m_A^{-1}(I) \cap \valset_{\typereal}) \cup
(m_A^{-1}(I) \cap  (\markovone_{\lambdap} \setminus \valset_{\typereal}))$. Observe first that $(m_A^{-1}(I) \cap  (\markovone_{\lambdap} \setminus \valset_{\typereal})) \in \{\emptyset, \markovone_{\lambdap} \setminus \valset_{\typereal}\} \subseteq \Lambda_{\mathcal P}$ (the last inclusion is given by the actions $a$ with $a \in \types$). Then, we observe that $m_A^{-1}(I)\cap \valset_{\typereal} = \{\makereal r \mid f(r) \in I \wedge \makereal r \in A\} = m_{\RR}^{-1}(I) \cap A$. From there, since  $(\Lambda_{\mathcal P})_{\mid \valset_{\typereal}} \supseteq (\Lambda_{\mathcal F})_{\mid \valset_{\typereal}}$, and $m_{\RR}^{-1}(I) \in  (\Lambda_{\mathcal F})_{\mid \valset_{\typereal}}$, it holds that $ m_A^{-1}(I)\cap \valset_{\typereal} \in \Lambda_{\mathcal P}$, and we can conclude.
\item For every $s \in \markovone_{\lambdap}$, the map $A \in \Lambda \mapsto h_{a}(s,A)$ is a sub-probability measure. We see that immediately from the fact that $\lmponelambda(\mathcal P)$ is a LMP, and $\Lambda$ is a sub sigma-algebra of $\Sigma_{\markovone_{\lambdap}}$.
        \end{itemize}
\end{proof}

%\begin{lemma}\label{lemma:event_bisimulation_parametric_lmps_1}
 %   Let $\mathcal F$ be any family of mesureable functions $\RR \to [0,1]$, and $\mathcal P$ be a family of measureable function such that the Borel $\sigma$-algebra over $\RR$ coincides with the sigma algebra generated by $\mathcal P$--i.e. the smallest $\sigma$-algebra that contains all the $p^{-1}(A)$ for  $A \in \borels{[0,1]}, p \in \mathcal P$. Then:
  %  $$ \sim^{\lmponelambda( \mathcal P)}_{event}\, \subseteq \, \sim^{\lmponelambda(\mathcal F)}_{event}.$$
   %% where $\mathcal B = \{(r \in \RR \mapsto (1 \text{ if }r \in A, \, 0 \text{ otherwise } )) \mid A \in \borels{\RR}\}$.
  %\end{lemma}
\begin{proof}[Proof of Lemma~\ref{corollary:event_bisimulation_parametric_lmps_1}]
  Let $\Lambda_{\mathcal F}, \Lambda$ be the event bisimilarities in $\lmponelambda( \mathcal F)$ and $\lmponelambda$ respectively.
  %Let $\Lambda_{\mathcal P}$, $\Lambda_{\mathcal F}$ be the event bisimilarities in $\lmponelambda( \mathcal P)$ and $\lmponelambda( \mathcal F)$ respectively.
  Observe that $\valset_{\typereal}$ is a self-contained fragment of both $\lmponelambda( \mathcal F)$ and $\lmponelambda$: as a consequence $\Lambda_{\mathcal F} \cap \valset_{\typereal}$, and $\Lambda \cap \valset_{\typereal} $ can be caracterised as the event bisimulations on the LMPs  $(\valset_{\typereal}, \Sigma_{\mid \typereal}, \{h_f \mid f \in \mathcal F_{\leq} \cup \{\typereal\}\})$ and $(\valset_{\typereal}, \Sigma_{\mid \typereal}, \{h_f \mid f \in {\mathcal F} \cup \{\typereal\}\})$ respectively. Using the definition of event bisimulation, we obtain that for every $f \in \mathcal F$, the map
  $\makereal r \in (\valset_{\typereal}, \Lambda_{\mathcal F} \cap \valset_{\typereal} ) \to h_f(r,\RR) = f(r) \in (\RR,\borels \RR)$ is measureable. It is a known result from the literature--see for instance~\cite{}, caracterisation of a $\sigma$-algebra generated by a collection of mappings--that this condition implies that $\Lambda_{\mathcal F}\cap \valset_{\typereal} $ contains the $\sigma$-algebra generated by $\mathcal F$. By hypothesis, the $\sigma$-algebra generated by $\mathcal F$ coincides with $\borels \RR$, meaning that $\borels \RR = \Sigma\cap \valset_{\typereal} \supseteq (\Lambda_{\mathcal F} \cap \valset_{\typereal} \supseteq \borels \RR \supseteq (\Lambda \cap \valset_{\typereal})$, and from there we can conclude by applying Lemma~\ref{lemma:aux_variant_event_bis}. For the reverse inclusion, we can do the same reasonning in exchanging $\mathcal F$ and $\mathcal F_{\leq}$, using the fact that also $\mathcal F_{\leq}$ is Borel-generating.

  \end{proof}

\ugo{The following constraint should rather go in the Introduction. Otherwise 
it seems 
rather ad-hoc here. Example 2, instead, should I think go to a dedicated 
section at the end 
of the paper.}
%\begin{remark}\label{remark:primitive_rational_tests}
%In the following, we will suppose that all functions $r \in \RR \mapsto ?(r < q)$ are in $\mathcal{C}_1$ for $\Lambda_{\probb}$. As a consequence--see Remark~\cite{remark:ground_type_actions}, Lemma~\ref{lemma:state_simulation_parametric} and Corollary~\ref{corollary:event_bisimulation_parametric_lmps_1}, we can replace the family of ground type actions $(h_f(\makereal {r'}) := (f(r') \cdot \dirac {r'})_{f: \RR \to [0,1] \in \mathcal{C}_1}$ by $(h_{q}(\makereal {r'}) := (?(q < r')) \cdot \dirac {r'})_{q \in \QQ}$, whithout modifying neither the state nor the event bisimulation. 
 % \end{remark}
%\begin{example}\label{example:choice_ground_actions_distinct}
 % Here, we illustrate the point made in Remark~\ref{remark:ground_type_actions} above. Let us consider the programs $\sample$ and $\sample +1$, of type $\typereal$. Observe that they are not bisimilar on $\lmponelambda$: we can see that by the testing caracterisation, by considering the test $\testone := \evalact.(?(q<\cdot))$. However, if we replace the actions for $\typereal$ values with $(h_r(\makereal {r'}) := (?(r = r')) \cdot \dirac {r'})_{r \in \RR}$. Then we can show--again by looking at the testing caracterisation--that those two programs are event bisimilar. We can also show that they are not state bisimilar, no matter which of these two choices we take for ground type actions.
  %\end{example}
}

\section{Continuity to the Rescue}\label{sect:feller}
As already mentioned, most of the literature on relational reasoning on LMPs 
focuses on cases where labels are countable. A notable exception 
is~\cite{fijalkow2017expressiveness}, in which it is shown that adding 
a \emph{continuity} condition on the LMP transition function allows, in the 
case of \emph{uncountably} many labels, to recover the coincidence between 
some notions of equivalence. We now briefly present the setting 
from~\cite{fijalkow2017expressiveness}, showing that, unfortunately, it cannot be applied to LMPs like $\lmponelambda$ or $\lmponelambdac$.
\begin{definition}\label{def:ctf}
  Let $\lmpone = (\markovone, \actsone, \{h_a \mid a \in \actsone\})$ be a LMP, such that $\markovone$ is a standard Borel space. We say it has a \emph{continuous transition function} if, whenever  $a \in \actsone$, $X \in \Sigma_{\markovone}$, the function $s \in \markovone \mapsto h_a(s,X) \in [0,1] $ is continuous.
  \end{definition}

Unfortunately, the continuity requirement on the underlying transition function 
is not suitable for probabilistic transition systems built out of the 
operational semantics of higher-order programming languages like $\lambdap$. 
Intuitively---and as illustrated in 
Remark~\ref{remark:continuous_transion_function_operational} below --- this comes 
from the fact that Definition \ref{def:ctf} forces any transition to 
either have a fuzzy 
probabilistic behaviour, or to go into a constant state \emph{not} depending on 
the reals in the state of departure. As a consequence, it excludes 
deterministic (labelled) transitions in which a value is passed as an 
argument to a function. Since such steps 
exist independently on any continuity requirements, the continuous 
transition function constraint does not hold in (the LMPs underlying) any 
meaningful fragment of 
$\lambdap$. 
\begin{remark}\label{remark:continuous_transion_function_operational}
  We illustrate on an example why neither $\lmponelambdac$ nor $\lmponelambda$ 
  have continuous transition functions. We start from the action $\evalact$, 
  and the measureable set $X = \{\makereal 0\}$. We consider the family of 
  terms $\termone_n = \op(\makereal {r_n})$, where $(r_n)_{n \in \NN}$ is a 
  sequence of non-zero numbers converging to $0$, and $f$ is 
  the identity, thus a continuous function on $\RR$. In the standard Borel 
  space $\markovone_{\Lambdacont}$ (or $\markovone_{\Lambda_{\probb}}$), we define a sequence $(u_n)_{n \in \NN}$ by $u_n:=(\termone_n,\typereal)$. Observe that $u_n 
  \xrightarrow[n\rightarrow \infty]{} (\op{(\makereal 0)},\typereal)$ in the $\markovone_{\Lambdacont}$ (or $\markovone_{\Lambda_{\probb}}$) topology. By contrast, we 
  see that  $h_{\evalact}(\termone_n,\typereal)(X)=0$ for every $n$, while 
  $h_{\evalact}(\op(\makereal 0),\typereal)(X)=1$; thus $h_{\evalact}$ is not a 
  continuous transition function.
\end{remark}
We thus need to go towards a different notion of continuity, which is what we 
are going to do in the next subsection.

\subsection{Feller-Continuous LMPs}
\begin{figure}
\begin{center}
	\scalebox{0.6}{
\begin{tikzpicture}[set/.style={fill=cyan,fill opacity=0.1}]

 \draw[fill=yellow,fill opacity=0.05,rotate =0, draw=none] (0,1.7) ellipse (6cm and 3.2cm);
 \draw[fill=cyan,fill opacity=0.1,rotate =0, draw=none] (0,1.7) ellipse (6cm and 3.2cm);
  \node at (0,-0.7){
  \begin{minipage}{0.4 \textwidth}
  \centering
  analytical states spaces, \\ uncountable labels
\end{minipage}
};
  \draw[fill=purple,fill opacity=0.15,rotate =25, draw= none] (2,1) ellipse (3.2cm and 1.6cm);
  \node at (-2.8,3){\begin{minipage}{0.2 \textwidth} \centering labelwise continuous transition function \end{minipage}};
   \node at (2.8,3){\begin{minipage}{0.2 \textwidth} \centering Feller-continuous \\ kernels  \end{minipage}};
  \draw[fill=orange,fill opacity=0.15,rotate =-25, draw=none] (-2,1) ellipse (3.2cm and 1.6cm);
  %\node[draw,circle,minimum size =2.8cm,fill=red!50, fill opacity=0.5] (circle1) at (0,-0.5){};
  %\node at (0,-0.5){\begin{minipage}{0.2 \textwidth} \centering discrete states \\ countable labels  \end{minipage}};
  \node at (-1.5,1.8){\begin{minipage}{0.12 \textwidth} \centering
\begin{align*}
  \sim_{\statet} &= \sim_{\event}
\end{align*}
  \end{minipage}};
  \node at (0,4){$\bullet$};
  \node at (0,4.3){
  $\lmponelambda$
};
 \node (A) at (2.5,1.7){$\bullet$};
  \node at (2.5,1.35){
  $\lmponelambdac$
  };
  % \node (B) at (4,1.2){$\bullet$};
 % \node at (4.5,1.2){
 % $\lmponelambdac$
 % };
 % \draw[->, bend left] (A) to  node[midway, above] {$\sim_{event}$  } (B);
 % \draw[->, bend right] (A) to  node[midway, below] {$ \sim_{state}$} (B);
 %  \node at (3.5,1.2){= };
%\node at (0,4.3){$\bullet$};
%  \node at (0,4.9){\begin{minipage}{0.2 \textwidth} \centering
% From~\cite{clerc2019expressiveness} \\ $\sim_{state} \, \neq \, \sim_{event}$
 % \end{minipage} }; 
  
   % \node at (0,3){\begin{minipage}{0.15 \textwidth} \centering
%\begin{align*}
 % \sim_{event}  \\ &= \sim_{logic}
%\end{align*}
 % \end{minipage}};

  \end{tikzpicture}}
\end{center}
\caption{LMP-equivalences for Analytical States Spaces and Uncountable Labels}\label{fig:comparaison_equivalences_2}
\end{figure}
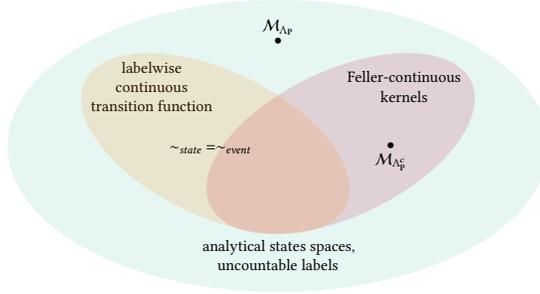

Our goal now is to introduce a different continuity requirement on LMPs, that 
both holds for $\lmponelambdac$ and enforces coincidence of event and state 
bisimilarities.
\longv{\subsubsection{Definition of Feller-Continuous Kernels}}
We first introduce a notion of convergence for measures on a standard Borel space, 
which is in fact standard~\cite{swart2013markov, parthasarathy2005probability}. 
Recall that for any program $M$, $\sem M$ is a measure on the standard Borel space 
$\valset$.
\begin{definition}
  Let $X$ be a standard Borel space, equipped with its Borel $\sigma$-algebra.
    Let $(\mu_n)_{n \in \NN}$ be a sequence of measures over $X$. We say that 
    $(\mu_n)_{n \in \NN}$ \emph{converges weakly} towards $\mu$ when
    for every bounded and continuous $f: X \rightarrow \RR $ it holds that
    $$\lim_{n \rightarrow \infty} \int_X f . d\mu_n = \int_X f . d\mu .$$
\end{definition}

%The notion of weak convergence corresponds to a \emph{topology} (actually, the space of sub-probability measures of a Polish %space equipped with the weak convergence is also a Polish space).
The notion of weak convergence corresponds to a \emph{topology} (actually, the space of sub-probability measures of a standard Borel space equipped with the weak convergence forms a Polish space, and 
thus a standard Borel space).
Since we have equipped the space of measures with a topology, all this gives us 
a notion of continuity for kernels, called \emph{Feller continuity} in the 
literature (see, e.g.,~\cite{swart2013markov}).
\begin{definition}[Feller-Continuous Kernels]\label{def:feller_lmp}
    Let $X,Y$ be two standard Borel spaces.
We say that a sub-probability kernel $k: X \kernel Y $ is \emph{Feller 
continuous} when the function $f_k: X \rightarrow \distrs Y$ is continous when 
equipping $\distrs Y$ with the weak topology. 
\end{definition}

 \shortv{It can be shown that Feller continuity is preserved by composition 
 (see~\cite{LV})}. \longv{Being a Feller kernel is a well-behaved notion: in 
 particular, we can observe that whenever a LMP $\lmpone$ is actually a
  (non-probabilistic) LTS, its transition kernel is Feller continuous. It is a 
  consequence of the following Lemma~\ref{notation:deterministic_generation}.
\begin{lemma}\label{notation:deterministic_generation}
For every measurable function $f: \RR^k \rightarrow \RR^p$, we denote by $\overline f$ the associated deterministic kernel $f: \RR^k \kernelp \RR^p$ defined as $\vec r \mapsto \dirac{f(\vec r)}$. Similarly, for any distribution $\distrone \in \distrs {\RR^n}$, we denote by $\overline \distrone$ the associated kernel $\RR^0 \kernelp \RR^n$. Whenever $f$ is continuous, $\overline f$ is Feller continous. Moreover $\overline \distrone$ is always Feller continuous.
\end{lemma}}
We now use this notion of Feller kernel to enforce additional continuity 
requirements on LMPs.
\begin{definition}
  Let $\lmpone = (\markovone, \actsone, \{h_a \mid a \in \actsone\})$ be a 
  measureably labelled Markov process. We say that $\lmpone$ is \emph{Feller 
  continuous} whenever $\markovone$ and $\actsone$ are two standard Borel spaces, and:
  \begin{itemize}
  \item for every $a \in \actsone$, the kernel $h_a: \markovone \times \Sigma_{\markovone} \rightarrow [0,1]$ is Feller continuous;
    \item for every $s \in \markovone$ the kernel $(a \in \actsone) \times (A \in \Sigma_{\markovone}) \mapsto h_a(s,A)$ is Feller continuous.
    \end{itemize}
\end{definition}
Please notice that it is enough for a Feller LMP  to be continuous on 
$\actsone$ 
and $\markovone$ \emph{separately}, and that we do not ask it to be jointly 
continuous on $\actsone \times \markovone$, which is in general a stronger 
condition~\cite{piotrowski1985separate}.
Being a Feller continuous LMP is a well-behaved notion: in particular, we can 
observe that whenever a LMP $\lmpone$ is actually a (non-probabilistic) LTS, 
its transition kernel is Feller. The same holds when actions and states are 
both countable. 
%\shortv{The proof, as well as other well-behaved properties of Feller kernel, can be found in the long version.}\longv{
Our aim is to show that this notion of Feller continuity can indeed be applied to the applicative LMP for $\lambdacont$, but not the one for $\lambdap$.

\begin{remark}\label{remark:lmplambda_not_feller_continuous}
  Let us now discuss why $\lmponelambda$ is \emph{not} Feller continuous: we 
  consider the action $\evalact$, and we show that the kernel $h_{\evalact}: 
  \markovone_{\Lambda_{\probb}} \times \Sigma_{\markovone_{\Lambda_{\probb}}} 
  \rightarrow [0,1]$ is not Feller continuous.  We consider the sequence of 
  terms defined as:
  $\termone_n = \op{(\makereal{r_n})} $, where $f: \RR \to \RR$ is the measureable function defined as $f(r) = 1$ if $r<0$, and $0$ otherwise, and $(r_n)_{n \in \NN}$ is a sequence of strictly negative real numbers tending towards $0$. We see that $h_{\evalact}(\termone_n,\typereal) = \dirac{\makereal 1}$ for every $n \in \NN$, but $h_{\evalact}( \op{(\makereal{0})},\typereal) = \dirac{\makereal 0}$. Since it is not the case that the stationary sequence $(\dirac{\makereal 1})_{n \in \NN}$ converges weakly towards $\dirac{\makereal 0}$--with respect to the Borel topology on $\markovone_{\Lambda_{\probb}}$--we can conclude that $\lmponelambda$ is not Feller continuous.
  \end{remark}
Remark~\ref{remark:lmplambda_not_feller_continuous} indeed shows that 
$\lmponelambda$ is not Feller continuous. Note that our
counterexample is based on the existence of \emph{non-continuous} 
primitive real functions in $\lambdap$. We will deal with the Feller continuity 
of $\lmponelambdac$ in Section~\ref{sect:completeness} below.

\longv{
\subsubsection{Technical Lemmas on Feller Continuous Kernels}

%are able to transform any weakly convergent sequences of (sub)-distributions in a sequence of proper distributions:
In this section, we show structural properties of Feller kernels, that we will need later in the proofs: stability of Feller constraints under countable coproducts, compositionnality of Feller kernels, and caracterisation of weak convergence of sub-probability measures as weak convergence for the associated proper probability measures. Before stating formally those lemmas, let us first give some elments about their relevance: we are interested in
%The first family of structural property we are interested in is stability under countable coproducts:
stability under the infinite coproduct construction because all polish space that we use for programs--$\valset,\terms...$-- are built as infinite coproduct; we are interested in composing Feller kernels because it allow us to see Feller kernels as a sub-category of Markov kernels; and we are interested in reducing the weak convergence decision problem from sub-probability measure to probability measure because, while we really need sub-probability kernels model in our semantics the probability of non-termination, it allows us to use result on proper probability kernels from the literature.

We show stability by countable coproduct first for weak convergence of measures-- in Lemma~\ref{lemma:contable_union_weak_convergence}, and for the class of Feller kernels-- in Lemma~\ref{lemma:countable_sum_kernels}.
\begin{notation}
Let $(X_n)_{n \in \NN}$ be a countable family of (disjoint) Polish spaces, and $X:= \cup_{n \in \NN} X_n$ endowed with the disjoint union topology. For $\distrone \in \distrs X$, we define $\distrone^{X_n} \in \distrs X_n$ as the restriction of $\distrone$ to $X_n$.
  \end{notation}
\begin{lemma}\label{lemma:contable_union_weak_convergence}
  Let $(X_n)_{n \in \NN}$ be a countable family of (disjoint) polish spaces, and $X:= \cup_{n \in \NN} X_n$ endowed with the disjoint union topology. Suppose $(\distrone_m)_{m \in \NN}$ a sequence of (sub)-distributions over $X$, and $\distrtwo \in \distrs X$. Suppose that the sequence $(\distrone_m)_{m \in \NN}$  converges weakly towards  $\distrtwo$.  Then each of the sequence $(\distrone^{X_n}_m)_{m \in \NN}$ converges weakly towards $\distrtwo^{X_n}$. 
  %if and only if each of the sequence $(\distrone^{X_n}_m)_{m \in \NN}$ converges weakly towards $\distrtwo^{X_n}$.
  \end{lemma}
  \begin{proof}
Observe that for every $n$, for every bounded continuous function $f: X_n \rightarrow \RR$, for every measure $\mu$ on $X$, $\int_{X_n} f \cdot d \mu^{X_n} = \int_{X} i_n(f) \cdot d\mu$, where $i_n(f) : x \in X  \mapsto f(x) \text{ if } x \in X_n, \, 0 \text{ otherwise}$. Moreover, whenever $f$ is bounded continuous, also $i_n(f)$ is bounded continuous for every $n \in \NN$. From there, we can conclude. 
\end{proof}
\begin{lemma}\label{lemma:countable_sum_kernels}
  Let $(X_n)_{n \in \NN}$ be a countable family of (disjoint) Polish spaces, and $Z$ a Polish space. Suppose a family $f_n:X_n \kernel Z$. We define $\bigsqcup_{n \in \NN} f_n: \bigcup_{n \in \NN}X  \kernel Z$, where $\bigcup_{n \in \NN}X$ is the Polish space obtained as the disjoint union of the $X_n$--with the disjoint union topology:
$$\bigsqcup_{n \in \NN} f_n(s) = f_{n_0}(s) \text{ with }X_{n_0}\text{ such that }s \in X_{n_0}. $$
  Then, if all the $f_n$ are Feller, also $ \bigsqcup_{n \in \NN} f_n$ is Feller. 
\end{lemma}
\begin{proof}
  Let $(\distrone_m)_{m \in \NN}$ a sequence of (sub)-distributions over $\bigsqcup_{n \in \NN} X_n $ that converges weakly towards some distribution $\distrtwo$.
  Observe that for every $\distrone$, $(\bigsqcup_{n \in \NN} f_n) (\distrone) = \sum_{n \in \NN}f_n(\distrone^{X_n})$, with $\distrone^{X_n}$ represents the restriction of $\distrone$ to $X_n$. By Lemma~\ref{lemma:contable_union_weak_convergence}, we can see that for every $n \in \NN$, the sequence $(\distrone_m^{X_n})_{m \in \NN}$ converge weakly towards $\distrtwo^{X_n}$. Since all the $f_m$ are Feller, we know that for every $n$, the sequence $(f_n(\distrone_m^{X_n}))_{m \in \NN}$ converge weakly towards $f_n(\distrtwo^{X_n})$. As a consequence $(\sum_{n \in \NN}f_n(\distrone_m^{X_n}))_{m \in \NN})$ converge weakly towards   $(\sum_{n \in \NN}f_n(\distrtwo^{X_n}))_{m \in \NN})$, which concludes the proof.
  \end{proof}

We now show how we can reduce the question of weak convergence for sub-probability measures to weak convergence for (proper) probability measures. Recall the notation $\mu_{\bot}$ introduced in Notation~\ref{notation:distrib_bot} to talk about the proper probability measure associated to a sub-probability measure.
%and Lemma~\ref{lemma:aux_subdistrs}, we can reduce the question of weak convergence for sub-probability measures to weak convergence for (proper) probability measures.
 \begin{lemma}\label{lemma:aux_weakly}
Let $(\mu_n)_{n \in \NN}$ be a sequence of sub-distributions that converges weakly towards $\mu$. Then the sequence $((\mu_n)_\bot)_{n \in \NN}$ converges weakly towards $\mu_{\bot}$.
    \end{lemma}
    \begin{proof}
Recall that by definition, for every $\distrtwo \in \distrs X$, $\distrtwo_\bot  = \distrtwo + (1 - \int_X d\distrtwo)\cdot \dirac{\bot}$. By hypothesis we know that $(\mu_n)$ converge weakly towards $\mu$, thus its is 
enough to show that for every $f: Z_{\bot} \rightarrow \RR$ bounded continuous, $\int_{Z_\bot} f \cdot d((1 - \int_X d\mu_n) \cdot \dirac \bot) \rightarrow \int_{Z_\bot} f \cdot d((1 - \int_X d\mu) \cdot \dirac \bot)$. Observe that:
    \begin{align*}
\int_{Z_\bot} f \cdot d((1 - \int_X d \mu_n) \cdot \dirac \bot) &= f(\bot) \cdot (1 - \int_X d\mu_n) \\
& \rightarrow  f(\bot) \cdot (1 - \int_X d\mu) \text{ again by hypothesis on } (\mu)_{n \in \NN} \\
& =\int_{Z_\bot} f \cdot d((1 - \int_X d\mu) \cdot \dirac \bot).
\end{align*}
\end{proof}
%Compositionnality of Feller continuous probability kernels is a known result in the literature.    
\begin{lemma}\label{lemma:feller_compositional}
Feller continuous probability kernels are compositional.
  \end{lemma}
  \begin{proof}
A composition result for Feller kernel is shown in~\cite{swart2013markov}, as a corollary of Proposition 3.5 there, but with the additional constraints that $X,Y,Z$ should be \emph{compact}--which is not adapted for our purposes. We show it here for general metric spaces.%--in the best of our knowledge, that's the first proof of composition for Feller kernels in this general setting.
We divide the proof in two steps: in a first step, we show the result for proper probability kernels, and in a second step we extend it to sub-probability kernels using  Lemma~\ref{lemma:aux_weakly}.

Let $f:X \kernel Y$ and $g:Y \kernel Z$ two (proper) probability kernels. Let $(x_n)_{n \in \NN}$ be a sequence in $X$ that converge towards $x$. The well-known \emph{portmanteau theorem} for weak convergence--see for instance Theorem 3.1 in~\cite{Ethier-Kurz}--tells us that a sequence of probability measures $(\mu_n)_{n \in \NN}$ over an arbitrary metric space $S$ converges weakly towards $\mu$ when for every open set $O \subseteq S$, $\liminf_{n\rightarrow \infty} \mu_n(O) \geq \mu(O)$. It means that we can rewrite our goal as proving:
\begin{equation}\label{eq:goal_compositionality}
\forall O \text{ open in }Z, \quad \liminf_{n\rightarrow \infty} (g \circ f)(x_n,O) \geq (g \circ f) (x, O)
\end{equation}
Let $O$ be an open set in $Z$.
Recall that by definition of kernel composition, for every $a \in X$, $(g \circ f)(a,O) = \int_{y} g(y,O) \cdot f(a,dy)$. By hypothesis $f$ is Feller continuous, thus the sequence of measure $f(x_n,\cdot)$ converge weakly towards $f(x,\cdot)$. Now, we use Proposition 1.4.18 in~\cite{hernandez2012markov} that tells us that whenever a
%bounded sequence of finite measures
sequence of probability measure $(\mu_n)_{n \in \NN}$ over an arbitrary metric space $S$ converges weakly towards $\mu$, then for every function $\alpha: S \rightarrow \RR_{\geq 0}$ \emph{lower semi-continuous (l.s.c.)}, then $\liminf_{n \rightarrow \infty}\int \alpha \cdot d\mu_n \geq \int \alpha \cdot d\mu$. %Observe that any sequence of sub-probability measures is a bounded sequence of finite measures, and
Recall that l.s.c. means that for every $a \in S$, $\liminf_{a_n \rightarrow a} \alpha(a_n) \geq \alpha(a)$. We now consider the function $\alpha_{O}:y \in Y \mapsto g(y,O) \in \RR_{+}$: we can see using the hypothesis on Feller continuity for $g$ that $\alpha_O$ is l.s.c. Indeed, let $(y_n)_{n \in \NN}$ a sequence in $Y$ that converges towards $y$: $\liminf_{n \rightarrow \infty} \alpha_O(y_n,O) =  \liminf_{n \rightarrow \infty} g(y_n,O) \geq g(y,O) = \alpha_O(y,O)$, by the portmanteau theorem and since by hypothesis on $g$ the sequence of measures $g(y_n,\cdot)$ converge weakly towards $g(y,\cdot)$. As a consequence, we can apply  Proposition 1.4.18 in~\cite{hernandez2012markov} to $\alpha_O$, and we obtain:
$$\liminf_{n \rightarrow \infty}\int_y \alpha_O(y) \cdot f(x_n,dy) \geq \int_y \alpha_O(y) \cdot f(x,dy),$$
which is exactly~\eqref{eq:goal_compositionality}, thus we can conclude.

    Our goal is now to generalize the result to the case of sub-probability kernels, using Lemma~\ref{lemma:aux_weakly}. Let $k: X \kernel Y$ and $h: Y \kernel Z$ two sub-probability kernels. We can transform them into proper probability kernel, by considering $k_{\bot}: X_{\bot} \kernel Y_{\bot}$, defined as $k_\bot(\bot) = \dirac{\bot}$, and for $x \in X$, $k_{\bot}(x) = (k(x))_{\bot}$, and similarly for $h$. The first step is to check that $h_{\bot} \circ k_{\bot} = (h \circ k)_{\bot}$: it is immediate. The second step consists in proving that for every kernel $g: X \kernel Z$, $g$ is Feller if and only of $g_{\bot}$ is Feller:
    \begin{itemize}
    \item Suppose that $g$ is Feller. Let $(x_n \rightarrow_{n \rightarrow \infty} x)$ a convergence sequence in $X_{\bot}$. Observe that there are two possible cases:
    \begin{itemize}
    \item either $x = \bot$, and there exists some $N$, such that for every $n \geq N$, $x_n = \bot$: in that case, since the sequence is stationnary we obtain immediately the result;
    \item or $x \in X$, and then there exists some $N$, such that for every $n \geq N$, $x_n \in X$: in that case for every $n \geq N$,  $g_{\bot}(x_n) = g(x_n) + (1 - \int_X g(x_n)) \cdot \dirac \bot = (g(x_n))_{\bot}$. By hypothesis on $g$, we know that $g(x_n) \rightarrow g(x)$ for the weak topology. From there, we can apply Lemma~\ref{lemma:aux_weakly}, and we get that $(g(x_n))_{\bot} \rightarrow (g(x))_{\bot}$, thus we can conclude. 

    \end{itemize}
    \item Suppose now that $g_{\bot}$ is Feller continuous; our goal is to show that also $g$ is Feller continuous. Let $(x_n \rightarrow_{n \rightarrow \infty} x)$ a convergence sequence in $X$, and $f : X \rightarrow \RR$ a continuous bounded function. We can see that for every $y \in X$, $\int_{Z} f \cdot d(g(y)) = \int_{Z_\bot} f'\cdot d(g_{\bot}(y))$, where $f' : X_{\bot} \rightarrow \RR$ is defined by $f'(\bot) = 0$, and $f'(z) = f(z)$ for $z \in X$. We see immediately that $f'$ is bounded, and by definition of the coproduct of Polish space, it is also continuous. So we can apply the Feller hypothesis on $g_{\bot}$, and we obtain the result.   
    \end{itemize}
    By combining these two steps with composition of Feller-continuous probability kernels, we obtain that also Feller continuous sub-probability kernels are compositionnal.
    \end{proof}
%  \begin{remark}
%We conjuncture that it is also the case that Feller continuous sub-probability 
%kernels are compositionnal, but we don't need to prove it here. \remarque{Would 
%be better to state it for all sub-probability kernels.}
 %   \end{remark}

}

\subsection{State and Event Bisimilarity Coincide for Feller-Continuous LMPs}

\begin{notation}
In this section, we fix a Feller continuous LMP $\lmpone$ with $\actsone$ as 
set of labels, and we fix $\actsone^{\QQ}$ a countably dense subset of 
$\actsone$ (such a dense subset always exists, since  Polish spaces are 
separable). We write $\lmpone^{\QQ}$ for the LMP extracted from $\lmpone$ by 
retaining only the actions in $\actsone^{\QQ}$. Observe that in the particular 
case of the LMP $ \lmponelambdac$ --- i.e. the applicative LMP for the language 
$\lambdacont$ where all primitives are continuous --- we can for instance 
define $ \lmponelambdac^{\QQ}$ as in Definition~\ref{def:lmplambdap}, but 
restricting to values built from rationals.
\end{notation}

\subsubsection{A Key Result from Optimal Transport}
 The proofs in this section will use in a crucial way a result coming from the literature on \emph{optimal transport} (see~\cite{villani2008optimal} for an introduction).
 %Thus, before going into the proofs, we present very briefly this field, and the connection with our setting.
A coupling $\omega \in \Omega(\distrone,\distrtwo)$ is said to be \emph{optimal} with respect to some measurable cost function $c:X \times Y \rightarrow \RR$ whenever it maximizes the quantity $\int c\cdot d\omega$. The field of optimal transport consists broadly in the study of such optimal couplings. For a Borel relation $\relone$, we can express the relator $\relatortwo$ from Definition~\ref{def:relator_two}, by associating an appropriate cost function $c_\relone$ to the relation $\relone$: then it holds that $\distrone \relatortwo \relone \distrtwo$ if and only if the cost of the optimal transport from $\distrone$ to $\distrtwo$ with respect to this cost function is $0$. We will make the construction of $c_\relone$ more precise in the proof of Corollary~\ref{corollary:weak_conv_couplings}.
This caracterisation of $\relatortwo$ in terms of optimal transport gives us a way to import results from this field. In this section, the key result we use comes from~\cite{villani2008optimal}, and is stated in Theorem~\ref{theorem:villani} below. It can be read as the stability of the operator $\relatortwo$ from Definition~\ref{def:relator_two} under weak convergence of measures.

  \begin{theorem}[from~\cite{villani2008optimal} ]\label{theorem:villani}
Let $X$ and $Y$ be
standard Borel spaces, and let $c: X \times Y \rightarrow \RR$  be a real-valued bounded continuous cost
function.
%Let (c k ) k∈N be a sequence of continuous cost functions converging uniformly to c on X × Y.
Let $(\mu_k)_{k\in \NN}$, and $(\nu_k)_{k\in \NN}$
be sequences of probability measures on $X$ and $Y$ respectively. Assume
that $\mu_k$ converges to $\mu$ (resp. $\nu_k$ converges to $\nu$) weakly. For each $k$, let
$\pi_k\in \Omega(\mu_k,\nu_k)$ be an optimal coupling with respect to $c$.
%If
%∀k ∈ N,c k dπ k < +∞,
Then, there exists a coupling $\pi \in \Omega(\mu,\nu)$ such that, up to extraction of a subsequence, the $\pi_k$ converges weakly to $\pi$, and 
%, up to extraction of a subsequence, $\pi_k$ converges weakly to some
 moreover $\pi$ is an optimal coupling.
    \end{theorem}
\begin{corollary}\label{corollary:weak_conv_couplings}
Let $R$ be a relation on standard Borel spaces $X,Y$ such that $R$ is a closed 
set. Let $(\mu_n)_{n \in \NN}$, and $(\nu_n)_{n \in \NN}$ be two sequences of 
sub-distributions such that $(\mu_n)_{n \in \NN}$ and $(\nu_n)_{n \in \NN}$ 
converge weakly towards $\mu$ and $\nu$ respectively, and moreover for every 
$n$, it holds that $\mu_n (\relatortwo R) \nu_n$. Then $\mu 
(\relatortwo R) \nu$.
    \end{corollary}
     \begin{proof}
    %First, we transform all sub-distributions on $X$, $Y$ into the corrsponding distributions on the polish spaces $X \sqcup \{\bot\}$, $Y \sqcup \{\bot\}$, where $\sqcup$ is the disjoint union operation on polish spaces. If $\eta$ is a sub-distribution on $Z$, we note $\eta_\bot$ the corresponding distribution on $Z \sqcup \{\bot\}$.
    \longv{Using Lemma~\ref{lemma:aux_subdistrs} and 
    Lemma~\ref{lemma:aux_weakly}, we see that we}\shortv{We} can suppose 
    without loss of generality\shortv{---see~\cite{LV}---}\longv{ }that the $\mu_n, \nu_n \ldots$ are proper distributions.
    We define a bounded continuous function $c_R : X \times Y \rightarrow \RR_{\geq 0}$ such that $c_R$ is zero on $R$, and strictly positive everywhere else. We take
    $c_R(z) = \inf \{1\} \cup \{\norm{z-w}\mid w \in R \}_{X \times Y} .$
    Observe that this function is continuous \emph{because of} the closeness of 
    $R$. \shortv{From there---see again~\cite{LV}---we are able to apply 
    Theorem~\ref{theorem:villani} to $c_R$, and conclude using the fact that 
    the couplings of weight $0$ for $c_R$ are exactly the couplings 
    compatible with $\relone$.}
    \longv{
    Since $\mu_n \relatortwo R \nu_n$, we know that there exists a coupling $\pi_n$ of zero weight, i.e. such that $\pi_n(X \times Y \setminus R) = 0$ (observe that $R$, $(X \times Y \setminus R)$,\ldots are measureable, since $R$ is closed). From there, and since $c$ is positive everywhere, zero on $R$ and bounded by $1$, we can see that:
\begin{equation}\label{eq:pi_weak_conv}
  0 \leq \int_{X \times Y}c. d\pi_n \leq \pi_n(X \setminus R) = 0
\end{equation}
  From there, we can apply Theorem~\ref{theorem:villani}, and we see that there exists a coupling $\pi \in \Pi(\mu,\nu)$ such that, up to extraction of a subsequence, the $\pi_k$ converges weakly to $\pi$, and 
%, up to extraction of a subsequence, $\pi_k$ converges weakly to some
  moreover $\pi$ is an optimal coupling.

  In order to conclude, it is enough to show that $\pi(X \setminus R) = 0$. To do that, we combine Equation~\eqref{eq:pi_weak_conv} with the fact that the $(\pi_n)_{n \in \NN}$ converge weakly towards $\pi$, and we obtain:
  $$0=\int_{X \times Y}c. d\pi_n = \int_{X \times Y} c.d\pi =  \int_{X \times Y \setminus R} c.d\pi.$$
  Since $c$ is non-negative, by Corollary 2.5.4 from~\cite{bogachev2007measure1}, it holds that $c_{\mid {(X \times Y \setminus R)}}$ is $0$ $\pi$-almost everywhere.
Since $c$ is strictly positive on $(X\times Y) \setminus R$, we can conclude that $\pi((X \times Y) \setminus R) = 0$.
%--see for instance Proposition 4.15 in~\cite{notes_integration}
}
%using Lemma~\ref{lemma:aux_strict_positive}.  
    \end{proof}

\subsubsection{Application to Feller-Continuous LMPs}
%To show the coincidence between state and event bisimilarity, we% are going to 
%start with a LMP whose labels are both countable and a \emph{dense 
%approximation} of the labels in $\lmpone$. Recall that such an LMP always 
%exists for Feller continuous LMPs, because we ask for the labels set to be a 
%Polish space, thus separable. 

We start with studying the bisimilarity landscape for the LMP $\lmpone^{\QQ}$. First, since $\actsone^{\QQ}$ is countable, we see that event and state 
bisimilarity coincide on $\lmpone^{\QQ}$ (by Theorem~\ref{prop:logical_caracterisation_state_c}). From now on, we will denote this relation $ \sim^{\lmpone^{\QQ}}$, in order to emphasize that there is no difference here between state and event bisimilarities. We now look at how we can relate $\sim^{\lmpone^{\QQ}}$ to the state and event bisimilarities on $\lmpone$: our end goal is to use Feller continuity constraints to show that all these relations coincide.
The first step is to show, using Corollary~\eqref{corollary:weak_conv_couplings} and Feller continuity of $\lmpone$, that the topological closure of $\sim^{\lmpone^{\QQ}}$ is a state bisimulation for $\lmpone$, thus is contained into $\sim^{\lmpone}_{state}$.
%The first step in our proof consists in showing that the graph of $\sim^{\lmpone^{\QQ}}$ is a closed set, in the topological sense.
%\subsubsection{$\sim^{\lmpone^{\QQ}}$ is topologically closed.}
%Our first step is to show that $\sim_{state}^{\lmpone^{\QQ}}$, seen as a subset of $\markovone \times \markovone$, is topologically closed. To do that, we show that its topological closure also is a bisimulation on $\lmpone^{\QQ}$:

\begin{proposition}\label{prop:closed_bisim}
 The topological closure of  $\sim^{\lmpone^{\QQ}}$ is a state bisimulation for $\lmpone$. 
    \end{proposition}
    \begin{proof}
We write $\overline{\sim^{\lmpone^{\QQ}}}$ for the topological closure of 
$\sim^{\lmpone^{\QQ}} $, i.e.,
$\overline {\sim^{\lmpone^{\QQ}}}$ is the set of all pairs $(x,y) \in 
\markovone \times \markovone$ such that there exists a 
$(\sim^{\lmpone^{\QQ}})$-sequences $ (x_n,y_n)_{n \in \NN}$ with $x= \lim_{n 
\rightarrow \infty} x_n$, and $y = \lim_{n \rightarrow \infty} y_n$. Let $s,t 
\in \markovone$ be such that $s \sim^{\lmpone^{\QQ}} t$, and $a \in \actsone$. 
Our goal from here is to show that  for every $a \in \actsone$, $h_a(s) \Gamma 
\relone h_a(t)$. Let us first unfold our hypothesis:
\begin{varenumerate}
\item Since $\actsone^{\QQ}$ is dense in $\actsone$, there exists a sequence $(a_m)_{m \in \NN}$ of $\actsone^{\QQ}$ elements such that $a_m \xrightarrow[n\rightarrow \infty]{} a$;
\item Since $\lmpone$ is a Feller continuous LMP, $x_n \xrightarrow[n 
\rightarrow \infty]{}{x}$, $y_n \xrightarrow[n \rightarrow \infty]{}{y}$, $a_m 
\xrightarrow[m \rightarrow \infty]{}{a}$ implies that: $\forall m, h_{a_m}(x_n) 
\xrightarrow[n \rightarrow \infty]{}{h_{a_m}(x)}$, $h_{a_m}(y_n) \xrightarrow[n 
\rightarrow \infty]{}{h_{a_m}(y)}$, and $\forall z \in \{x,y\}, h_{a_m}(z) 
\xrightarrow[m \rightarrow \infty]{}{h_a(z)}$, thanks to weak convergence;
\item For every $n \in \NN$, since $\relaop {x_n} {\sim^{\lmpone^{\QQ}}}{ 
y_n}$, it holds that for every $m \in \NN$, $\relaop{h_{a_m}(x_n)}{ \Gamma 
{(\sim^{\lmpone^{\QQ}})}} {h_{a_m}(y_n)}$, thus since $(\sim^{\lmpone^{\QQ}})$ 
is contained into its topological closure--and by monotonicity of the $\Gamma$ 
operator, we have that $\relaop{h_{a_m}(x_n)}{ \Gamma 
{(\overline{\sim^{\lmpone^{\QQ}}})}} {h_{a_m}(y_n)}$.
\end{varenumerate}
Recall that since $(\overline{\sim^{\lmpone^{\QQ}}})$ is 
closed, it is in particular Borel, thus we obtain by applying Theorem~\ref{thm:theta-equal-gamma} that ${ \Gamma 
{(\overline{\sim^{\lmpone^{\QQ}}})}} ={ \Theta
{(\overline{\sim^{\lmpone^{\QQ}}})}} $.
From there (and again since $(\overline{\sim^{\lmpone^{\QQ}}})$ is 
closed), we can apply Corrollary~\ref{corollary:weak_conv_couplings}: first we 
fix a $m$, and we make $n$ tends towards infinity, obtaining that 
$\relaop{h_{a_m}(x)}{ \Gamma {(\overline{\sim^{\lmpone^{\QQ}}})}} 
{h_{a_m}(y)}$. Now, since this is true for every $m$, we can apply once again 
Corrollary~\ref{corollary:weak_conv_couplings} by making $m$ tends toward 
infinity, and we obtain that $\relaop{h_{a}(x)}{ \Gamma 
{(\overline{\sim^{\lmpone^{\QQ}}})}} {h_{a}(y)}$, which ends the proof.
\end{proof}

We are now able to combine Proposition\eqref{prop:closed_bisim} with the logical caracterisation of state bisimilarity for LMPs with countable labels (Proposition~\eqref{prop:logical_caracterisation_state_c}) to show a logical caracterisation of state bisimilarity for Feller LMPs. Moreover, we obtain the aditional result that bisimilarity on Feller LMPs is a Borel relation, which is not necessarily the case for bisimilarities on arbitrary LMPs, as we will discuss in more details later in Section~\ref{sect:perspective}.

%On the other hand, since $\actsone^{\QQ} \subseteq \actsone$, it is immediate that the bisimilarity for $\lmpone^{\QQ}$ includes both state and the event bisimilarity in 
%$\lmpone$, i.e., $\sim_{state}^{\lmpone}\, \subseteq \,\sim_{event}^{\lmpone}\, 
%\subseteq \sim^{\lmpone^{\QQ}}$. \longv{In this section, our goal is to exploit
%the fact that $\lmpone^{\QQ}$ is a \emph{dense} approximation to show that
%$\sim^{\lmpone^{\QQ}}$ is a state bisimulation on  $\lmpone$, thus
 % $\sim^{\lmpone^{\QQ}}\, \subseteq \,\sim_{state}^{\lmpone}$. From there, we can see that all the inclusions above are equalities, and in particular the state and the event bisimilarity on $\lmpone$ coincide.}
%\shortv{In this section, we exploit the fact that $\lmpone^{\QQ}$ is a 
%\emph{dense} approximation to show that all inclusions above are equalities, 
%thus $\sim_{state}^{\lmpone}= \sim_{event}^{\lmpone}$.}
%\longv{To do that, we use a key tool from the literature, which is going to allows us to put together the definition of the probabilistic lifting of a relation, and the Feller continuity additional constraints on the kernels. This tool comes from the literature on \emph{optimal transport}:}

\begin{theorem}\label{prop:coincide_event_state_generic}
Let $\lmpone$ is a Feller continuous LMP.
Then it holds that 
  ${\stackrel{\lmpone}{\sim_{\statet}}} = {\stackrel{\lmpone}{\sim_{\logic}}} 
  = {\stackrel{\lmpone}{\sim_{\event}}} = {\stackrel{\lmpone}{\sim_{\test}}}$,
  and moreover the graph of these relations is a topological closed set, thus a Borel set.
%, then state and the event bisimilarity 
%on $\lmpone$ coincide.
 %Suppose that that  $\sim^{\lmpone^{\QQ}}$ is a state bisimulation on  $\lmpone$. Then it holds that $\sim_{state}^{\lmpone}$ and $\sim^{event}_{\lmpone}$ coincide.
%$\sim^{state}_{\lanone}\, \subseteq \,\sim^{event}_{\lanone}\, \subseteq \sim^{state,\QQ}_\lanone$.
\end{theorem}
\begin{proof}
Recall that we know already from the literature (see Theorem~\ref{prop:logical_caracterisation}) that ${\stackrel{\lmpone}{\sim_{\statet}}} \subseteq 
  {\stackrel{\lmpone}{\sim_{\event}}}
  = {\stackrel{\lmpone}{\sim_{\logic}}} 
  = {\stackrel{\lmpone}{\sim_{\test}}}$, thus it is enough to show that ${\stackrel{\lmpone}{\sim_{\statet}}} =
  {\stackrel{\lmpone}{\sim_{\event}}}$.
The main tool here is Proposition~\ref{prop:closed_bisim}: it tells us that the topological closure of $\sim^{\lmpone^{\QQ}}$ is a state bisimulation on $\lmpone$. Recall that by definition, $\sim_{\statet}^{\lmpone}$ is the largest state bisimulation on $\lmpone$; it means that:
\begin{equation}\label{eq:collapse_chain_1}
\overline{\sim^{\lmpone^{\QQ}}} \subseteq \sim_{\statet}^{\lmpone} \subseteq 
  {\stackrel{\lmpone}{\sim_{\event}}}.
\end{equation}
We are now going to build an inclusion chain in the reverse direction.
%Recall that we know by Proposition~\ref{proposition:state_incl_event_bisim}  $\sim_{state}^{\lmpone}\, \subseteq \,\sim_{event}^{\lmpone}$. Moreover, s
Since $\lmpone^{\QQ}$ is an approximation of $\lmpone$, we know that $\sim_{\event}^{\lmpone}\, \subseteq \sim^{\lmpone^{\QQ}}$ (it can be deduced for instance from Lemma~\ref{th:event_bisim_approx}), and since $\sim^{\lmpone^{\QQ}}$ is contained in its topological closure, it holds that $ \sim^{\lmpone^{\QQ}} \subseteq \overline{\sim^{\lmpone^{\QQ}}}$. Summing up, we have:
\begin{equation}\label{eq:collapse_chain_2}
\sim_{event}^{\lmpone} \,\subseteq\, \sim^{\lmpone^{\QQ}} \, \subseteq \, \overline{\sim^{\lmpone^{\QQ}}}.
\end{equation}
Combining~\eqref{eq:collapse_chain_1} and~\eqref{eq:collapse_chain_2}, we obtain that 
$\sim_{\statet}^{\lmpone}\, = \,\sim_{event}^{\lmpone} \, = \, \sim^{\lmpone^{\QQ}} \, = \, \overline{\sim^{\lmpone^{\QQ}}}, $
which ends the proof.
\end{proof}

\section{Full Abstraction for $\lambdacont$}\label{sect:completeness}
 In this section, we show that the operational semantics of 
 $\lambdacont$ is Feller continuous and prove, as a main consequence of 
 that, a powerful full abstraction theorem stating that all the 
 equivalences on $\lambdacont$ considered so far coincide. To achieve such a goal, 
 we rely on the characterization of terms as pre-terms (Section~\ref{sect:lambda}) 
 and show that the 
 evaluation of any pair of programs $\termone[\vec{\makereal{r}}]$ and 
 $\termone[\vec{\makereal{s}}]$, 
 (where $\vec r$ and $\vec s$ are vectors of real numbers) is essentially the 
 same, and only the reals \emph{inside} the output value change. Moreover, the  
 reals produced in output by $\termone[\vec{\makereal{r}}]$ depend continuously 
 from $\vec r$. To formally express this idea, we define a \emph{refinement} of 
 the evaluation semantics of Section~\ref{sect:lambda} 
 that keeps track separately of the evolution of the program 
 structure --- which is discrete --- and of the evolution of the continuous data. As 
 a consequence, such a refinement needs to be a semantics on \emph{pre-programs}; 
 we call this semantics the \emph{modular semantics} for pre-programs. 
\begin{notation}
To define the modular semantics, we use the following notation:
\begin{varitemize}
    \item We define $\mathcal E_k = \{(P,f) \mid P \text{ a pre-value with }p\text{ holes, }f:\RR^k \kernelp {\RR^p} \text{ a Feller continuous kernel}\}$, 
    for $k \in \NN$.
     \item If $P$ is a pre-term with $l_P$ holes and a free variable $x$ that 
     occurs $l_x$ times in $P$, and $V$ is a pre-value with $l_V$ holes, we 
     write $\subst P x V$ for the pre-term with $l_{\subst P x V} = l_P + 
     l_x\cdot l_V$ holes, obtained by doing the substitution and re-indexing 
     the holes canonically to obtain an increasing enumeration while reading 
     the terms. We write $\mathit{subst}_{[\cdot]}$ for the associated function 
     $\mathit{subst}^{[\cdot]}_{P,V} : \RR^{l_P + l_V} \rightarrow 
     \RR^{l_{\subst P x V}}$;
    \item
    Finally, recall that we write $\lambda$ for 
    the uniform distribution on $[0,1]$, and that given kernels $f$ and $g$ 
    and a function $\phi: \mathbb{R}^n \to \mathbb{R}^m$,
    we write $g \circ f$ for the composition of $f$ and $g$ (in the category of kernels) 
    and  $\overline{\phi}: \mathbb{R}^n \kernelp \mathbb{R}^m $ for the deterministic kernel 
    associated to $\phi$.
    \end{varitemize}
    \end{notation}
\begin{definition}\label{def:modular_approx_semantics}
    We define \emph{modular approximation semantics} as a function mapping a pre-term $M$ with $k$ holes to a non-decreasing family $(\sem{M}_{\star}^{(n)})_{n \in \NN}$ 
    of discrete sub-distributions over $\mathcal E_k$. The definition, by induction on $n$, is in Figure~\ref{figure:modular-sem}.
    %: we take $\sem{M}_{\star}^{(0)} = \emptyset$, the empty sub-distribution, and for $n >0$ we define:
    Since $(\sem{M}_{\star}^{(n)})_{n \in \NN}$ is a non-decreasing family on $\distrs {\mathcal E_k}$, it has a supremum, which is also a discrete distribution over $\mathcal E_k$, that we note $\sem{M}_\star$.
    \end{definition}
\begin{figure*}[t]
\begin{center}
        $$\sem{M}_{\star}^{(0)} = \emptyset; \qquad \sem{V}_{\star}^{(n+1)} = 
        \dirac{(V,\overline{id})} \text{ for } V \text{ pre-value}; \qquad \sem{\makereal{\op}([]^1, \ldots,[]^n)}_{\star}^{(n+1)} = \dirac{([],\overline{\op})};$$ 
      $$\sem{\sample}_{\star}^{(n+1)} = \dirac{([],\overline{\lambda})}; \qquad \sem{\unfold{\fold \valone}}_\star^{n+1} = \sem{\valone}_\star^{n}; $$
      $$\sem{ (\lambda x.P)V}_{\star}^{(n+1)} = \sum_{(W,f)} \sem{\subst P V x}^{(n)}_{\star}(W,f) \cdot \dirac{(W, f\circ \overline{\mathit{subst}^{\hole{\cdot}}_{P,V}})} ;$$
$$\sem{\casesum{\inject{\hat{\imath}}{\valone}}{\termone_i}}_\star^{n+1} = 
\sum_{(W,f)} \sem{\subst{\termone_{\hat{\imath}}}{\varone}{\valone}}_\star^{n} 
(W,f) \cdot \dirac{(W, f\circ 
\overline{\mathit{subst}^{\hole{\cdot}}_{\termone_{\hat{\imath}},\valone}})}; $$
       \begin{align*}
        \sem{\seq {N_1} {N_2}}_{\star}^{(n+1)} =& \\
        \sum_{(W_1,f_1)} \sem{N_1}_{\star}^{(n)}(W_1,f_1)&\cdot 
        \sum_{(W_2,f_2)} \sem{\subst {N_2}{x}{W_1}}_{\star}^{(n)} (W_2,f_2) 
        \cdot \dirac{(W_2,  f_2 \circ 
        \overline{\mathit{subst}^{\hole{}}_{N_2,W_1}}\circ (f_1 \times id))}.
        \end{align*}
\end{center}
\caption{Modular Semantics for $\Lambda_{\probb}^c$ }
	\label{figure:modular-sem}
\end{figure*}

Before proceeding any further, let us observe that we can characterize the semantics of programs using our modular semantics on pre-terms.

  \begin{lemma}\label{lemma:carac_modular_sem}
    For any pre-term $M$ with $p$ holes and $\vec r \in \RR^{p}$, we have:
\begin{align}
  & \forall n \in \NN, \sem{M[\vec{\makereal{r}}]}^{(n)} = \sum \sem{M}^{(n)}_\star(V,f) \cdot V[\makereal{f(\vec r)}] \label{eq:approx_modular} \\
  &\sem{M[\makereal{\vec r}]} = \sum \sem{M}_\star(V,f) \cdot V[\makereal{f(\vec r)}] 
  \label{eq:modular} 
  \end{align}
    \end{lemma}
    \shortv{
\begin{proof}
The proof, by induction on $n$, can be found in the long version \cite{LV}.
\end{proof}
}
\longv{
  \begin{proof}
    We first show~\eqref{eq:approx_modular} by induction on $n$. 
    \begin{itemize}
    \item if $n=0$, the result holds;
    \item we suppose that~\eqref{eq:approx_modular} is true for $n \in \NN$. We then proceed by case analysis on the syntax of $M$. We suppose that $M$ has $p$ holes.
      \begin{itemize}
      \item if $M = V$ with $V$ a pre-value, then  $\sem{M}_{\star}^{(n+1)} = \dirac{(V,\overline{id})}$ and for every $\vec r \in \RR^p$, $\sem {M[\vec r]}^{n+1}= \dirac{V[\vec r]}$. Since $V[\overline{id}(\vec r)] = V[\dirac {\vec r}] = \dirac{V[\vec r]}$, the result holds;
      \item if $M = \underline{F}([]^1, \ldots,[]^n)$ then  $\sem{M}_{\star}^{(n+1)}) = \dirac{(W,\overline{F})}$ with $W=[]$. Let $\vec r \in \RR^p$: it holds that $\sem {M[\vec r]}^{n+1} = F(\vec r)$. Since $W[\overline{F}(\vec r)] = \dirac{F(\vec r)}$, the result holds;
      \item if $M = \text{sample}$, then $p=0$, and $\sem{M}_{\star}^{(n+1)}) = \dirac{(W,f)}$ with $W = []$ and $f = (\star \mapsto{\lambda})$. We see that $\sem M = \lambda$. Moreover, $W[f]=\lambda$, thus the result holds;
      \item if $M = (\lambda x.P)W$, then $p = p_1+p_2$, with $P$ a pre-term with $p_1$ holes, and $V$ a pre-value with $p_2$ holes. Let $l$ be the number of holes for $\subst P x W$: the substitution implies uniquely a function $\phi:\RR^p \rightarrow \RR^l$.
        Let $\vec r = (\vec r_1,\vec r_2) \in \RR^{p_1} \times \RR^{p_2}$. Observe that:
        \begin{align*}
          \sem {M[\vec r]}^{n+1} &= \sem{\subst {P[\vec {r_1}]}{x}{W[\vec{r_2}]}}^{n}  \\
          &=\sem{(\subst P x W)[\phi(\vec r)]} \text{ by definition of } \phi \\
          &= \sum_{(V,f)} \sem{\subst P x W}^{(n)}_\star(V,f) \cdot V[f(\phi(\vec r))] \text{ by induction hypothesis} \\
          &= \sum_{(V,f)} \sem{\subst P x W}^{(n)}_\star(V,f) \cdot V[(f \circ \phi)(\vec r))] \\
          &= \sum_{(V,f)} \sem{M}^{(n+1)}_\star(V,f\circ \phi) \cdot V[(f \circ \phi)(\vec r))] \text{ by Definition~\ref{def:modular_approx_semantics}}\\
           &= \sum_{(V,g)} \sem{M}^{(n+1)}_\star(V,g) \cdot V[g(\vec r))] \text{ since all }g \in \support{\sem M^{n+1}} \text{ are of the form }f\circ \phi;
          \end{align*}
        so the result holds.
      \item if $M= (\letin x {N_1} {N_2})$, then let $p = p_1 + p_2$ with $N_1$ a pre-term with $p_1$ holes, and $N_2$ a pre-term with $p_2$ holes. Let $\vec r = (\vec r_1,\vec r_2) \in \RR^{p_1} \times \RR^{p_2}$. Observe that:
        \begin{align*}
          &\sem {M[\vec r]}^{n+1} = (L \in \valset \mapsto \sem {\subst {N_2[\vec r_2]}x {L}})^{n}\circ \sem{N_1[\vec r_1]}^{n} \\
          &= (L \in \valset \mapsto \sem {\subst {N_2[\vec r_2]}x {L}})^{n}\circ ( \sum_{(V,f)} \sem{N_1}_{\star}^{n}(V,f) \cdot V[f(\vec r_1)] )\text{ by induction hypothesis};\\
          &= \sum_{(V,f)} \sem{N_1}_{\star}^{n}(V,f) \cdot (L \in \valset \mapsto \sem {\subst {N_2[\vec r_2]}x {L}})^{n}\circ (  V[f(\vec r_1)] )\text{ by Lemma~\ref{lemma:addition_kernels}};\\
          &= \sum_{(V,f)} \sem{N_1}_{\star}^{n}(V,f) \cdot (\vec z  \mapsto \sem {\subst {(N_2[\vec r_2])}x {V[\vec z]}})^{n}\circ (  f(\vec r_1) ) \text{ by change of variable }\\
          &= \sum_{(V,f)} \sem{N_1}_{\star}^{n}(V,f) \cdot (\vec z  \mapsto \sem {(\subst {N_2}{x}{V})[\phi_{V}(z,\vec r_2)]}^{n}\circ (  f(\vec r_1) ) \text{ by definition of }\phi_V \text{in Def~\ref{def:modular_approx_semantics}}\\
          &= \sum_{(V,f)} \sem{N_1}_{\star}^{n}(V,f) \cdot (\vec z  \mapsto \left(
            \sum_{W,g} \sem{ \subst {N_2}{x}{V} }^{(n)}_\star(W,g) \cdot W[g(\phi_{V}(z,\vec r_2))]
            \right))\circ (  f(\vec r_1) ) \text{ by IH}\\
            &= \sum_{(V,f)} \sem{N_1}_{\star}^{n}(V,f) \cdot \sum_{(W,g)} \sem{ \subst {N_2}{x}{V} }^{(n)}_\star(W,g) \cdot (\vec z  \mapsto W[g(\phi_{V}(z,\vec r_2))])\circ (  f(\vec r_1) ) \text{ by Lemma~\ref{lemma:addition_kernels}}\\
            &= \sum_{(V,f)} \sem{N_1}_{\star}^{n}(V,f) \cdot \sum_{(W,g)} \sem{ \subst {N_2}{x}{V} }^{(n)}_\star(W,g) \cdot \left(((\vec z  \mapsto W[z])\circ h_{V,f,g}) \right(\vec r) \\
            & \qquad \qquad \qquad \text{ with } h_{V,f,g}:\RR^{p_1+p_2} \kernel \RR^{l_{W}}, h_{V,f,g}=g \circ \phi_V \circ (f,\overline{id}_{\RR^{p_2}}) \\
            &= \sum_{(V,f),(W,g)} \sem{M}_{\star}^{(n+1)}(W,h_{V,f,g})\cdot W[h_{V,f,g}(\vec r)] \text{ by Definition~\ref{def:modular_approx_semantics}}\\
             &= \sum_{(V',f')} \sem{M}_{\star}^{(n+1)}(V',f')\cdot V'[f(\vec r)] \text{ since all }(V',f') \in \support{\sem{M}_{\star}^{(n+1)}} \text{ are of the form }(W,h_{V,f,g}),
        \end{align*}
        and that concludes the proof.
        \end{itemize}
    \end{itemize}
    We now show~\eqref{eq:modular} using~\eqref{eq:approx_modular}:
    \begin{align*}
      \sem{M[\vec r]} &= \sup_{n \in \NN} \sem{M[\vec r]}^{(n)} \quad \text{ by definition of the operational semantics}\\
      &= \sup_{n \in \NN} \sum_{(V,f)} \sem{M}^{(n)}_\star(V,f) \cdot V[f(\vec r)] \text{ by }~\eqref{eq:approx_modular} \\
      & = \sum_{(V,f)} (\sup_{n \in \NN} \sem{M}^{(n)}_{\star}(V,f)) \cdot V[f(\vec r)] \text{ because } \forall (V,f) \text{ the sequence } (\sem{M}^{(n)}_{\star}(V,f))_{n \in \NN}\text{ is non-decreasing} \\
      &= \sum_{(V,f)} ( \sem{M}_{\star}(V,f)) \cdot V[f(\vec r)] \text{ by definition of } \sem{M}_\star.
      \end{align*}
    \end{proof}}
The notion of modular approximation semantics from Definition~\ref{def:modular_approx_semantics} has been designed as a mean to express that the semantics of a program $\termone[\vec{\makereal{r}}] \in \lambdacont$ depends continuously on $\vec r$. We are now able to show this formally.
 \begin{proposition}\label{theorem:weakly_conv_sem}
Let $(M_n)_{n \in \NN}$ be a sequence of terms that converges towards $M$. Then the sequence of distributions $(\sem{M_n})_{n \in \NN}$ converges weakly towards $\sem M$.
  \end{proposition}
   \begin{proof}
    Observe that the fact that  $(M_n)_{n \in \NN}$ is a sequence of terms that converges towards $M$ actually means that there exists a pre-term $N$ (with $k$ holes), a real-valued vector $\vec r \in \RR^k$, and a sequence $\vec {r_n} \in \RR^k$ such that $\vec{r_n} \rightarrow \vec r$ and
    $M = N[\vec{\makereal{r}}]$, $M_n = N[\vec{\makereal{r}}_n].$
    %\shortv{
   %The proof, that uses Lemma~\ref{lemma:carac_modular_sem} in order to switch from the semantics on programs to the modular semantics on pre-programs, and the compositionnality of Feller continuous probability kernels from Lemma~\ref{lemma:feller_compositional} can be found in the long version.}
    Consequently, to prove that $(\sem{M_n})_{n \in \NN}$ converges weakly towards $\sem M$, 
    it is enough to prove the following:
      let $g:\valset \rightarrow \RR$ be a bounded and continuous function
      %, and $M$ be the bound of $g$. Let 
      and $\epsilon >0$. Then, we need to show: $\exists N \in \NN, \forall n \geq N,  \lvert \int_\valset g . d{\sem{N[\vec{\makereal{r}}_n]}} - \int_\valset g . d{\sem {N[\vec{\makereal{r}}]}} \rvert \leq \epsilon.$
    Using the equivalence between semantics on programs and modular semantics on pre-programs given by Lemma~\ref{lemma:carac_modular_sem}, we see that:
    \begin{align*}  
    %A:= 
    &\left\lvert \int_\valset g . d{\sem{N[\vec{\makereal{r}}_n]}} - 
    \int_\valset g . d{\sem N[\vec{\makereal{r}}]} \right\rvert 
    \\
     \longv{
     &\qquad=  \left\lvert \int_\valset g . d{(\sum \sem{N}_\star(V,f) \cdot 
      V[\makereal{f(\vec {r_n})}] )} - \int_\valset g . d{(\sum 
      \sem{N}_\star(V,f) \cdot V[\makereal{f(\vec r)}]  )} \right\rvert 
      \\
      }
 &\qquad=  \left\lvert \sum_{(V,f) \in 
 \support{\sem{N}_\star}}(\sem{N}_\star(V,f) \cdot (\int_\valset g . d 
 V[\makereal{f(\vec {r_n})}] - \int_\valset g. d V[\makereal{f(\vec r)}]) 
 \right\rvert.
   \end{align*}
   \shortv{From there, we are able to conclude using the fact that all the $f:\RR^k \kernelp {\RR^p}$ are Feller continuous kernels. More details can be found in the long version 
   \cite{LV}.}
   \longv{
    Since $\sem{N}_\star$ is a \emph{discrete} sub-probability distribution, there exists a finite set $\mathcal E'_k \subseteq \mathcal E_k$ such that $\sem{N}_\star(\mathcal E_k \setminus \mathcal E'_k) \leq \epsilon/(2 M)$. From there, since $g$ is bounded by $M$, and both  $V[f(\vec {r_n})],  V[f(\vec {r})]$ are probability measure, we see that:
    \begin{align*}
A & \leq \lvert \sum_{(V,f) \in \mathcal E'_k}(\sem{N}_\star(V,f) \cdot (\int_\valset g . d V[f(\vec {r_n})] - \int_\valset g. d V[f(\vec r)]) \rvert + \frac \epsilon 2
    \end{align*}
    
   At that point, in order to conclude, it is enough to show that for every $(V,f)$ in the finite set $\mathcal E'_k$, it holds that  $\lvert \int_\valset g.d V[f(\vec {r_n})] - \int_\valset g. d V[f(\vec r)]) \rvert \rightarrow_{n \rightarrow \infty} 0$. Since by hypothesis $g$ is bounded continuous, it means that it is enough to show that the probability kernel $K_{V,f}:=\vec y \in \RR^k \mapsto V[f(\vec y)] \in \distrsp \valset$ is Feller continuous. First, observe that $K_{(V,f)}$ can be written as the composition of two probability kernels: $K_{(V,f)} = I_{V}\circ f$, with $I_V: \vec z \in \RR^p \mapsto \dirac{V[\vec z]}$ (with $p$ the number of holes in $V$). By hypothesis, the kernel $f$ is Feller continuous. Moreover, $I_V$ also is Feller continuous, since the function $z\in \RR^{p} \mapsto  V[\vec z] \in \valset$ is continuous. From there, we can conclude using Lemma~\ref{lemma:feller_compositional} that $K_{(V,f)}$, as the composition of two Feller continuous kernel, is also Feller continuous.}
    \end{proof}
Our goal now is to use Proposition~\ref{theorem:weakly_conv_sem} to show that $\lmponelambdac$ is a Feller LMP, in the sense of Definition~\ref{def:feller_lmp}. 
That amounts to prove Feller continuity of two families of kernels: $h_a: \markovone_{\lambdacont} \kernel \markovone_{\lambdacont}$, 
for $a \in \actsone_{\lambdacont}$; and $\widehat h^s = (a \in \actsone_{\lambdacont}) \times (A \in \Sigma_{\markovone_{\lambdacont}}) \mapsto h_a(s,A)$, for $s \in \markovone_{\lambdacont}$.

\hide{   Looking at the definition of state and event bisimilarity, we can show already the following lemma--its proof can be found in the long version--that allows us to connect the state and event bisimilarity on $\lmponelambdac$ and $\lmponelambdac'$ respectively.
   \begin{lemma}\label{lemma:aux_lmp_equivalent}
$\sim_{state}^{\lmponelambdac} = \sim_{state}^{\lmponelambdac'}$, and $\sim_{event}^{\lmponelambdac'} \subseteq \sim_{event}^{\lmponelambdac}$.  
     \end{lemma}
   \begin{proof}
     We split in three the proof of Lemma~\ref{lemma:aux_lmp_equivalent}:
  \begin{itemize}
  \item First, we show that any state bisimulation $\relone$ for $\lmponelambdac$ is also a state bisimulation for $\lmponelambdac'$. Let be $s,t \in \markovone_{\lambdacont}$ such that $s \relone t$, and $a \in \actsone_{\lambdacont}$.  We know by hypothesis that $h_a(s) \Gamma \relone h_a(t)$, and we want to show that $h_a'(s) \Gamma \relone h_a'(t)$. Observe that if $a$ is not of the form $(\stackrel{r}{=})$, $h_a'(u) = h_a(u)$ for $u \in \{s,t\}$, thus the result is immediate. If $a=(\stackrel{r}{=})$, observe that the only relevant case is when both $s$ and $t$ are values of type $\typereal$, i.e. $s= \makereal{{r_{1}}}$ and $t = \makereal{{r_2}}$. We can deduce from $s \relone t$ that $r_1 = r_2$, thus $s = t$, and we can conclude. 
  \item Secondly,  we show that any state bisimulation $\relone$ for $\lmponelambdac'$ is also a state bisimulation for $\lmponelambdac$. The reasonning is exactly as before, and uses the fact that whenever two $\typereal$ values  $\makereal{{r_{1}}}$ and $\makereal{r_2}$ are bisimilar in $\lmponelambdac'$, then $r_1 = r_2$.
    \item Finally, we show that any event bisimulation for $\lmponelambdac'$ is also an event bisimulation for $\lmponelambdac$. Let $\Lambda$ be an event bisimulation for $\lmponelambdac$, that is a sub-$\sigma$-algebra $\Lambda$ of $\Sigma$,
      such that $(\markovone_{\lambdacont}, \Lambda, \{h'_a \mid a \in \actsone_{\lambdacont}\})$ is a LMP. We want to show that also $(\markovone_{\lambdacont}, \Lambda, \{h_a \mid a \in \actsone_{\lambdacont}\})$, i.e. that for every $a \in \actsone$, the restriction of $h_a$ to the map $h_a:\markovone_{\lambdacont} \times \Lambda \rightarrow [0,1]$ is a sub-probability kernel. First, observe that if $a$ is not of the form $(\stackrel{r}{=})$, $h_a' = h_a$, thus the result is immediate. Then, let be $r \in \RR$, and let us show that the map  $h_{(\stackrel{r}{=})}:\markovone_{\lambdacont} \times \Lambda \rightarrow [0,1]$ is indeed a kernel. To do that, we need to show the two conditions below, that comes from Definition~\ref{def:probability_kernel}:
      \begin{itemize}
      \item For every $A \in \Lambda$, the map $s \in \markovone_{\lambdacont} \mapsto h_{(\stackrel{r}{=})}(s,A)$ is measureable: what we need to show is that for every $I$ in $\borels{[0,1]}$:
        $$\{ s \in \markovone_{\lambdacont} \mid h_{(\stackrel{r}{=})}(s,A) \in I \} \in \Lambda.$$
        Observe that, in practice, $h_{(\stackrel{r}{=})}(s,A) \in \{0,1\}$, thus it is sufficient (since $\Lambda$ is closed by finite union and complementation) to show the result for $I = \{1\}$.   We wan see that: if $r \in A$: $\{ s \in \markovone_{\lambdacont} \mid h_{(\stackrel{r}{=})}(s,A) =1 \} = \{r\} = \{ s \in \markovone_{\lambdacont} \mid h'_{(\stackrel{r}{=})}(s,A) =1 \}$, and from there we can conclude, using the fact that $\Lambda$ is an event bisimulation on $\lmponelambdac'$. If $r \not \in A$, $\{ s \in \markovone_{\lambdacont} \mid h_{(\stackrel{r}{=})}(s,A) =1 \} = \emptyset \in \Lambda$ since $\Lambda$ is a sigma-algebra. 
      \item For every $s \in \markovone_{\lambdacont}$, the map $A \in \Lambda \mapsto h_{(\stackrel{r}{=})}(s,A)$ is a sub-probability measure. We see that immediately from the fact that $\lmponelambdac$ is a LMP, and $\Lambda$ is a sub sigma-algebra of $\Sigma_{\markovone_{\lambdacont}}$.
        \end{itemize}
    \end{itemize}
  \end{proof}
}
\shortv{
Before doing these proofs, however, we need a technical lemma on how to combine Feller kernels on a measurable space build as a countable coproduct of standard Borel spaces. In fact, 
recall from Definition~\ref{definition:terms_as_polish_spaces} that we 
endowed the set of terms and values with a (countable) coproduct standard Borel space structure.
\begin{lemma}\label{lemma:countable_sum_kernels}
  Let $(X_n)_{n \in \NN}$ be a countable family of (disjoint) standard Borel spaces, and $Z$ 
  be a standard Borel space. Let $\{f_n:X_n \kernel Z\}_{n \in \NN}$ be a family of 
  kernels, and let us write  $\bigcup_{n \in \NN}X$ for the standard Borel space obtained by 
  endowing the disjoint union of the $X_n$s with the coproduct topology. 
  We define $\bigsqcup_{n \in \NN} f_n: \bigcup_{n \in \NN}X  \kernel Z$ as:
$$\bigsqcup_{n \in \NN} f_n(s) = f_{n_0}(s) \text{ with }{n_0 \in \NN}\text{ unique such that }s \in X_{n_0}. $$
  Then, if all the $f_n$ are Feller continuous, also $ \bigsqcup_{n \in \NN} f_n$ is Feller continuous. 
\end{lemma}
}
We are now ready to show that all relevant kernels for the LMP $\lmponelambdac$ are Feller continuous.  We do this below in Lemma~\ref{lemma:acts_kernels_feller_cont} and Lemma~\ref{lemma:states_kernels_feller_cont}.
\begin{lemma}\label{lemma:acts_kernels_feller_cont}
 For every $a \in \actsone$, the kernel $h_a: \markovone_{\lambdacont} \kernel \markovone_{\lambdacont}$ is Feller continuous.  
\end{lemma}
\begin{proof}
We present only the most interesting cases, the others can be found in the long version of this paper, and are obtained by small variations on the proofs below.
  \begin{varitemize}
    \item If $a = \evalact$, then the thesis directly follows from Proposition~\ref{theorem:weakly_conv_sem}.
    \longv{\item for $a = \typeone$ with $\typeone \in \types$:
Observe that $h_\typeone = \text{id}^{\valset_\typeone} \sqcup \text{id}^{\terms_\typeone} \sqcup 0$, using the notation of Lemma~\ref{lemma:countable_sum_kernels}. From Lemma~\ref{lemma:countable_sum_kernels}, we see that it is enough to show that both $\text{id}^{\valset_\typeone}: \valset_{\typeone} \kernel \markovone_{\lambdacont} $ and $\text{id}^{\terms_\typeone} \terms_{\typeone} \kernel \markovone_{\lambdacont} $ are Feller continuous. Since  for every $\distrone \in \distrs{\valset_\typeone}$, $\text{id}^{\terms_\typeone}(\distrone) = \distrone$, this is obviously the case.}
\item If $a = \valone \in \valset_{\typeone}$, then we observe that 
$h_{\valone} = \bigsqcup_{\typetwo} h^{\typeone}_\valone \sqcup 0$, where $h_\valone^{\typeone}: \valset_{\typeone \to \typetwo} \kernel \markovone_{\lambdacont}$. From Lemma~\ref{lemma:countable_sum_kernels}, we see that it is enough to show that for every $\typetwo$, and every $\valone \in \valset_\typeone$, the kernel $h_\valone^{\typeone}: \valset_{\typeone \to \typetwo} \kernel \markovone_{\lambdacont}$
  is Feller. Moreover, we see that $h_\valone^{\typeone}$ is deterministically generated  by the function $\widehat{h}_\valone^{\typeone \to \typetwo} = 
  \valtwo \in \valset_{\typeone \to \typetwo} \mapsto \valtwo \valone \in \markovone_{\lambdacont}$. Looking at the topologies on $\valset_{\typeone}$ and $\markovone_{\lambdacont}$, we see immediately that $\widehat{h}_\valone^{\typeone}$ is a continuous function.
  %Let $(\distrone_n)_{n \in \NN}$ be a sequence of (sub)-distributions over $\valset_{\typeone \to \typetwo}$ that converges weakly towards some distribution $\distrtwo$. Let be $f$ a bounded continuous function $\valset_{\typeone \to \typetwo}$
  %We see that $h_\valone(\distrone_n) = \sum_{}$
\longv{\item If $a = case(\hat{\imath})$, with $i \in I$. As above we can write $h_a$ as a coproduct kernel, i.e. we split $\markovone_{\lambdacont}$  as $\markovone_{\lambdacont} =\bigsqcup_{{(\typeone_j)}_{j \in I}} \{\hat{\imath}\}\times \valset_{\typeone_i} \sqcup Z$. Observe that this splitting is indeed consistent with Definition~\ref{def:lmplambdap}, in the following sense: if we start from the component $\valset_{\sumtype{j \in I}{\typeone_j}}$ from $\markovone_{\lambdacont}$, then it is itself defined--see Definition~\ref{definition:terms_as_polish_spaces} as the disjoint union of components of the form $\{\hat {\jmath}\} \times \valset_{\typeone_j}$.
  %the component that is we split into disjoint components that can be decomposed into the disjoint components for $\markovone_{\lambdacont}$ in Definition~\ref{def:lmplambdap}.
  From there, we can write:  $h_a = \bigsqcup_{{\mathcal F = (\typeone_j)}_{j \in I}} h_{\hat{\imath}}^{\mathcal F} \sqcup 0$, where $h_{\hat{\imath}}^{\mathcal F}: \{\hat {\imath}\} \times \valset_{\typeone_i} \kernel \markovone_{\lambdacont}$ is defined as  $h_{\hat{\imath}}^{\mathcal F}(\hat{\imath},\valone) = \dirac \valone$. From Lemma~\ref{lemma:countable_sum_kernels}, we see that it is enough to show that $h_{\hat{\imath}}^{\mathcal F}$ is Feller continuous. Then, we show easily that  $h_{\hat{\imath}}^{\mathcal F}$ is deterministically generated by a continuous function, and we can conclude.
  \item if $a = \unboxact$}
  %\item if $a = (\stackrel{r}{=})$ with $r \in \RR$. Let us split $\markovone_{\lambdacont}$  as $\markovone_{\lambdacont} = \valset_{\typereal} \sqcup Z$, and observe that $h_{\stackrel{r}{=}} = h_{\stackrel{r}{=}}^{\typereal} \sqcup 0$. As before, we apply Lemma~\ref{lemma:countable_sum_kernels}, and we see that it is enough to show that $h_{\stackrel{r}{=}}^{\typereal}$ is Feller continuous.  Let $(\mu_n)_{n \in \NN}$ be a sequence of measures over $\RR$ that converges weakly towards $\mu$. Let $f:\RR \rightarrow \RR$ bounded continuous. Observe that for every distribution over reals $\nu$,  $\int_{\RR} f d(h_{\stackrel{r}{=}}^{\typereal}(\nu)) = \int_{\RR} f \times g d\nu$, with $g(r') =  (1 - \max(1, \lvert r' - r \rvert )$. Since $f \times g:\RR \rightarrow \RR$ is also bounded continuous, we can conclude.
  \item If $a = \overset{\text{?}}{\leq} q$, then we first recall that we have a 
   comparison operator $op_{\leq} \in  \mathcal{C}_2$ such that 
   $op_{\leq}: \RR \times \RR \to [0,1]$, and  $op_{\leq}(x,y) = 1$ if and only if $x \leq y$. Since all primitive functions are continuous, $op_{\leq}$ is continuous too.  Let us split $\markovone_{\lambdacont}$  as $\markovone_{\lambdacont} = \valset_{\typereal} \sqcup Z$, and observe that $h_a = h_a^{\valset_\typereal} \sqcup 0$. As before, we apply Lemma~\ref{lemma:countable_sum_kernels}, and we see that it is enough to show that $h_a^{\valset_\typereal}$ is Feller continuous. Recall that $h_a^{\valset_\typereal}(\makereal r) = p_{\leq}(r,q) \cdot \dirac r$
  %\item if $a = (p)$ with $p: \RR \rightarrow[0,1] \in \mathcal{C}_1$. Let us split $\markovone_{\lambdacont}$  as $\markovone_{\lambdacont} = \valset_{\typereal} \sqcup Z$, and observe that $h_p = h_p^{\typereal} \sqcup 0$. As before, we apply Lemma~\ref{lemma:countable_sum_kernels}, and we see that it is enough to show that $h_p^{\valset_\typereal}$ is Feller continuous.
  Let $(\mu_n)_{n \in \NN}$ be a sequence of measures over $\RR$ that converges weakly towards $\mu$ and let $f:\RR \rightarrow \RR$ be a bounded continuous function. 
  Notice that for every distribution $\nu$ over reals, we have:
  \begin{equation}\label{eq:aux_fc}
  \int_{\RR} f d(h_{a}^{\valset_\typereal}(\nu)) = \int_{\RR} f \times op_{\leq}(\cdot,q) d\nu.
\end{equation}
Since $op_{\leq}$ is bounded and continuous, we deduce from weak convergence of $(\mu_n)_{n \in \NN}$ that $\int_{\RR} f \times op_{\leq}(\cdot,q) d\mu_n \rightarrow_{n \rightarrow \infty} \int_{\RR} f \times op_{\leq}(\cdot,q) d\mu$. By combining this with~\eqref{eq:aux_fc}, we conclude that $\int_{\RR} f d(h_{a}^{\valset_\typereal}(\mu_n)) \rightarrow_{n \rightarrow \infty}  \int_{\RR} f d(h_{a}^{\valset_\typereal}(\mu))$. Since it is true for every bounded continuous function $f$, it means that $h_a^{\valset_\typereal}(\mu_n)$ converges weakly towards $h_a^{\valset_\typereal}(\mu)$, and we can conclude from there.
    \end{varitemize}
\end{proof}

\begin{lemma}\label{lemma:states_kernels_feller_cont}
 For every $s \in \markovone_{\lambdacont}$, the kernel $\widehat{h}_s: \actsone_{\lambdacont} \kernel \markovone_{\lambdacont}$ is Feller continuous.
\end{lemma}

\begin{proof}
  We do the proof by case analysis on $s$. We present  the most interesting 
  cases only\shortv{ (the others can be found in the long version of this paper 
  \cite{LV})}.
    \begin{varitemize}
    \item If $s = \termone$, with $\termone \in \valset_{\typeone}$, then 
    we split $\actsone_{\lambdacont}$ as $\actsone_{\lambdacont} = \{\evalact\} \sqcup \{ \typeone \} \sqcup Z $ (observe that this splitting is indeed consistent with Definition~\ref{def:lmplambdap}), and we write: $\widehat{h}_s = \widehat{h}_s^{\{\evalact\}} \sqcup \widehat{h}_s^{\{\typeone\}}\sqcup Z$, with $\widehat{h}_s^{\{\evalact\}}:\{\evalact\} \kernel  \markovone_{\lambdacont}$ and $\widehat{h}_s^{\{\typeone\}}:\{\typeone\} \kernel  \markovone_{\lambdacont}$. It is immediate that those two kernels are Feller, since they are defined on one-point spaces, so we can conclude the thesis by Lemma~\ref{lemma:countable_sum_kernels}.
    \longv{\item if $s= \makereal {r'}$, we split $\actsone_{\lambdacont}$ as $\actsone_{\lambdacont} = \sqcup_{q \in \QQ}\{(\op_{\leq q})\}\sqcup \{\typereal\} \sqcup Z$.  Using Lemma~\ref{lemma:countable_sum_kernels}, we can conclude.}
    %and similar arguments as above on one-point spaces, we can conclude.
    %we see that it is enough to show that the kernel $\widehat{h}_s^{=} : r \in \RR \mapsto  (1 - \max(1, \lvert r' - r \rvert )) \cdot \dirac{\makereal{r'}}$ is Feller. Observe that for any bounded countinuous $f: \RR \rightarrow \RR$, and for any distribution $\nu$ over $\RR$, $\int_{\RR} f d(\widehat{h}_s^{=}(\nu)) = f(r')\cdot \int_{r \in \RR} (1 - \max(1, \lvert r' - r \rvert )) d\nu$. From there, and since the function $r \in \RR \mapsto  (1 - \max(1, \lvert r' - r \rvert ))$ is bounded continuous, we can conclude.
      \item If $s = \valtwo$, with $\valtwo$ of type $\typeone \to \typetwo$, then 
      we see that using Lemma~\ref{lemma:countable_sum_kernels} as above, it is enough to show that the kernel $\widehat{h}_s^{app} : \valone \in \valset_{\typeone} \mapsto  \dirac{\valtwo \valone}$ is Feller. Observe that this kernel can be written as the deterministic kernel $\overline f$, with $f(\valone) = \valtwo \valone$. Since $f$ is continuous, we can conclude the thesis because kernels deterministically generated by a continuous function are Feller continuous.      
      \end{varitemize}
\end{proof}

From Lemma~\ref{lemma:acts_kernels_feller_cont} and Lemma~\ref{lemma:states_kernels_feller_cont}, we immediately obtain Feller continuity of 
the applicative LMP for our continuous $\lambda$-calculus $\lmponelambdac$:
\begin{proposition}\label{prop:applicative_feller_continuity}
$\lmponelambdac$ is Feller continuous.
\end{proposition}

Finally, combining Proposition~\ref{prop:applicative_feller_continuity} and Theorem~\ref{prop:coincide_event_state_generic}, we obtain the main result of this section: 
the \emph{full abstraction theorem}. The latter states that all notions of bisimilarity we consider coincide on $\lmponelambdac$. Consequently, we obtain a characterization of bisimilarity by means of tests, and we can then use such a characterization to prove full abstraction of bisimilarity with respect to context equivalence, this way 
showing that all the notions of equivalence on $\lambdacont$ considered so far coincide. 

\begin{theorem}[Full Abstraction Theorem]\label{theorem:full_abstraction_lambdac}
%\sim_{state}^{\lmponelambdac}  =  \sim_{event}^{\lmponelambdac}$.
 $${\stackrel{\lmponelambdac}{\sim_{\statet}}} = {\stackrel{\lmponelambdac}{\sim_{\logic}}} 
  = {\stackrel{\lmponelambdac}{\sim_{\event}}} = {\stackrel{\lmponelambdac}{\sim_{\test}}} = {\stackrel{\lmponelambdac}{\equiv_{\mathit{ctx}}}}.
  $$
  \end{theorem}
\begin{proof}
By Proposition~\ref{prop:applicative_feller_continuity} and Theorem~\ref{prop:coincide_event_state_generic}, on $\lmponelambdac$ we have: 
${\sim_{\statet}} = \sim_{\logic}= \sim_{\event} = \sim_{\test}$. The inclusion ${\stackrel{\lmponelambdac}{\sim_{\statet}}} \subseteq {\stackrel{\lambdacont}{\equiv_{\mathit{ctx}}}}$ 
is obtained exactly as in the proof of soundness of applicative 
bisimilarity for $\lambdap$~\cite{LagoG19} (recall that state bisimilarity on 
$\lmponelambdac$ corresponds to applicative bisimilarity for $\lambdacont$ 
as defined by \citet{LagoG19}).
To conclude the proof, we need to show that ${\stackrel{\lambdacont}{\equiv_{\mathit{ctx}}}} 
\subseteq {\stackrel{\lmponelambdac}{\sim_{test}}}$. The proof is based on the following construction: for every test 
		$\testone 
		\in \tests$, we build by induction on $\testone$ a pair of contexts 
		$(\ctxone^{\valset, \typeone}_{\testone})$ and 
		$(\ctxone^{\terms,\typeone}_{\testone})$ such that the following 
		invariant holds, for all values $\valone \in \valset_{\typeone}$ and 
		terms $\termone \in \terms_{\typeone}$.
		\begin{align*}
			\probsucc \testone {(\valone,\typeone)} = \sem{ 
				{\ctxone^{\valset,\typeone}_{\testone}} \holefill\valone} (1)
			&, \quad
			\support{\sem{ {\ctxone^{\valset,\typeone}_{\testone}} 
					\holefill\valone}} \subseteq [0,1] 
			\\
			\probsucc \testone {(\termone,\typeone)} = \sem{ 
				{\ctxone^{\terms,\typeone}_{\testone}} \holefill\termone} (1)
			&, \quad
			\support{\sem{ {\ctxone^{\terms,\typeone}_{\testone}} 
					\holefill\valone}} \subseteq [0,1] .
		\end{align*}       
		The construction is similar to the one used in the proof of 
		full abstraction of bisimilarity for discrete probabilistic 
		$\lambda$-calculi~\cite{CDL14, DBLP:conf/birthday/CrubilleLSV15}, 
		the most interesting case now being the 
		one of $\testone = \overset{\text{?}}{\leq} q \cdot \testone'$, 
		with $q \in \QQ$. We take $\ctxone^{\mathcal U,\typeone}_{\testone} = 
		\makereal 0$, 
		if $\typeone \neq \typereal$ or $\mathcal U \neq \valset$; and 
		$\ctxone^{\mathcal \valset,\typereal}_{\testone} = 
		\makereal{\mathit{op}_{\leq}}(\sample,\makereal{\mathit{op}_{\leq}}(q,[])) \makereal{*} 
		\ctxone^{\valset,\typereal}_{\testone'}[]$, otherwise.
		We then see that when two programs $\termone,\termtwo$ of type 
		$\typeone$ are context equivalent, then they have the same probability 
		of success for any $\testone \in 
		\tests$, and from there we can deduce 
		that they are test bisimilar.
  %Recall that we know already that $\sim_{state}^{\lmponelambdac}  \subseteq \sim_{event}^{\lmponelambdac}$, since it is true for every LMP--see Proposition~\ref{proposition:state_incl_event_bisim}. We now look at the other direction, that we show using aour auxilliary LMP $\lmponelambdac'$, and the connection between $\lmponelambdac' $ and $\lmponelambdac$: first, since $\lmponelambdac'$ is Feller--from Proposition~\ref{prop:applicative_feller_continuity}--we obtain from Theorem~\ref{prop:coincide_event_state_generic} that $\sim_{state}^{\lmponelambdac'}  = \sim_{event}^{\lmponelambdac'}$. From there, and using Lemma~\ref{lemma:aux_lmp_equivalent}, we can form the following inclusion chain: $\sim_{state}^{\lmponelambdac} = \sim_{state}^{\lmponelambdac'} = \sim_{event}^{\lmponelambdac'} \subseteq \sim_{event}^{\lmponelambdac}$. \textcolor{blue}{Warning: it is not the good direction !!!}
  \end{proof}

\section{Perspectives on the General Case}\label{sect:perspective}
After having proved that continuity of the first-order functions implies 
coincidence between the four notions of equivalence we have introduced in 
Section~\ref{sect:equivalences_on_lmps}, we go back to the general case, namely 
the one of $\lambdap$. The only difference with respect to $\lambdacont$ is 
that first-order functions are not necessarily continuous. This, as we are 
going to see, turns out to have very strong consequences on the nature of 
coinductive relational reasoning. Before proceeding further, 
let us consider an example, which very clearly shows where the crux is.

\begin{example}\label{ex:counter_example_state}
	Consider the following two programs:
	\begin{align*}
		& \termone: \quad \seq{\sample}{(\abs y {\ifth {x=y}{1}{0}})} \\
		& \termtwo: \quad \seq{\sample}{\abs y 0}.
	\end{align*}
	The programs $\termone$ and $\termtwo$ are contextually equivalent. The 
	easiest 
	way to prove that is by going through a denotational semantics of 
	$\lambdap$.
	\citet{DBLP:journals/pacmpl/VakarKS19} have recently introduced quasi-Borel 
	domains 
	and have proved their adequacy as a denotational semantics for a Bayesian 
	$\lambda$-calculus with recursive types subsuming $\lambdap$. 
	Quasi-Borel domains thus provides an adequate semantics for $\lambdap$ too, 
	and 
	one can show the denotational interpretations of $\termone$ and $\termtwo$ 
	to coincide.\footnote{ 
		The proof of denotational equality of $\termone$ and $\termtwo$ (and 
		thus of their contextual 
		equivalence) has been communicated to the authors by Dario Stein based 
		on previous 
		work by \citet{DBLP:journals/pacmpl/SabokSSW21}. \shortv{Details can be 
		found in the long version \cite{LV}}.
	} 
	\longv{The proof is done in the denotational model based on Quasi-Borel 
		spaces, for a monadic call-by-value probabilistic language. We can show 
		that their proof also holds for $\lambdap$ by formalizing an embedding 
		of 
		$\lambdap$ into that language.} As highlighted 
	by \citet{DBLP:journals/pacmpl/SabokSSW21} relying on a connection with the 
	privacy 
	rule in $\nu$-calculus, the equivalence between $\termone$ and $\termtwo$ 
	can be understood 
	from
	a security-oriented viewpoint: we see the sampling of $x$ as the generation 
	of 
	a 
	\emph{private name} --- or in cryptographic terms, a private key. Indeed, no 
	adversary --- i.e.~external context---can \emph{a priori} guess the private 
	key $x$ with non-negligible probability. On the other hand, are these 
	programs also state bisimilar? Actually, state bisimilarity of $\termone$ 
	and $\termtwo$ seems very hard to prove. How, then, about event 
	bisimilarity? 
\end{example}

\subsection{Event Bisimulation is \emph{Unsound}}

\shortv{
	\citet{LagoG19} have proved soundness of state bisimilarity 
	with respect to context equivalence in $\lmponelambda$:
	\begin{proposition}[\cite{LagoG19}]\label{prop:LagoG19}
		${\sim_{\statet}^{\lmponelambda}} \subseteq { \equiv^{\mathit{ctx}}}$.
		% The state bisimilarity on $\lmponelambda$ is sound with respect to 
		%context 
		% equivalence 
		% for $\Lambda_{\probb}$:
		% \begin{align*}
		%   &\forall \typeone \in \typeone, \forall \termone,\termtwo \in 
		%   \terms_\typeone,\qquad  
		%   (\termone,\typeone) \sim_{state}^{\lmponelambda} 
		%(\termtwo,\typeone) 
		%\quad 
		%   \Longrightarrow 
		%   \quad \termone \equiv^{ctx} \termtwo; \\
		%    &\forall \typeone \in \typeone, \forall \valone,\valtwo \in 
		%    \terms_\typeone,\qquad  
		%    (\valone,\typeone) \sim_{state}^{\lmponelambda} (\valtwo,\typeone) 
		%\quad 
		%    \Longrightarrow 
		%    \quad \valone \equiv^{ctx} \valtwo. 
		%   \end{align*}
	\end{proposition}
At that point, the reader may wonder whether it is possible to use state 
bisimilarity and Proposition~\ref{prop:LagoG19} to prove 
Example~\ref{ex:counter_example_state}. Unfortunately that would not be so 
easy, because state bisimilarity when applied to the programs from 
Example~\ref{ex:counter_example_state} makes seemingly intractable 
measurability problems to appear, as we illustrate further in 
Example~\ref{example:state_sim_fragment}. In turn, 
Theorem~\ref{prop:logical_caracterisation} tells us 
that state bisimilarity is included in event bisimilarity. Event 
bisimilarity, then, would seem to be a natural candidate to go towards 
contextual equivalence. 
\begin{example}
If we look again to the programs $\termone$ and $\termtwo$ from 
Example~\ref{ex:counter_example_state} --- that are context equivalent --- we 
can show explicitly that they are also event bisimilar \shortv{(for a formal 
proof, 
see~\cite{LV})}. Indeed, recall from 
the testing characterization of event bisimilarity given by 
Theorem~\ref{prop:logical_caracterisation}, that an \enquote{adversary} for 
event bisimilarity can only do three things: evaluating
a program (that we see as an agent in some cryptographic game); passing 
messages to this agent; doing two separate \emph{copies} of any current state 
of the agent. Moreover, the adversary has to decide \emph{before the game  
starts} which sequence of actions to play. Intuitively, we can see that none of 
its admissible strategies allows the adversary to distinguish between 
$\termone$ and $\termtwo$ with non-zero probability.
\end{example}
}
\longv{ We look now at programs $\termone$ and $\termtwo$ from 
	Example~\ref{ex:counter_example_state} from the point of view of event 
	bisimulation. We are going to show that they are event bisimilar; to do 
	this we 
	first need to present additional tools from~\cite{danos2005almost}.
	\begin{definition}
		A $\pi$-system on $\markovone$ is a subset of $\markovone$ closed
		under (finite) intersections and containing $\markovone$.
	\end{definition}
	\begin{lemma}[From~\cite{danos2005almost}]\label{lemma:pi_system_event_bisim}
		Let $\pi$ be a $\pi$-system. Then $\sigma(\pi)$ is an event 
		bisimulation if $\pi$ is \emph{stabl.e}
	\end{lemma}
	%\begin{lemma}
	%The states $\termone$ and $\termtwo$ are \emph{event bisimilar}.
	%\end{lemma}
	\begin{definition}
		We note $\termthree_r = \abs \varone {\ifth{(\varone==r)}{1}{0}}$. 
		We consider the following $\pi$-system:
		\begin{align*}
			\pi & := \left\{\markovone, \emptyset, \{\termone, \termtwo\}, 
			\{1\},\{0\} \right\}  \cup \{A_r \mid r \in \RR\}  \cup \{B_I \mid 
			I \subseteq \RR \text{ finite set}\} \\
			A_r & := \{ \termthree_r \} \\
			B_I & := \{\abs \varone 0\} \cup \{\termthree_r \mid r \not \in I \}
		\end{align*}
		We take $\Lambda$ as the $\sigma$-algebra generated by $\pi$. %We now 
		%need to show that $(\markovone, \Lambda, \{h_a \mid a \in \actsone 
		%\})$ 
		%is a LMP.
	\end{definition}
	\begin{lemma}\label{stable_pi_system}
		$\pi$ is stable.
	\end{lemma}
	\begin{proof}
		We need to check that for every $A \in \Lambda$, $a \in \actsone$,
		and $q \in [0,1[$, it holds that $ \act a q A \in \pi$. Here, we
		give some examples of $\act a q A$, And we can check that they
		are indeed in $\pi$:
		\begin{align*}
			& \act{eval}{q}{A_r} = \emptyset; \\
			& \act{eval}{q}{B_I} = \{\termone, \termtwo\};  \\
			& \act{r}{q}{\{0\}} = B_{\{r\}}; \\
			& \act{r}{q}{\{1\}}= A_r  \ldots
		\end{align*}
		\TODO{Is this a proper enumeration of all the sets in the
			form $\act a q A$?}
	\end{proof}
	\begin{proposition}
		The states $\termone$ and $\termtwo$ are \emph{event bisimilar}.
	\end{proposition}
	\begin{proof}
		It can be directly obtained by combining Lemma~\ref{stable_pi_system}
		and Lemma~\ref{lemma:pi_system_event_bisim}. 
	\end{proof}
}
Using the testing characterization from 
Theorem~\ref{prop:logical_caracterisation}, we can in fact show that 
context equivalence is contained in event applicative bisimilarity. 
\begin{proposition}
	${\equiv^{\mathit{ctx}}} \subseteq {\sim_{\event}^{\lmponelambda}}$.
\end{proposition}
\begin{proof}
Recall that  event bisimilarity is caracterised by tests, i.e. ${\sim_{\event}^{\lmponelambda}} = {\sim_{\test}^{\lmponelambda}}$. From there, the proof is done exactly as in the proof of Theorem~\eqref{theorem:full_abstraction_lambdac}, by building for every test a context that emulates it.
\end{proof}
The problem with event bisimilarity, at least when spelled out for $\Lambda_\probb$, is that you 
simply become too coarse, thus unsound for context equivalence. The problem 
here lies in the fact that programs can do 
strictly more than what an \enquote{adversary} for event bisimilarity does: indeed they 
can somehow perform \emph{man-in-the-middle attacks}; by contrast once the adversary from event 
bisimilarity has created two copies of a process, it cannot make them \emph{interact} with each other. This is exploited in the 
following example.

\begin{example}\label{ex:counter_example_soundness}
	Let us consider the following two programs:
	\begin{align*}
		\termone &: \quad \seq{\sample}{\abs y{((\ifth {x=y}{1}{0})\oplus x )}} \\
		\termtwo &: \quad \seq{\sample}{\abs y({0}\oplus x )},
	\end{align*}
	where $\oplus$ is as in Example~\ref{example:distributions}.
	The two programs first sample a real number $x$, to be thought of as a 
	private key. Then, with probability $\frac 1 2$, they behave as the 
	corresponding process of Example~\ref{ex:counter_example_state}, but with 
	probability $\frac 1 2$ they just \emph{reveal} their secret key. It means 
	that a context, seen as an adversary, can perform the following attack: 
	first, it \emph{initiates} a session, which means that the process needs to 
	choose its secret key. Then the adversary creates two copies of the underlying process: 
	from there, it obtains with probability $\frac 1 4$ the original process of 
	Example~\ref{ex:counter_example_state}, and the secret key, and it can 
	conclude its attack by passing the key to the process. Indeed, 
	it turns out that $\termone$ and $\termtwo$ are not context 
	equivalent, and we can prove it by building a context that distinguishes them. 
	\shortv{We take $C = (\abs z {z (z 1)})[\,]$ and observe 
	(a formal proof for this can be found in~\cite{LV}) that: 
	$\sem{C[\termone]} = 
			\frac 1 2 \lambda + \frac 1 4 \dirac 1  + \frac 1 4 \dirac 0 
	$ and $\sem{C[\termtwo]} 
			= \frac 1 2 \lambda +  \frac 1 2 \dirac 0$,
                where $\lambda$ denotes the uniform distribution on 
                $[0,1]$.}\longv{
		We take $C = (\abs z {z (z 1)})[\,].$
		To compute $\sem{C[\termone]}$ and $\sem{C[\termtwo]}$, 
		we the intermediate functions $f_\termone, f_\termtwo : \RR \rightarrow 
		\valset$ defined thus:
		\begin{align*}
			f_\termone (r) &:= \abs y ((\ifth {r==y}{1}{0})\oplus r ) 
			\\
			f_\termtwo(r) &:=  \abs y (0 \oplus r).
		\end{align*}
		This way, we have 
		$\sem{\termone}(A) = \lambda\{r \mid f_\termone(r) \in A \}$ and 
		$\sem{\termtwo}(A) = \lambda\{r \mid f_{\termtwo}(r) \in A \}$,
		% \begin{align*}
		%   \sem{\termone}(A) &= \lambda\{r \mid f_\termone(r) \in A \}, 
		%   &\quad
		%   \sem{\termtwo}(A) &= \lambda\{r \mid f_{\termtwo}(r) \in A \},
		% \end{align*}
		where we recall that $\lambda$ stands for the Lebesgue measure on 
		$[0,1]$. 
		For $\termthree \in \{\termone, \termtwo\}$, we thus have
		$\sem{C[\termthree]}(A) = \int \sem{f_\termthree (r) (f_\termthree (r) 
			1)} (A) \lambda(dr)$.
		We need now to compute $\sem{f_\termthree (r) (f_\termthree (r) 1)}$ 
		for 
		$\termthree \in \{\termone, \termtwo\}$: observe that the probabilistic 
		behavior here is purely \emph{discrete}:
		\begin{align*}
			\sem{f_\termone (r) 1} &= \frac 1 2 \dirac{\delta_1(r)} + \frac 1 2 
			\dirac r,  
			%\qquad \text{ with } \delta_1(r) = 1 \text{ if } r=1, \, 0 \text{ 
			%otherwise.}  
			&\text{ }
			\sem{f_\termtwo (r) 1} &= \frac 1 2 \dirac{0} + \frac 1 2 \dirac r,
		\end{align*}
		with $\delta_1(r) = 1$ if $r=1$, and $\delta_1(r) = 0$ otherwise.
		From there, we can show that
		\begin{align*}
			\sem{f_\termone (r) (f_\termone (r) 1)} &= \frac 1 2 \dirac r + 
			\frac 1 4 \dirac 1 + \frac 1 4 \dirac{\delta_{0,1}(r)}\\
			\sem{f_\termtwo (r) (f_\termtwo (r) 1)} &= \frac 1 2 \dirac 0 + 
			\frac 1 2 \dirac r,
		\end{align*}
		so that we can compute
		\begin{align*}
			\sem{C[\termone]} &= 
			\frac 1 2 \lambda + \frac 1 4 \dirac 1  + \frac 1 4 \dirac 0 
			& \quad
			\sem{C[\termtwo]} 
			&= \frac 1 2 \lambda +  \frac 1 2 \dirac 0.
		\end{align*}
		%   We are now ready to finally compute  $\sem{C[\termone]}$ and 
		%$\sem{C[\termtwo]}$:
		%   \begin{align*}
		%     \sem{C[\termone]} &=  A \mapsto \int (\frac 1 2 \dirac r + \frac 
		%1 4 \dirac 1 + \frac 1 4 \dirac{\delta_{0,1}(r)}) (A) \lambda(dr) \\
		%     &= A \mapsto \frac 1 2 \int ( \dirac r)(A) \lambda(dr) + \frac 1 
		%4 \int (\dirac 1) (A)\lambda(dr)  + \frac 1 4 \int 
		%(\dirac{\delta_{0,1}(r)}) (A) \lambda(dr) \\
		%      &= \frac 1 2 \lambda + \frac 1 4 \dirac 1 + \frac 1 4 \int 
		%(\dirac{0}) (A) \lambda(dr)  \qquad \text{ since }\lambda(\{0,1\}) = 0 
		%\\
		%     &= \frac 1 2 \lambda + \frac 1 4 \dirac 1  + \frac 1 4 \dirac 0 \\
		% & \, \\
		%     \sem{C[\termtwo]} &= A \mapsto \int (\frac 1 2 \dirac r + \frac 1 
		%2 \dirac 0) (A) \lambda(dr) \\
		%     &=  A \mapsto \frac 1 2 \int ( \dirac r)(A) \lambda(dr) + \frac 1 
		%2\int (\dirac 0) (A)\lambda(dr)  \\
		%      &= \frac 1 2 \lambda +  \frac 1 2 \dirac 0 \\
		%   \end{align*}
	} From there, it is enough to compose $C$ with a context $D$ that tests the 
	equality to $1$.
\end{example}

The programs $\termone, \termtwo$ of Example~\ref{ex:counter_example_soundness} 
are designed as a counter-example to the soundness of event bisimilarity. We 
have shown in Example~\ref{ex:counter_example_soundness} that they are not context \emph{equivalent}: now, we are going to show that they are event \emph{bisimilar}, 
using our approximation Lemma~\ref{th:event_bisim_approx}.
\begin{example}\label{ex:counter_example_soundness2}
	Let $\valsettwo$ be any countable set of values.  
	First, we define a restriction of the LMP $\lmponelambda$ from 
	Definition~\ref{def:lmplambdap}: we keep the same state space, but we 
	restrict the transitions to the ones labelled by a countable restriction of 
	$\actsone_{\Lambda_{\probb}}$, that depends on $\valsettwo$. More formally, 
	we take: 
    $$
    \lmpone^{\text{app}}_{\valsettwo} := (\markovone_{\Lambda_{\probb}}, 
    \actsone^{\text{app}}_{\valsettwo}, \{h_a \mid a \in   
    \actsone^{\text{app}}_{\valsettwo}\} ),
    $$
    where $\actsone^{\text{app}}_{\valsettwo}$ is the standard Borel space defined 
    by taking the restriction of 
    $\actsone_{\Lambda_{\probb}}$ to the set
    $\types \cup \valsettwo \cup 
    \{\evalact \}\cup \{ \overset{\text{?}}{\leq} q \mid q \in \QQ \}
  	\cup \{\text{case}(\hat{\imath}) \mid i \in I\} \cup\{ \unboxact\}$.
	We are going to show that 
	the programs $\termone$ and $\termtwo$ are bisimilar in each of the 
	$\lmpone^{\text{app}}_{\valsettwo}$ --- that have countable labels and 
	analytical state spaces, thus state bisimilarity and event bisimilarity 
	coincide there (by Theorem~\ref{prop:logical_caracterisation_state_c}). 
	From this point, we will be able to conclude by 
	Lemma~\ref{th:event_bisim_approx} that $\termone$ and $\termtwo$ are also 
	event bisimilar in $\lmponelambda$.
        
	We write $\relone_\valsettwo$ for the bisimilarity  on $\lmpone^{\text{app}}_{\valsettwo}$.
	We define: 
	$S= \{(\termone, \termtwo)\} \cup \{(\hat{f_\termone(r)}, 
	\hat{f_\termtwo(r)}) \mid r \in \RR, r\not \in \valsettwo \} \cup \{(r,r) 
	\mid r \in \mathbb \RR\},
	$
	where 
	$f_\termone, f_\termtwo : \RR \rightarrow \valset$ are defined thus:
	\begin{align*}
		f_\termone (r) &:= \abs y ((\ifth {r=y}{1}{0})\oplus r );
		&
		f_\termtwo(r) &:=  \abs y (0 \oplus r);
	\end{align*}
	and we show that the reflexive closure of $S$ --- that we note $\relone$ ---  
	is a bisimulation. That is: for each $(a,b) \in S$ and action $\alpha$, $a 
	\rightarrow _\alpha \mu_a$, $b \rightarrow_\alpha \mu_b$ implies that 
	$\mu_a \Gamma \relone \mu_b$.
	\begin{varitemize}
		\item Let $r \in \RR$ such that $r \not \in \valsettwo$, and $\alpha$ be a 
		\emph{relevant} action for $(f_\termone(r),f_{\termtwo}(r))$. It means 
		that $\alpha$ is an applicative action with a value of real type, so is 
		of the form $r' \in \RR$ with $r' \in \valsettwo$. Observe that in 
		particular, it means that $r' \neq r$. So we need to show that  
		$\sem{f_{\termone}(r)r'} \Gamma \relone \sem{f_{\termtwo}(r)r'}$. Since 
		we know that $r \neq r'$, we can see
		$\sem{f_{\termone}(r)r'} = \sem{f_{\termtwo}(r)r'} = \frac 1 2 \dirac 
		0 + \frac 1 2 \dirac r,$
		which allows us---since $\relone$ is reflexive---to immediately conclude 
		that
		$\sem{f_{\termone}(r)r'} \Gamma \relone \sem{f_{\termtwo}(r)r'}$.
		\item We look now at the pair $(\termone,\termtwo) \in S$. The only 
		relevant action there is the $\text{eval}$ action, so we can 
		reformulate our goal as:
		$\sem \termone \Gamma \relone \sem \termtwo$.
		To do that, we consider any $A \in \Sigma^{\text{app}}$ which is 
		$\relone$-closed. We need to show that $\sem{\termone}(A) = 
		\sem{\termtwo}(A)$. First, looking at the operational semantics for 
		$\termone$ and $\termtwo$ that we have computed in 
		Example~\ref{ex:counter_example_soundness}, we see that (notice that 
		$\lambda(\valsettwo) = 0$):
		\shortv{
			\begin{align*}
				\sem \termone (A) 
				&= \lambda\{ r \in \RR\setminus{\valsettwo} \mid f_\termone(r) 
				\in A \};
				&
				\sem \termtwo (A) &= 
				\lambda\{ r \in \RR\setminus{\valsettwo} \mid f_\termtwo(r) \in A 
				\}.
			\end{align*}
		}
		\longv{
			\begin{align*}
				\sem \termone (A) 
				&= \lambda\{ r \in \RR \mid f_\termone(r) \in A \} 
				\\ 
				&= \lambda\{ r \in \RR\setminus{\valsettwo} \mid f_\termone(r) 
				\in A \} + \lambda\{ r \in \RR \cap {\valsettwo} \mid 
				f_\termone(r) \in A \}  
				\\ 
				&=  \lambda\{ r \in \RR\setminus{\valsettwo} \mid f_\termone(r) 
				\in A \} \quad \text{ since }\lambda(\valsettwo) = 0 \\ \, \text{ 
					and similarly:} 
				&\\       \sem \termtwo (A) &= \lambda\{ r \in \RR \mid 
				f_\termtwo(r) \in A \} = 
				\lambda\{ r \in \RR\setminus{\valsettwo} \mid f_\termtwo(r) \in A 
				\}
			\end{align*}
		}
		Recall that $A$ is $R$-closed: it means that for any $r \in \valsettwo$, 
		$f_\termone(r) \in A$ if and only if $f_\termtwo(r) \in A$, thus $\{ r 
		\in \RR\setminus{\valsettwo} \mid f_\termone(r) \in A \} = \{ r \in 
		\RR\setminus{\valsettwo} \mid f_\termtwo(r) \in A \}$, which ends the 
		proof.
	\end{varitemize}
\end{example}

\begin{theorem}
	${\sim_{\event}^{\lmponelambda}} \not \subseteq {\equiv^{\mathit{ctx}}}$. 
\end{theorem}
\begin{proof}
	Directly from Example~\ref{ex:counter_example_soundness}, 
	Example~\ref{ex:counter_example_soundness2}, and our approximation 
	Lemma~\ref{th:event_bisim_approx}.
\end{proof}

\subsection{On Full Abstraction for State Bisimilarity}
\label{subsect:on-full-abstraction-for-state-bisimilarity}

%\begin{proof}
%\textcolor{blue}{say something about extending the soundness proof to our more 
%discriminating notion of context equivalence.}
%\end{proof}

Unfortunately, $\lmponelambda$ is \emph{not} one among the LMPs to which 
the correspondence results from Section~\ref{sect:equivalences_on_lmps} can be 
applied. Indeed, as we have just proved, logical equivalence, testing 
equivalence, and event bisimilarity are in fact \emph{not sound}. 
%The deep reason for that does not lie in the nature of $\lmponelambda$'s 
%\emph{state space} (which turns out to be analytic), but rather on its 
%\emph{actions}, which do not form a countable set: actions model parameter 
%passing, and that any real number can be a parameter. There is nothing in 
%$\lmponelambda$ guaranteeing the dynamics to be continuous not only on the 
%state, but also on the action. \ugo{Raphaelle: could you please add a couple 
%%%of sentences providing a deeper insight, here?}
%Is there anything in $\lambdap$ which is somehow responsible for that? Is 
%there 
%any class of LMPs \emph{with uncountably many labels} for which the results 
%from Section~\ref{sect:equivalences_on_lmps} still hold? In the next two 
%sections, we will provide some deep insights about the two questions above, 
%technically speaking the main contributions of this paper.
Completeness of applicative bisimilarity was left as an open question 
in the work by \citet{LagoG19}. We still cannot give a definite answer to that 
question, but some further light on the subject can indeed be given.

\begin{example}\label{example:state_sim_fragment}
Consider, once again, the programs $\termone$ and $\termtwo$ from 
Example~\ref{ex:counter_example_state}. The relevant fragment of the underlying 
LMPs is in Figure~\ref{fig:counter_example}, call it $\lmpone_{\mid}$.   
\begin{figure}
	\centering
	\scalebox{0.6}{
		\begin{tikzpicture}
			\node (M)[draw, rounded corners=5pt] at (0,1){$\termone$};
			\node (N)[draw, rounded corners=5pt] at (7,1){$\termtwo$};
			\node (N2)[draw, rounded corners = 5pt] at (7,-2){$\abs \varone 0$};
			\node (M2a)[draw, rounded corners = 5pt, text width=2.7cm] at 
			(-3,-2)
			{$\abs y {} {\textbf{if } {(y=r)}}$ \\ $\textbf{then } 1 \textbf{ 
			else } 0$};
			\node (M2c)[ minimum height=1cm] at (- 0.5,-2){\ldots};
			\node (M2b)[draw, rounded corners = 5pt, text width=2.7cm] at (2,-2)
			{$\abs y {} {\textbf{if } {(y=r')}}$ \\ $\textbf{then } 1 \textbf{ 
			else } 0$};
			\node at (4.1,-2){\ldots};
			\node (aux) at (0,-0.5){};
			\draw (M)--(aux) node[near end, left]{eval};
			\draw[->] (aux)--(M2a);
			\draw[->] (aux)-- (M2b);
			\draw[->] (aux)-- (M2c);
			\draw[->] (N)--(N2) node[midway, right] {eval};
			\node (zero)[draw, rounded corners = 5pt] at (3.5,-5){$0$};
			\node (un)[draw, rounded corners = 5pt] at (-1,-5){$1$};
			\draw[->] (N2)--(zero) node[midway, right] {$r$};
			\draw[->, bend left] (N2) to  node[midway, right] {$r'$} (zero);
			\draw[->, bend right] (N2) to node[midway, right] {$r''$}(zero);
			\draw[->, color=red] (M2a)--(un) node[midway, left] {$r$};
			\draw[->, color=red] (M2a) to  node[midway, right] {$r'$} (zero);
			\draw[->, bend right, color=red] (M2a) to node[midway, right] 
			{$r''$}(zero);
			
			\draw[->, color=blue] (M2b)--(un) node[midway, left] {$r'$};
			\draw[->, color=blue] (M2b) to  node[midway, right] {$r$} (zero);
			\draw[->, bend right, color=blue] (M2b) to node[midway, right] 
			{$r''$}(zero);
			\path(zero) edge [loop below] node {$=0$} (zero);
			\path(un) edge [loop below] node {$=1$} (un);
	\end{tikzpicture}}
	\caption{The programs $\termone, \termtwo$: it holds that $\termone \not 
		\sim_{\statet} 
		\termtwo$, $\termone \equiv^{\mathit{ctx}} \termtwo$, $\termone 
		\sim_{\event} 
		\termtwo$.}\label{fig:counter_example}
\end{figure}
Observe that $\lmpone_{\mid}$ is a \emph{stable fragment} of $\lmponelambda$, in the sense that from every state $s$ in this fragment, the probability to end up again in this fragment after doing any action $a$ is $1$.
%for every state $s$ in $\lmpone_{\mid}$, for any action in $\actsone_{\lambdap}$, there exists a subset $A \subseteq \lmpone_{\mid}$ such that $h_a^{\lmponelambda}(s)(A)=1$. 
Clearly, this is only a fragment of the whole applicative LMP, but \emph{in 
this fragment}, it is easy to show that $\termone$ and $\termtwo$ are not 
bisimilar. Let $\relone$ be a bisimulation. We can first observe easily that 
$(1,0) \not \in\relone$. From there, we can see that for each state of the form:
$$
\termthree_r: \quad \abs y {} {\mathbf{if}\  {(y = r)}}\ \mathbf{then}\ 1\
\mathbf{else}\ 0,
$$
it holds that $(\termthree_r, \abs \varone 0)\not \in \relone$: that is 
because the action $a= r$ goes to $1$ with 	probability $1$ when starting from 
the state $\termthree_r$, and to 	$0$ with probability $0$ when starting 
from $\abs \varone 	0$. Looking at Figure~\ref{fig:counter_example}, we see 
that there 	are actually no other states in $\lmpone_{\mid}$ such that 
$(v,\abs 	\varone 0) \in \relone$. As a corollary, the set $\{\abs 
\varone 	0\}$ is $\relone$-closed. Since moreover it is measureable, it 
means 	that $\{ \abs \varone 0 \} \in \Sigma(\relone)$.  From there, we 
end 	the proof by contradiction. Suppose that $\termone \relone 	\termtwo$. 
We should have $h_{\evalact}(\termone, \{\abs \varone 0\}) = 	
h_{\mathit{eval}}(\termtwo, \{\abs \varone 0\})$, but it is not the case 	
since $h_{\evalact}(\termone, \{\abs \varone 0\}) = 0$, while 	
$h_{\evalact}(\termtwo, \{\abs \varone 0\})=1$.

It should be noted that to show 
that $\{\abs \varone 0\}$ is $\relone$-closed, we implicitely use the fact that 
there are no other states in the LMP, which is not the case if we consider the
whole applicative LMP instead of just a fragment. In order to adapt the proof for the whole LMP, we would need to show the existence of a $\relone$-closed measurable set $X$ that contains $\lambda x.0$, and not more that a negligible number of the $\termthree_r$.
 We were however not able to show or disprove the existence of such a set $X$. 
 A natural choice for $X$ could seem to be the $\relone$-equivalence class of 
 $\abs \varone 0$ instead of simply $\{\abs \varone 0\}$. However, while this 
 set is indeed $\relone$-closed, it is however potentially not measurable. We 
 can indeed show the existence of primitive functions which makes the 
 equivalence class of $\{\abs \varone 0\}$ non measurable \shortv{(the proof, 
 that can be found in the long version \cite{LV}, use the same primitive 
 function as in 
 Proposition~\ref{prop:contextual-equivalence-is-not-measurable})}. 
\end{example}

%It is possible to show--see more details on the long version--that for event bisimulation,
It would be intuitively appealing to be able to say that whenever two states are not bisimilar in a stable fragment, as in Example~\ref{example:state_sim_fragment}, then they are also not bisimilar in the whole LMP. We call this property a \emph{locality property}: it formalizes the philosophical requirement that whenever two processes in a transition system are not equivalent, this should stay true when we add to the global system some other processes that do not interact at all with the original processes. From a more practical point of view, locality is appealing too: a non-local relation on a LMP is indeed much less usable than a local one, since no matter how simple are the processes we are really interested in, we are forced to consider the whole LMP with an uncountable state space.

We can see, using its logical characterization, that this locality property holds for 
$\sim_{\event}$. Indeed, whether a formula $\phi$ holds for some state $s$ does not depend on the behaviour of states independent from $s$. By contrast, we are not able to show locality in general for $\sim_{\statet}$ when both the states and the action space are standard Borel; crucially we are not able to show it in the special case of the LMP $\lmponelambda$. We have not been able either to find a concrete proof that locality does not hold for $\lmponelambda$. The crux of the matter, as illustrated in Example~\ref{example:state_sim_fragment}, is the definition of the relator $\relator$ from Definition~\ref{def:relation-lifting}. As mentioned in Definition~\ref{def:relator_two}, there exists an alternative notion of relator $\relatortwo$ based on \emph{couplings}, that coincide with $\relator$ in the restricted case of Borel relations 
(cf. Theorem~\ref{thm:theta-equal-gamma}). While the relation we would obtain when replacing $\relator$ by $\relatortwo$ would be local, it is not know whether it is transitive, and thus an equivalence. Moreover this new state bisimilarity would not be complete, since it can be shown that it distinguishes the two programs of Example~\ref{example:state_sim_fragment}.

\subsection{Bounding Non-Measurability in Descriptive Set-Theory}

For context equivalence, we can go beyond our negative non-measurability result in Proposition~\ref{prop:contextual-equivalence-is-not-measurable}, by using tools from descriptive set theory~\cite{kechris} to bound \emph{how much} non-measurable $\ctxeq$ is. 
We can make the notion
of non-measurability more precise relying on the so-called projective 
hierarchy~\cite{kechris} --- 
the latter starting from Borel sets, and defining bigger and bigger classes 
using projections and complementation.  More precisely, we can show that 
$\ctxeq$ is in the class 
of \emph{co-analytic} sets, i.e. those that are defined as the complementary of an \emph{analytic} set.
%--- while analytic sets are defined as the set obtained as the continuous image of a Borel set. 
All of that means that (the graph of) context equivalence is in the first layer of the projective hierarchy.
\shortv{(The proof of Lemma~\ref{lemma:graph_equiv_ctx_analytic}
strictly follows the one of Proposition~\ref{prop:eq_context_lambdacont_Borel} and it can be 
found the long version of this paper for details.)}

\begin{lemma}\label{lemma:graph_equiv_ctx_analytic}
 The graph of ${\ctxeq} \subseteq \programs \times \programs$ is a co-analytic set, and thus for every program $M$, the equivalence class of $M$ by $\ctxeq$ is a co-analytic set.
\end{lemma}
\longv{
\begin{proof}
First, since the class of co-analytic sets is stable by countable union 
(see~\cite{kechris}, Chapter 4), we see that it is sufficient to show that for 
every type $\typeone$, $\equiv^{ctx}$ restricted to programs of type $\typeone$ 
is co-analytic. Moreover, it holds that:
\begin{align*}
\equiv^{ctx}_\typeone &= \bigcap_{\substack{C \text{ a pre-context} \\ \text{ with }  k\text{ holes }}} \bigcap_{q \in \QQ \cup{+\infty}} \quad \{(M,N) \mid \forall \vec r \in \RR^{k}, \, \sem {C[\vec r][M]}(]-\infty,q]) = \sem{C[\vec r][N]}(]-\infty,q)\}
 \end{align*}
We show the equation above in exactly the same way as for $\lambdacont$ in the 
proof of Proposition~\ref{prop:eq_context_lambdacont_Borel}.
%By definition of context equivalence, we can write:
%\begin{align*}
%\equiv^{ctx}_\typeone &= \cap_{C \text{ a pre-context with }k\text{ holes }} \{(M,N) \mid \forall \vec r \in \RR^{k}, \, \sem {C[\vec r][M]} = \sem{C[\vec r][N]}\}\\
% &= \cap_{C \text{ a pre-context with }k\text{ holes }} \cap_{q \in \QQ \cup{+\infty}}\{(M,N) \mid \forall \vec r \in \RR^{k}, \, \sem {C[\vec r][M]}(]-\infty,q]) = \sem{C[\vec r][N]}(]-\infty,q)\}\\
% & \qquad \qquad \qquad \qquad \qquad \qquad  \text{ by Example~\ref{ex:enough_primitives} and Lemma~\ref{lemma:expectation_as_observable_1}}.
%\end{align*}
The class of co-analytic sets is also closed by countable intersections, thus it is enough to show that for every pre-context $C$, and every $q \in \QQ \cup \{+ \infty\}$, the set $S_{C,q} :=\{(M,N) \mid \forall \vec r \in \RR^{k}, \, \sem {C[\vec r][M]}(]-\infty,q]) = \sem{C[\vec r][N]}(]-\infty,q)\}$ is co-analytic. As a first step, observe that we can rewrite $S_{C,q}$ as the complementary of:
\begin{equation}\label{eq:projection_from_borel}
\overline{S_{C,q}} := \{(M,N) \mid \exists r \in \RR^{k},(\vec r,M,N) \in A_{C,q} \},
\end{equation}
where $A_{C,q} = \{(r,M,N) \mid \sem {C[\vec r][M]}(]-\infty,q]) \neq \sem{C[\vec r][N]}(]-\infty,q)\} \subseteq \RR^{k} \times \terms_{\typeone} \times \terms_{\typeone}$. We can reformulate Equation~\ref{eq:projection_from_borel} as saying that $\overline{S_{C,q}}$ is the image of $A_{C,q}$ by the projection function from $\RR^{k} \times \terms_{\typeone} \times \terms_{\typeone}$ to $ \terms_{\typeone} \times \terms_{\typeone}$. In order to end the proof, it is enough from there to show that $A_{C,q}$ is a Borel set: indeed it is known that the projections of Borel sets are analytical (see~\cite{kechris}, Exercise 14.3), and as soon as we know that $\overline{S_{C,q}}$ is analytical, we can deduce that its complementary $S_{C,q}$ is co-analytical.
We now show that  $A_{C,q}$ is a Borel set: to do that we write it as $A_{C,q} = (h_C,q)^{-1}{(a,b) \in \RR^{2} \mid a \neq b}$, where $h_C,q : \RR^{k} \times \terms_{\typeone} \times \terms_{\typeone} \rightarrow \RR^2$ is the sequential composition $h_{C,q}: = f_C ; (g \times g)$, where  we define the functions $f_{C}$, $g_q$ as follows:
\begin{align*}
& f_{C}: (\vec r,M,N) \in  \RR^{k} \times \terms_{\typeone} \times \terms_{\typeone} \mapsto (C[\vec r][M],C[\vec r][N]) \in \terms_{\typereal} \times \terms_{\typereal}\\
& g_q:P \in \terms_{\typereal}  \mapsto \sem{P} (]-\infty,q]) \in \RR.
\end{align*}
We defined the $\sigma$-algebra on $\terms$ in such a way that the syntaxical operations used in the definition of the terms grammar are all measurable; as a consequence the function $f_C$, that can be decomposed using such elementary operations, is measurable. Moreover, since we have shown that the operational semantics $\sem{\,}$ is a Markov kernel, we obtain as an immediate consequence that $g_q$ is measureable. As a consequence, $h_{C,q}$ is also measurable, as composition and cartesian products of measurable functions. From there, and since $\{(a,b) \in \RR^{2} \mid a \neq b\}$ is a Borel set, we obtain that $A_{C,q}$ is a Borel set, so we can conclude.
\end{proof}
}
We can now ask what is the position of
the other notions of equivalence considered 
in this paper in the projective hierarchy: are their graphs measurable? Co-analytical? Are 
they in some levels of the projective hierarchy, at least? 
%An interest of the projective hierarchy is that when we accept the axiom of \emph{projective determinacy} (see~\cite{kechris}), all sets in the projective hierarchy becomes \emph{universally measurable}, meaning that even though they are not necessarily measurable, we are able to attribute a weight to it with respect to any Borel probability measure.
For event bisimilarity, we can show using its logical characterization that the situation is the same as for context equivalence: there exists a choice of primitive functions making (the graph of)
$\sim_{\event}$ not measurable. Nonetheless, $\sim_{\event}$ is co-analytic. 
For state bisimilarity, instead, things are more complicated: 
as for event bisimilarity and contextual equivalence, we can show that there exists a choice of primitive functions making $\sim_{\statet}$ not measurable; but in contrast 
to co-analiticity of $\sim_{\event}$, we are 
unable to show that $\sim_{\statet}$ stays in some level of the projective hierarchy.

\longv{
	\begin{remark}
		We could ask ourselves what is the mesurability status of the 
		equivalence class for $\sim_{\logic}$, when we consider the logical 
		equivalence on the LMPs we use for $\Lambda_P$ programs. On the one 
		hand, we can see that it is \emph{at least as complex} as contextual 
		equivalence, meaning that in the same way that we were able to build a 
		mesurable primitive function that make the graph of $\equiv^{ctx}$ 
		non-mesurable, the same primitive function makes the graph of 
		$\sim_{\logic}$ non mesurable. It is because,
		whenever two values $V,W$ are of type $\typereal \rightarrow 
		\typereal$, we can show that $V \equiv^{ctx} W$ if and only if $V \sim_{\logic} W$.
		Indeed, we obtain immediately the $\Leftarrow$ implication from 
		our completeness result Proposition~\ref{prop:logical_caracterisation}. In the other direction, we 
		see that for the type $\typereal \rightarrow \typereal$, the formulas 
		$\{\act{r}q{\act{=r'} 1{\top}} \mid r, r' \in \RR, q \in \QQ\}$ are 
		enough to completely caracterise context equivalence on values. As a 
		consequence, and since we have  shown that the graph of context 
		equivalence at this type $\typereal \rightarrow \typereal$ is not 
		Borel, the graph of $\sim_{\logic}$ is not a Borel set.
	\end{remark}
	\begin{lemma}
		Let  $(\markovone, \actsone, \{h_a \mid a \in \actsone\})$ be a 
		measurably labelled LMP, with moreover $\markovone$ and $\actsone$ 
		polish spaces equiped with their Borel $\sigma$-algebras.
		%whenever $A \subseteq \markovone \times \actsone \times \markovone$ is 
		%a Borel set, the set...
		Then the graph of $\sim_{\logic}$ is a co-analytical subset of 
		$\programs \times \programs$.
	\end{lemma}
	\begin{proof}
		First, observe that since $\actsone$ is a polish space, we can write 
		the set of all formulas as a countable coproduct of polish space. Since 
		polish spaces are stable by countable coproduct, we can equipp the set 
		of all formulas $\Logic$ with the structure of a polish space. More 
		precisely, we can define pre-logical formula, that are to logical 
		formula the same as terms are to pre-terms, i.e. logical formulas with 
		holes for actions. More precisely, pre-logical formulas are generated 
		by the grammar $\psi \bnf \top \mid \psi_1 \wedge \psi_2 \mid 
		\act{[]^n} q \psi$ with $n \in \NN$, and we ask for aditional 
		constraints that the holes  $[]^1, \ldots []^n$ appear in this order.
		
		In order to do the proof, we first look at the set $\{(\phi, s) \mid s 
		\in \sem \phi\} \subseteq \Logic \times \markovone$.
		First, we can show that for any pre-formula $\psi$, the set 
		$A_{\psi}:=\{ (\vec a, s) \mid s \in \sem {\psi[\vec a]}\}$ is a Borel 
		set: the proof is by induction of the size of $\psi$.
		\begin{itemize}
			\item if $\psi = \top$, $A_\psi = \{\star\} \times \markovone$, 
			where $\star$ represents the unique element of $\actsone^0$, and 
			the result is immediate;
			\item if $\psi = \psi_1 \wedge \psi_2$: $A_\psi = \{((\vec a_1,\vec 
			a_2),s)\mid (\vec a_1,s) \in A_{\psi_1} \wedge (\vec a_2,s) \in 
			A_{\psi_2} \}$, and from there we can conclude using the fact that 
			Borel sets are stable by finite intersections;
			\item  if $\psi ={\act {[]^1} q {\psi'}}$: recall that $\sem 
			{\psi[a_1, \vec {a_2}]} =  \{s \mid h_{a_1}(s,\sem{\psi'[\vec 
			a_2]}) > q \}$. Thus:
			$$A_{\psi} = \{(a_1,\vec a_2,s) \mid h_{a_1}(s,\sem{\psi'[\vec 
			a_2]}) > q  \}.$$
			By induction hypothesis, we know that $\{\vec a_2,t \mid t \in \sem 
			{\psi'[\vec a_2]}\}$ is a Borel set, i.e. a measurable set. From 
			there, and using~\cite{kechris} 17.25, we can deduce that the map 
			$u:(\vec a_2, \mu) \in \actsone^n \times \distrs{\markovone} 
			\mapsto \mu(\psi'[\vec {a_2}]) \in \RR$ is mesurable.
			We can now use the map $u$ by pre-composing it with the map 
			$v:(s,a) \in \markovone \times \actsone \mapsto h_a(s) \in 
			\distrs{\markovone}$: since the LMP we consider is measurable 
			labelled, this map is measurable.
			Summing up, we obtain:
			$$A_\psi = ((v \times id_{\actsone^n}); u)^{-1}(]-\infty,q[), $$
			and since we have established measurability of $v$ and $u$, we 
			obtain that $A_\psi$ is Borel.
			
			We now want to use this result to show that the graph of 
			$\sim_{\logic}$ is co-analytical. To do that, we observe that 
			$\{(s,t) \mid s \sim_{\logic} t \} = \cap_{\psi \text{ a 
			pre-formula}}\{(s,t) \mid \forall \vec a, (\vec a, (s,t)) \in 
			B_\psi\}$, where $B_\psi = \{\phi,s,t \mid (s,\phi) \in A_\psi, 
			(t,\phi) \in A_\psi)\} \cup  \{\phi,s,t \mid (s,\phi) \not \in 
			A_\psi, (t,\phi) \not \in A_\psi)\}$. By stability of measurable 
			substes by finite unions and complementation, we obtain that 
			$B_\psi$ is a measurable subset of $\actsone^n \times 
			\markovone^2$. From there, we can use (as in our proof of the same 
			statement for context equivalence), that a universal quantification 
			over  a Borel subset leads to a co-analytical space, and we can 
			conclude.
			%From there, we can conclude by looking at the definition of a 
			%Markov kernel, that the function $\widehat \phi: s \in \markovone 
			%\mapsto h_a(s,\sem{\phi'})$ is measurable: as a consequence, $\sem 
			%\phi = \widehat^{-1}(]-\infty,q[)$ is a Borel set.
		\end{itemize}
	\end{proof}
}

%The situation as it stands can however be read as a negative result on the utility of state bisimilarity to caracterise context equivalence on $\lambdap$: indeed we can see that either $\sim_{\statet}$ is not local, thus very difficult to use in practice, or it is local but in this case Example~\ref{example:state_sim_fragment} gives us a counter-example to full abstraction.  

%\ugo{Some more comments on the fact that the definition of state bisimilarity 
%is problematic, non-local, and that maybe using couplings we can at least give 
%a more manageable proof methodology, which however turns out not to be 
%complete.}

%\begin{proposition}
%	${\sim_{\statet}^{\lmponelambda}} \subsetneq { \equiv^{\mathit{ctx}}}$.
%\end{proposition}

\section{Related Work}
Programming languages with continuous probabilistic choices and higher-order 
functions have been the subject of numerous studies 
in the last ten years 
(e.g. \cite{borgstrom2016lambda,DBLP:conf/lics/StatonYWHK16,DBLP:journals/pacmpl/EhrhardPT18, 
DBLP:journals/pacmpl/VakarKS19}). This, in particular, thanks to the 
potential applications to 
Bayesian programming, which however require the presence of conditioning 
operators, which we have not considered in this work. 
Program equivalence for higher-order 
languages with continuous probabilities has been studied mostly by means 
of denotational models, the latter being based on nontrivial mathematical 
structures --- such as quasi-Borel spaces and domains~\cite{DBLP:conf/lics/StatonYWHK16,DBLP:conf/esop/Staton17,DBLP:journals/pacmpl/VakarKS19},
and measurable cones~\cite{DBLP:journals/pacmpl/EhrhardPT18} (in fact,
measurable spaces do not form a closed category~\cite{aumann}, and therefore they cannot 
provide an adequate semantics for higher-order languages). 
In this line of research, the recent work by \citet{DBLP:journals/pacmpl/SabokSSW21} 
connecting probabilistic behaviours 
and name generation greatly inspired us in finding 
the examples from Section~\ref{sect:perspective}.
Although all the aforementioned (denotational) semantics have been proved to be adequate 
for languages similar to (and sometimes even more expressive than~\cite{DBLP:journals/pacmpl/VakarKS19}) $\lambda_{\probb}$, to the best of the authors' knowledge 
full abstraction for such semantics is currently an open problem.

In contrast to the many results on denotationally-based program equivalence, 
only few works have focused on operationally-based notions of program equivalence. 
Among those, we mention the work by \citet{DBLP:journals/pacmpl/WandCGC18} 
(see also the work by \citet{DBLP:conf/esop/CulpepperC17}), where full abstraction 
of suitable 
bi-orthogonal logical relations 
on a stochastic $\lambda$-calculus with scoring is proved, and the work by 
\citet{LagoG19}, where soundness of a notion of applicative bisimilarity 
essentially coinciding with our notion of state bisimilarity on $\lmponelambda$ is
proved. Even if fully abstract, bi-orthogonal logical relations have a high logical 
complexity: to logically relate two programs one needs to execute 
such programs inside arbitrary evaluation contexts. 
This makes bi-orthogonal logical relations oftentimes not readily usable to prove equivalence 
between programs. In contrast, the applicative bisimilarity by \citet{LagoG19} is a
lightweight technique, although it is currently unknown whether it is fully abstract.

Applicative bisimilarity~\cite{Abramsky1990} is, together with open 
bisimilarity~\cite{Sangiorgi/TheLazyLambdaCalculusInAConcurrencyScenarion/1994,Lassen/BismulationUntypedLambdaCalculusBohmTrees,Lassen/EagerNormalFormBisimulation/2005}, among the most 
effective and best studied notions of coinductive equivalence for programming 
languages with higher-order functions. Applicative bisimilarity has been studied for 
different languages, including nondeterministic~\cite{Lassen1998,Ong/LICS/1993} and 
probabilistic~\cite{DalLagoSangiorgiAlberti/POPL/2014,CDL14,LagoG19} languages, 
and, more generally, 
languages with algebraic 
effects~\cite{DLGL17,DBLP:journals/tcs/LagoGT20,Gavazzo/ICTCS/2017,DBLP:phd/basesearch/Gavazzo19}. 
For all those languages, applicative bisimilarity has been proved to be a 
sound technique for contextual equivalence. Full-abstraction results, instead, 
are more delicate. Even if applicative bisimilarity is known to be fully abstract 
in presence of certain kinds of effects, there are effects 
for which full abstraction fails, the prime example
being nondeterminism~\cite{Lassen1998}. Applicative bisimilarity is also known  
to be unsound for specific (non-algebraic) effects, such as local 
names~\cite{KLS11}. The case of discrete probabilistic effects is among those 
for which full-abstraction results are possible, 
at least when programs are evaluated following the call-by-value 
discipline~\cite{CDL14}. This is somehow surprising, since in the case 
of $\lambda$-calculi with nondeterministic choice, this correspondence no 
longer holds~\cite{Lassen1998}. Probabilistic effects, in other words, appear 
to somehow 
\emph{lie in between} determinism and nondeterminism. This phenomenon can be 
read through the lenses provided by characterizations of notions of 
bisimulation through \emph{testing}~\cite{van2005domain,LarsenS91}. 
When one moves from discrete to continuous distributions the situation changes. 
As already mentioned, applicative bisimilarity is sound in presence of continuous 
probabilities \cite{LagoG19}, but full abstraction is currently an open problem.

There is a large body of work on generalization of 
the classic notion of probabilistic bisimilarity --- and its 
logical and testing characterizations --- as defined 
by \citet{LarsenS91} to the setting of continuous distributions and 
LMPs~\cite{Prakash2009}. In particular, starting with the seminal 
work by \citet{DBLP:conf/lics/BluteDEP97,DBLP:journals/iandc/DesharnaisEP02,DBLP:conf/lics/DesharnaisEP98}, the notions of state and event bisimilarity \cite{danos2005almost,danos2006bisimulation} (as well as their metric counterparts \cite{DBLP:conf/concur/DesharnaisGJP99,DBLP:journals/tcs/DesharnaisGJP04}) 
have been proposed as notions of equivalence for LMPs. 
In that respect, there is a 
vast literature that investigates the characterization of the aforementioned 
notions of bisimilarity through 
logic-based and testing-based 
notions~\cite{van2005domain,danos2005almost,danos2006bisimulation,fijalkow2017expressiveness,clerc2019expressiveness}.
Even if dealing with LMPs (and thus with uncountable state spaces), 
all the aforementioned works consider LMPs with \emph{countable} actions, 
and thus cannot be applied to the study of expressive higher-order languages, 
such as the ones studied in this paper. 
As an exception to that, we mention the work on bisimulation-based 
equivalences (and their characterizations) on LMPs with \emph{continuous time}
~\cite{DBLP:journals/jlp/DesharnaisP03,fijalkow2017expressiveness}. There, 
the set of actions of the LMP is replaced with a suitable notion of time, which 
usually have the 
cardinality of the continuum, and thus gives something closer to the LMPs we use 
in this paper. It should remarked, however, that even if one can in 
principle see time as a continuous set of actions, the role played by these two 
notions is different, as different are the goals and challenges one 
has when working with LMPs with continuous time and with continuous actions.

% but nonetheless the challenges/goal there are very different) 

%  Only a few of these 
% last 
% results, however, concern LMPs in which the actions have in turn the 
% cardinality of the continuum~\cite{fijalkow2017expressiveness}, and as we mentioned in 
% Section~\ref{sect:feller}, all this is not applicable here.

% Programming languages with continuous probabilistic choices and higher-order 
% functions have been the subject of numerous studies 
% in the last ten years (see, e.g., 
% \cite{borgstrom2016lambda,DBLP:conf/lics/StatonYWHK16,DBLP:journals/pacmpl/EhrhardPT18, 
% DBLP:journals/pacmpl/VakarKS19}). This, in particular, thanks to the 
% potential applications to 
% Bayesian programming, which however require the presence of conditioning 
% operators, which we will not consider in this work. Among the various works 
% that have appeared in the literature, the recent one due to Sabok et 
% al.~\cite{DBLP:journals/pacmpl/SabokSSW21}---in which the nature of full abstraction is 
% investigated and put in 
% relation with probabilistic choice---greatly inspired us in finding 
% the examples from Section~\ref{sect:perspective}.

\section{Conclusion}
This paper studies the nature of coinductive relational reasoning in the 
context of $\lambda$-calculi with continuous probabilistic choices. 
Surprisingly, 
the well-known correspondence between logical, testing, and bisimilarity 
equivalences is lost, with many equivalences turning out to be unsound for 
contextual equivalence. This, 
despite the fact that the theory of bisimilarity for 
probabilistic systems allows for a very smooth transition from the discrete to 
the continuous case. The deep reason for this  hiatus comes from the fact that 
\emph{uncountably} many actions are around, and this falsifies the 
characterization of bisimilarity by testing, namely the key step towards 
full-abstraction. 

More on the positive side, this paper introduces a new class of 
probabilistic systems with uncountable actions in which the transition function 
is itself continuous, in the sense of Feller. For such systems, surprisingly, 
it is once again true that testing precisely characterizes bisimilarity, and 
this holds for both state and event bisimilarity. The aforementioned result is 
put at work in the context of a calculus in which real numbers can only be 
manipulated through continuous functions. This way the result of full 
abstraction is restored. 

In this paper, we are \emph{not} concerned with evaluating the expressive power 
of the thus obtained calculus, namely of $\Lambdacont$. It is clear that 
not allowing a flow of information from reals to other types of 
data to happen --- therefore preventing the value of any real number argument 
from altering the control flow --- is a very strong restriction. However, it is 
not certain that this restriction is unacceptable, think for example at 
Bayesian programming. Along the same lines, we may wonder whether the 
$\Lambdacont$ \emph{contexts} are as powerful as the ones from 
$\Lambda_{\probb}$, i.e., whether two $\Lambdacont$-programs that are context 
equivalent in $\Lambdacont$ are still equivalent when we consider \emph{all} 
the contexts in $\Lambda_{\probb}$.  Such a full abstraction property does not 
seem to hold, since the higher-order nature of the considered languages could 
make it possible for a $\Lambda_{\probb}$ context $C$ to inject non-continuity 
inside a program $M$ of $\Lambdacont$. In any case, a thorough investigation 
along these lines is outside the scope of this paper, and left for future work.

\bibliography{biblio}

\begin{thebibliography}{65}
\providecommand{\natexlab}[1]{#1}
\providecommand{\url}[1]{\texttt{#1}}
\expandafter\ifx\csname urlstyle\endcsname\relax
  \providecommand{\doi}[1]{doi: #1}\else
  \providecommand{\doi}{doi: \begingroup \urlstyle{rm}\Url}\fi

\bibitem[Abramsky(1990)]{Abramsky1990}
Samson Abramsky.
\newblock The lazy lambda calculus.
\newblock In D.~Turner, editor, \emph{Research Topics in Functional
  Programming}, pages 65--117. Addison Wesley, 1990.

\bibitem[Aumann(1961)]{aumann}
Robert~J. Aumann.
\newblock {Borel structures for function spaces}.
\newblock \emph{Illinois Journal of Mathematics}, 5\penalty0 (4):\penalty0 614
  -- 630, 1961.

\bibitem[Backhouse and Hoogendijk(1993)]{DBLP:conf/ifip2-1/BackhouseH93}
Roland~Carl Backhouse and Paul~F. Hoogendijk.
\newblock Elements of a relational theory of datatypes.
\newblock In \emph{Formal Program Development - {IFIP} {TC2/WG} 2.1
  State-of-the-Art Report}, pages 7--42, 1993.

\bibitem[Barendregt(1984)]{Barendregt/Book/1984}
H.P. Barendregt.
\newblock \emph{The lambda calculus: its syntax and semantics}.
\newblock Studies in logic and the foundations of mathematics. North-Holland,
  1984.

\bibitem[Barr(1970)]{Barr/LMM/1970}
M.~Barr.
\newblock Relational algebras.
\newblock \emph{Lect. Notes Math.}, 137:\penalty0 39--55, 1970.

\bibitem[Beiglb{\"o}ck et~al.(2009)Beiglb{\"o}ck, L{\'e}onard, and
  Schachermayer]{beiglbock2009general}
Mathias Beiglb{\"o}ck, Christian L{\'e}onard, and Walter Schachermayer.
\newblock A general duality theorem for the monge--kantorovich transport
  problem.
\newblock \emph{arXiv preprint arXiv:0911.4347}, 2009.

\bibitem[Billingsley(1986)]{Billingsley}
Patrick Billingsley.
\newblock \emph{Probability and Measure}.
\newblock John Wiley and Sons, second edition, 1986.

\bibitem[Blount and Kouritzin(2010)]{blount2010convergence}
Douglas Blount and Michael~A Kouritzin.
\newblock On convergence determining and separating classes of functions.
\newblock \emph{Stochastic processes and their applications}, 120\penalty0
  (10):\penalty0 1898--1907, 2010.

\bibitem[Blute et~al.(1997)Blute, Desharnais, Edalat, and
  Panangaden]{DBLP:conf/lics/BluteDEP97}
Richard Blute, Jos{\'{e}}e Desharnais, Abbas Edalat, and Prakash Panangaden.
\newblock Bisimulation for labelled markov processes.
\newblock In \emph{Proceedings, 12th Annual {IEEE} Symposium on Logic in
  Computer Science, Warsaw, Poland, June 29 - July 2, 1997}, pages 149--158.
  {IEEE} Computer Society, 1997.
\newblock \doi{10.1109/LICS.1997.614943}.

\bibitem[Bogachev and Ruas(2007)]{bogachev2007measure1}
Vladimir~Igorevich Bogachev and Maria Aparecida~Soares Ruas.
\newblock \emph{Measure theory}, volume~1.
\newblock Springer, 2007.

\bibitem[Borgstr{\"o}m et~al.(2016)Borgstr{\"o}m, Dal~Lago, Gordon, and
  Szymczak]{borgstrom2016lambda}
Johannes Borgstr{\"o}m, Ugo Dal~Lago, Andrew~D Gordon, and Marcin Szymczak.
\newblock A lambda-calculus foundation for universal probabilistic programming.
\newblock \emph{ACM SIGPLAN Notices}, 51\penalty0 (9):\penalty0 33--46, 2016.

\bibitem[Clerc et~al.(2019)Clerc, Fijalkow, Klin, and
  Panangaden]{clerc2019expressiveness}
Florence Clerc, Nathana{\"e}l Fijalkow, Bartek Klin, and Prakash Panangaden.
\newblock Expressiveness of probabilistic modal logics: A gradual approach.
\newblock \emph{Information and Computation}, 267:\penalty0 145--163, 2019.

\bibitem[Crubill{\'{e}} and Dal~Lago(2014)]{CDL14}
Rapha{\"{e}}lle Crubill{\'{e}} and Ugo Dal~Lago.
\newblock On probabilistic applicative bisimulation and call-by-value
  {\(\lambda\)}-calculi.
\newblock In \emph{Proc. of {ESOP} 2014}, volume 8410 of \emph{LNCS}, pages
  209--228. Springer, 2014.

\bibitem[Crubill{\'{e}} et~al.(2015)Crubill{\'{e}}, Lago, Sangiorgi, and
  Vignudelli]{DBLP:conf/birthday/CrubilleLSV15}
Rapha{\"{e}}lle Crubill{\'{e}}, Ugo~Dal Lago, Davide Sangiorgi, and Valeria
  Vignudelli.
\newblock On applicative similarity, sequentiality, and full abstraction.
\newblock In \emph{Proc. of Symposium in Honor of Ernst-R{\"{u}}diger
  Olderog.}, volume 9360 of \emph{LNCS}, pages 65--82. Springer, 2015.

\bibitem[Culpepper and Cobb(2017)]{DBLP:conf/esop/CulpepperC17}
Ryan Culpepper and Andrew Cobb.
\newblock Contextual equivalence for probabilistic programs with continuous
  random variables and scoring.
\newblock In Hongseok Yang, editor, \emph{Proc. of {ESOP} 2017}, volume 10201
  of \emph{Lecture Notes in Computer Science}, pages 368--392. Springer, 2017.
\newblock \doi{10.1007/978-3-662-54434-1\_14}.

\bibitem[Dal~Lago et~al.(2014)Dal~Lago, Sangiorgi, and
  Alberti]{DalLagoSangiorgiAlberti/POPL/2014}
U.~Dal~Lago, D.~Sangiorgi, and M.~Alberti.
\newblock On coinductive equivalences for higher-order probabilistic functional
  programs.
\newblock In \emph{Proc. of {POPL} 2014}, pages 297--308, 2014.

\bibitem[Dal~Lago and Gavazzo(2019)]{LagoG19}
Ugo Dal~Lago and Francesco Gavazzo.
\newblock On bisimilarity in lambda calculi with continuous probabilistic
  choice.
\newblock In \emph{Proc. of {MFPS} 2019}, pages 121--141, 2019.

\bibitem[Dal~Lago et~al.(2017{\natexlab{a}})Dal~Lago, Gavazzo, and
  Levy]{DLGL17}
Ugo Dal~Lago, Francesco Gavazzo, and Paul~Blain Levy.
\newblock Effectful applicative bisimilarity: Monads, relators, and howe's
  method.
\newblock In \emph{Proc. of {LICS} 2017}, pages 1--12. {IEEE} Computer Society,
  2017{\natexlab{a}}.

\bibitem[Dal~Lago et~al.(2017{\natexlab{b}})Dal~Lago, Gavazzo, and
  Tanaka]{Gavazzo/ICTCS/2017}
Ugo Dal~Lago, Francesco Gavazzo, and Ryo Tanaka.
\newblock Effectful applicative similarity for call-by-name lambda calculi.
\newblock In \emph{Proc. of {ICTCS} 2017}, pages 87--98, 2017{\natexlab{b}}.

\bibitem[Dal~Lago et~al.(2020)Dal~Lago, Gavazzo, and
  Tanaka]{DBLP:journals/tcs/LagoGT20}
Ugo Dal~Lago, Francesco Gavazzo, and Ryo Tanaka.
\newblock Effectful applicative similarity for call-by-name lambda calculi.
\newblock \emph{Theor. Comput. Sci.}, 813:\penalty0 234--247, 2020.
\newblock \doi{10.1016/j.tcs.2019.12.025}.

\bibitem[Danos et~al.(2005)Danos, Laviolette, Desharnais, and
  Panangaden]{danos2005almost}
Vincent Danos, Fran{\c{c}}ois Laviolette, Jos{\'e}e Desharnais, and Prakash
  Panangaden.
\newblock Almost sure bisimulation in labelled markov processes.
\newblock 2005.

\bibitem[Danos et~al.(2006)Danos, Desharnais, Laviolette, and
  Panangaden]{danos2006bisimulation}
Vincent Danos, Jos{\'e}e Desharnais, Fran{\c{c}}ois Laviolette, and Prakash
  Panangaden.
\newblock Bisimulation and cocongruence for probabilistic systems.
\newblock \emph{Information and Computation}, 204\penalty0 (4):\penalty0
  503--523, 2006.

\bibitem[Desharnais and Panangaden(2003)]{DBLP:journals/jlp/DesharnaisP03}
Jos{\'{e}}e Desharnais and Prakash Panangaden.
\newblock Continuous stochastic logic characterizes bisimulation of
  continuous-time markov processes.
\newblock \emph{J. Log. Algebraic Methods Program.}, 56\penalty0
  (1-2):\penalty0 99--115, 2003.
\newblock \doi{10.1016/S1567-8326(02)00068-1}.

\bibitem[Desharnais et~al.(1998)Desharnais, Edalat, and
  Panangaden]{DBLP:conf/lics/DesharnaisEP98}
Jos{\'{e}}e Desharnais, Abbas Edalat, and Prakash Panangaden.
\newblock A logical characterization of bisimulation for labeled markov
  processes.
\newblock In \emph{Proc. of {LICS} 1998}, pages 478--487. {IEEE} Computer
  Society, 1998.
\newblock \doi{10.1109/LICS.1998.705681}.

\bibitem[Desharnais et~al.(1999)Desharnais, Gupta, Jagadeesan, and
  Panangaden]{DBLP:conf/concur/DesharnaisGJP99}
Jos{\'{e}}e Desharnais, Vineet Gupta, Radha Jagadeesan, and Prakash Panangaden.
\newblock Metrics for labeled markov systems.
\newblock In Jos C.~M. Baeten and Sjouke Mauw, editors, \emph{Proc. of {CONCUR}
  '99}, volume 1664 of \emph{Lecture Notes in Computer Science}, pages
  258--273. Springer, 1999.
\newblock \doi{10.1007/3-540-48320-9\_19}.

\bibitem[Desharnais et~al.(2002)Desharnais, Edalat, and
  Panangaden]{DBLP:journals/iandc/DesharnaisEP02}
Jos{\'{e}}e Desharnais, Abbas Edalat, and Prakash Panangaden.
\newblock Bisimulation for labelled markov processes.
\newblock \emph{Inf. Comput.}, 179\penalty0 (2):\penalty0 163--193, 2002.
\newblock \doi{10.1006/inco.2001.2962}.

\bibitem[Desharnais et~al.(2004)Desharnais, Gupta, Jagadeesan, and
  Panangaden]{DBLP:journals/tcs/DesharnaisGJP04}
Jos{\'{e}}e Desharnais, Vineet Gupta, Radha Jagadeesan, and Prakash Panangaden.
\newblock Metrics for labelled markov processes.
\newblock \emph{Theor. Comput. Sci.}, 318\penalty0 (3):\penalty0 323--354,
  2004.
\newblock \doi{10.1016/j.tcs.2003.09.013}.

\bibitem[Ehrhard et~al.(2018)Ehrhard, Pagani, and
  Tasson]{DBLP:journals/pacmpl/EhrhardPT18}
Thomas Ehrhard, Michele Pagani, and Christine Tasson.
\newblock Measurable cones and stable, measurable functions: a model for
  probabilistic higher-order programming.
\newblock \emph{Proc. {ACM} Program. Lang.}, 2\penalty0 ({POPL}):\penalty0
  59:1--59:28, 2018.

\bibitem[Ethier and Kurtz(2009)]{Ethier-Kurz}
Stewart~N Ethier and Thomas~G Kurtz.
\newblock \emph{Markov processes: characterization and convergence}, volume
  282.
\newblock John Wiley \& Sons, 2009.

\bibitem[Fijalkow et~al.(2017)Fijalkow, Klin, and
  Panangaden]{fijalkow2017expressiveness}
Nathana{\"e}l Fijalkow, Bartek Klin, and Prakash Panangaden.
\newblock Expressiveness of probabilistic modal logics, revisited.
\newblock In \emph{44th International Colloquium on Automata, Languages, and
  Programming (ICALP 2017)}. Schloss Dagstuhl-Leibniz-Zentrum fuer Informatik,
  2017.

\bibitem[Gavazzo(2019)]{DBLP:phd/basesearch/Gavazzo19}
Francesco Gavazzo.
\newblock \emph{Coinductive Equivalences and Metrics for Higher-order Languages
  with Algebraic Effects. (Equivalences coinductives et m{\'{e}}triques pour
  les langages d'ordre sup{\'{e}}rieur avec des effets alg{\'{e}}briques)}.
\newblock PhD thesis, University of Bologna, Italy, 2019.

\bibitem[Giry(1982)]{10.1007/BFb0092872}
Mich{\`e}le Giry.
\newblock A categorical approach to probability theory.
\newblock In \emph{Categorical Aspects of Topology and Analysis}, pages 68--85.
  Springer Berlin Heidelberg, 1982.

\bibitem[Hern{\'a}ndez-Lerma and Lasserre(2012)]{hernandez2012markov}
On{\'e}simo Hern{\'a}ndez-Lerma and Jean~B Lasserre.
\newblock \emph{Markov chains and invariant probabilities}, volume 211.
\newblock Birkh{\"a}user, 2012.

\bibitem[Katsumata and Sato(2015)]{katsumata2015codensity}
Shin-ya Katsumata and Tetsuya Sato.
\newblock Codensity liftings of monads.
\newblock In \emph{6th Conference on Algebra and Coalgebra in Computer Science
  (CALCO 2015)}. Schloss Dagstuhl-Leibniz-Zentrum fuer Informatik, 2015.

\bibitem[Katsumata et~al.(2018)Katsumata, Sato, and
  Uustalu]{DBLP:journals/lmcs/KatsumataSU18}
Shin{-}ya Katsumata, Tetsuya Sato, and Tarmo Uustalu.
\newblock Codensity lifting of monads and its dual.
\newblock \emph{Logical Methods in Computer Science}, 14\penalty0 (4), 2018.

\bibitem[Kechris(2012)]{kechris}
A.~Kechris.
\newblock \emph{Classical Descriptive Set Theory}.
\newblock Graduate Texts in Mathematics. Springer New York, 2012.

\bibitem[Kellerer(1984)]{kellerer1984duality}
Hans~G Kellerer.
\newblock Duality theorems for marginal problems.
\newblock \emph{Zeitschrift f{\"u}r Wahrscheinlichkeitstheorie und verwandte
  Gebiete}, 67\penalty0 (4):\penalty0 399--432, 1984.

\bibitem[Kelly(2005)]{Kelly/EnrichedCats}
Gregory~M. Kelly.
\newblock Basic concepts of enriched category theory.
\newblock \emph{Reprints in Theory and Applications of Categories}, \penalty0
  (10):\penalty0 1--136, 2005.

\bibitem[Koutavas et~al.(2011)Koutavas, Levy, and Sumii]{KLS11}
Vasileios Koutavas, Paul~Blain Levy, and Eijiro Sumii.
\newblock From applicative to environmental bisimulation.
\newblock In \emph{Proc. of {MFPS} 2011}, volume 276 of \emph{Electronic Notes
  in Theoretical Computer Science}, pages 215--235, 2011.

\bibitem[Larsen and Skou(1991)]{LarsenS91}
Kim~Guldstrand Larsen and Arne Skou.
\newblock Bisimulation through probabilistic testing.
\newblock \emph{Inf. Comput.}, 94\penalty0 (1):\penalty0 1--28, 1991.

\bibitem[Lassen(1999)]{Lassen/BismulationUntypedLambdaCalculusBohmTrees}
S{\o}ren~B. Lassen.
\newblock Bisimulation in untyped lambda calculus: B{\"{o}}hm trees and
  bisimulation up to context.
\newblock \emph{Electr. Notes Theor. Comput. Sci.}, 20:\penalty0 346--374,
  1999.

\bibitem[Lassen(2005)]{Lassen/EagerNormalFormBisimulation/2005}
S{\o}ren~B. Lassen.
\newblock Eager normal form bisimulation.
\newblock In \emph{Proceedings of {LICS} 2005}, pages 345--354, 2005.

\bibitem[Lassen(1998)]{Lassen1998}
Søren Lassen.
\newblock \emph{Relational {Reasoning} about {Functions} and {Nondeterminism}}.
\newblock {PhD} dissertation, Aarhus University, 1998.
\newblock BRICS Dissertation Series BRICS DS-98-2.

\bibitem[Levy et~al.(2003)Levy, Power, and Thielecke]{Levy/InfComp/2003}
Paul~Blain Levy, John Power, and Hayo Thielecke.
\newblock Modelling environments in call-by-value programming languages.
\newblock \emph{Inf. Comput.}, 185\penalty0 (2):\penalty0 182--210, 2003.
\newblock \doi{10.1016/S0890-5401(03)00088-9}.

\bibitem[Manes and Arbib(1986)]{manesArbib86}
Ernest~G. Manes and Michael~A. Arbib.
\newblock \emph{Algebraic Approaches to Program Ssemantics}.
\newblock Texts and Monographs in Computer Science. Springer, 1986.

\bibitem[Mislove et~al.(2007)Mislove, Pavlovic, and
  Worrell]{DBLP:journals/entcs/MislovePW07}
Michael~W. Mislove, Dusko Pavlovic, and James Worrell.
\newblock Labelled markov processes as generalised stochastic relations.
\newblock \emph{Electron. Notes Theor. Comput. Sci.}, 172:\penalty0 459--478,
  2007.
\newblock \doi{10.1016/j.entcs.2007.02.015}.
\newblock URL \url{https://doi.org/10.1016/j.entcs.2007.02.015}.

\bibitem[Morris(1969)]{Morris/PhDThesis}
J.~Morris.
\newblock \emph{Lambda Calculus Models of Programming Languages}.
\newblock PhD thesis, MIT, 1969.

\bibitem[Ong(1993)]{Ong/LICS/1993}
C.{-}H.~Luke Ong.
\newblock Non-determinism in a functional setting.
\newblock In \emph{Proc. of {LICS} 1993}, pages 275--286. {IEEE} Computer
  Society, 1993.

\bibitem[Ong(1988)]{Ong/PhDThesis/1988}
C.H.L. Ong.
\newblock \emph{The Lazy Lambda Calculus: An Investigation Into the Foundations
  of Functional Programming}.
\newblock University of London. Imperial College of Science and Technology,
  1988.

\bibitem[Panangaden(1999)]{DBLP:journals/entcs/Panangaden99}
Prakash Panangaden.
\newblock The category of markov kernels.
\newblock \emph{Electron. Notes Theor. Comput. Sci.}, 22:\penalty0 171--187,
  1999.
\newblock \doi{10.1016/S1571-0661(05)80602-4}.

\bibitem[Panangaden(2009)]{Prakash2009}
Prakash Panangaden.
\newblock \emph{Labelled Markov Processes}.
\newblock Imperial College Press, 2009.

\bibitem[Parthasarathy(2005)]{parthasarathy2005probability}
Kalyanapuram~Rangachari Parthasarathy.
\newblock \emph{Probability measures on metric spaces}, volume 352.
\newblock American Mathematical Soc., 2005.

\bibitem[Piotrowski(1985)]{piotrowski1985separate}
Zbigniew Piotrowski.
\newblock Separate and joint continuity.
\newblock \emph{Real Analysis Exchange}, 11\penalty0 (2):\penalty0 293--322,
  1985.

\bibitem[Sabok et~al.(2021)Sabok, Staton, Stein, and
  Wolman]{DBLP:journals/pacmpl/SabokSSW21}
Marcin Sabok, Sam Staton, Dario Stein, and Michael Wolman.
\newblock Probabilistic programming semantics for name generation.
\newblock \emph{Proc. {ACM} Program. Lang.}, 5\penalty0 ({POPL}):\penalty0
  1--29, 2021.

\bibitem[Sangiorgi(1994)]{Sangiorgi/TheLazyLambdaCalculusInAConcurrencyScenarion/1994}
Davide Sangiorgi.
\newblock The lazy lambda calculus in a concurrency scenario.
\newblock \emph{Inf. Comput.}, 111\penalty0 (1):\penalty0 120--153, 1994.

\bibitem[Staton(2017)]{DBLP:conf/esop/Staton17}
Sam Staton.
\newblock Commutative semantics for probabilistic programming.
\newblock In Hongseok Yang, editor, \emph{Proc. of {ESOP} 2017}, volume 10201
  of \emph{Lecture Notes in Computer Science}, pages 855--879. Springer, 2017.
\newblock \doi{10.1007/978-3-662-54434-1\_32}.

\bibitem[Staton et~al.(2016)Staton, Yang, Wood, Heunen, and
  Kammar]{DBLP:conf/lics/StatonYWHK16}
Sam Staton, Hongseok Yang, Frank~D. Wood, Chris Heunen, and Ohad Kammar.
\newblock Semantics for probabilistic programming: higher-order functions,
  continuous distributions, and soft constraints.
\newblock In \emph{Proc. of {LICS} '16}, pages 525--534, 2016.

\bibitem[Swart and Winter(2013)]{swart2013markov}
Jan Swart and Anita Winter.
\newblock Markov processes: Theory and examples.
\newblock 2013.
\newblock Course Notes.

\bibitem[Terraf(2011)]{terraf2011unprovability}
Pedro~S{\'a}nchez Terraf.
\newblock Unprovability of the logical characterization of bisimulation.
\newblock \emph{Information and Computation}, 209\penalty0 (7):\penalty0
  1048--1056, 2011.

\bibitem[V{\'{a}}k{\'{a}}r et~al.(2019)V{\'{a}}k{\'{a}}r, Kammar, and
  Staton]{DBLP:journals/pacmpl/VakarKS19}
Matthijs V{\'{a}}k{\'{a}}r, Ohad Kammar, and Sam Staton.
\newblock A domain theory for statistical probabilistic programming.
\newblock \emph{Proc. {ACM} Program. Lang.}, 3\penalty0 ({POPL}):\penalty0
  36:1--36:29, 2019.

\bibitem[Van~Breugel et~al.(2005)Van~Breugel, Mislove, Ouaknin, and
  Worrell]{van2005domain}
Franck Van~Breugel, Michael Mislove, Jo{\"e}l Ouaknin, and James Worrell.
\newblock Domain theory, testing and simulation for labelled markov processes.
\newblock \emph{Theoretical Computer Science}, 333\penalty0 (1-2):\penalty0
  171--197, 2005.

\bibitem[van Breugel et~al.(2005)van Breugel, Mislove, Ouaknine, and
  Worrell]{DBLP:journals/tcs/BreugelMOW05}
Franck van Breugel, Michael~W. Mislove, Jo{\"{e}}l Ouaknine, and James Worrell.
\newblock Domain theory, testing and simulation for labelled markov processes.
\newblock \emph{Theor. Comput. Sci.}, 333\penalty0 (1-2):\penalty0 171--197,
  2005.
\newblock \doi{10.1016/j.tcs.2004.10.021}.
\newblock URL \url{https://doi.org/10.1016/j.tcs.2004.10.021}.

\bibitem[Viglizzo(2005)]{DBLP:conf/calco/Viglizzo05}
Ignacio~D. Viglizzo.
\newblock Final sequences and final coalgebras for measurable spaces.
\newblock In Jos{\'{e}}~Luiz Fiadeiro, Neil Harman, Markus Roggenbach, and Jan
  J. M.~M. Rutten, editors, \emph{Algebra and Coalgebra in Computer Science:
  First International Conference, {CALCO} 2005, Swansea, UK, September 3-6,
  2005, Proceedings}, volume 3629 of \emph{Lecture Notes in Computer Science},
  pages 395--407. Springer, 2005.
\newblock \doi{10.1007/11548133\_25}.

\bibitem[Villani(2008)]{villani2008optimal}
C{\'e}dric Villani.
\newblock \emph{Optimal transport: old and new}, volume 338.
\newblock Springer Science \& Business Media, 2008.

\bibitem[Wand et~al.(2018)Wand, Culpepper, Giannakopoulos, and
  Cobb]{DBLP:journals/pacmpl/WandCGC18}
Mitchell Wand, Ryan Culpepper, Theophilos Giannakopoulos, and Andrew Cobb.
\newblock Contextual equivalence for a probabilistic language with continuous
  random variables and recursion.
\newblock \emph{Proc. {ACM} Program. Lang.}, 2\penalty0 ({ICFP}):\penalty0
  87:1--87:30, 2018.
\newblock \doi{10.1145/3236782}.

\end{thebibliography}
\bibliographystyle{plainnat}

\end{document}